\documentclass{article}
\usepackage{amsmath,amsfonts,amsthm,amssymb,amscd,color,xcolor,mathrsfs,eufrak, mathtools,mathrsfs,enumitem}
\usepackage{microtype}
\usepackage{hyperref}

 \usepackage{comment}

\colorlet{darkblue}{blue!50!black}
\colorlet{darkmagenta}{magenta!80!black}

\hypersetup{
    colorlinks,%
    citecolor=darkblue,%
    filecolor=red,%
    linkcolor=darkblue,%
    urlcolor=blue,%
    pdfnewwindow=true,%
    pdfstartview={FitH}
}

\newcommand{\p}{\partial}
\newcommand{\e}{\varepsilon}

\newcommand{\R}{{\mathbb R}}
\newcommand{\Z}{{\mathbb Z}}
\newcommand{\IP}{{\mathbb P}}

\newcommand{\I}{{\mathbb I}}
\newcommand{\E}{{\mathbb E}}
\newcommand{\T}{{\mathbb T}}

\newcommand{\N}{{\mathbb N}}

\newcommand{\bfS}{{\boldsymbol{S}}}
\newcommand{\DDD}{{\boldsymbol{D}}}
\newcommand{\QQQ}{{\boldsymbol{Q}}}

\newcommand{\eeta}{{\boldsymbol{{\eta}}}}
\newcommand{\mmu}{{\boldsymbol{{\mu}}}}
\newcommand{\nnu}{{\boldsymbol{{\nu}}}}
\newcommand{\tnnu}{{\boldsymbol{{\tilde\nu}}}}

\newcommand{\III}{{\boldsymbol{I}}}
\newcommand{\JJJ}{{\boldsymbol{J}}}

\newcommand{\TTT}{{\boldsymbol{T}}}
\newcommand{\xxx}{{\boldsymbol{x}}}
\newcommand{\XXX}{{\boldsymbol{X}}}

\newcommand{\yyy}{{\boldsymbol{y}}}

\newcommand{\UUpsilon}{{\boldsymbol\varUpsilon}}
\newcommand{\tUUpsilon}{{\boldsymbol{\widetilde\varUpsilon}}}
\newcommand{\bUUpsilon}{{\boldsymbol{\overline\varUpsilon}}}
\newcommand{\bUpsilon}{{\overline\varUpsilon}}
\newcommand{\tUpsilon}{{\widetilde\varUpsilon}}
\newcommand{\hUUpsilon}{{\boldsymbol{\widehat\varUpsilon}}}

\newcommand{\ssigma}{{\boldsymbol\sigma}}
\newcommand{\llambda}{{\boldsymbol\lambda}}

\newcommand{\ttheta}{{\boldsymbol\theta}}

\newcommand{\aA}{{\cal A}}
\newcommand{\BB}{{\cal B}}
\newcommand{\CC}{{\cal C}}
\newcommand{\DD}{{\cal D}}
\newcommand{\EE}{{\cal E}}
\newcommand{\FF}{{\cal F}}

\newcommand{\HH}{{\cal H}}
\newcommand{\KK}{{\cal K}}
\newcommand{\LL}{{\cal L}}
\newcommand{\MM}{{\cal M}}

\newcommand{\OO}{{\cal O}}
\newcommand{\PP}{{\cal P}}

\newcommand{\RR}{{\cal R}}

\newcommand{\UU}{{\cal U}}
\newcommand{\VV}{{\cal V}}

\newcommand{\XX}{{\cal X}}
\newcommand{\YY}{{\cal Y}}

\newcommand{\lag}{\langle}
\newcommand{\rag}{\rangle}

\newcommand{\dd}{{\textup d}}

\newcommand{\PPPP}{{\mathfrak P}}
\newcommand{\QQQQ}{{\mathfrak Q}}
\newcommand{\MMMM}{{\mathfrak M}}

\newcommand{\XXXX}{{\mathscr X}}
\newcommand{\KKKK}{{\mathscr K}}
\newcommand{\YYYY}{{\mathscr Y}}

\newcommand{\HHHH}{{\mathscr H}}
\newcommand{\EEEE}{{\mathscr E}}
\newcommand{\FFFF}{{\mathscr F}}
\newcommand{\GGGG}{{\mathscr G}}
\newcommand{\VVVV}{{\mathscr V}}
\newcommand{\frX}{{\mathfrak X}}
\newcommand{\bfrX}{\boldsymbol{\mathfrak X}}

\newcommand{\uuu}{{\boldsymbol{\mathit u}}}

\newcommand{\zzz}{{\boldsymbol{\mathit z}}}

\newcommand{\lspan}{\mathop{\rm span}\nolimits}

\newcommand{\supp}{\mathop{\rm supp}\nolimits}
\newcommand{\diver}{\mathop{\rm div}\nolimits}

\newcommand{\Ent}{\mathop{\rm Ent}\nolimits}
\newcommand{\ep}{\mathop{\rm ep}\nolimits}
\newcommand{\Image}{\mathop{\rm Image}\nolimits}

\theoremstyle{plain}
\newtheorem*{mtheorem}{Main Theorem}

\newtheorem{theorem}{Theorem}[section]
\newtheorem{lemma}[theorem]{Lemma}
\newtheorem{proposition}[theorem]{Proposition}
\newtheorem{corollary}[theorem]{Corollary}
\theoremstyle{definition}

\theoremstyle{remark}
\newtheorem{remark}[theorem]{Remark}

\numberwithin{equation}{section}

\makeatletter
\let\@fnsymbol\@arabic
\makeatother

\begin{document}

\title{Large deviations and entropy production \\in viscous fluid flows}
\date{\today}
\author{V.~Jak\v si\'c\footnotemark[1]
\and V.~Nersesyan\footnotemark[2]$^{ \ , 5}$
\and C.-A.~Pillet\footnotemark[3]
\and A.~Shirikyan\footnotemark[4]$^{ \ , 6}$}
\footnotetext[1]{Department of Mathematics and Statistics,
McGill University, 805 Sherbrooke Street West, Montreal, QC, H3A 2K6
Canada; e-mail: \href{mailto:Jaksic@math.mcgill.ca}{Jaksic@math.mcgill.ca}}
\footnotetext[2]{Laboratoire de Math\'ematiques, UMR CNRS 8100, UVSQ, Universit\'e Paris-Saclay, 45, av. des Etats-Unis, F-78035 Versailles, France;  e-mail: \href{mailto:Vahagn.Nersesyan@math.uvsq.fr}{Vahagn.Nersesyan@math.uvsq.fr}}
\footnotetext[3]{Aix Marseille Univ, Universit\'e de Toulon, CNRS, CPT, Marseille, France; e-mail: \href{mailto:Pillet@univ-tln.fr}{Pillet@univ-tln.fr}}
\footnotetext[4]{Department of Mathematics, University of Cergy--Pontoise, CNRS UMR 8088, 2 avenue Adolphe Chauvin, 95302 Cergy--Pontoise, France; 
e-mail: \href{mailto:Armen.Shirikyan@u-cergy.fr}{Armen.Shirikyan@u-cergy.fr}}
\footnotetext[5]{Centre de Recherches Math\'ematiques, CNRS UMI 3457, Universit\'e de Montr\'eal, Montr\'eal,  QC, H3C 3J7, Canada}
\footnotetext[6]{Department of Mathematics and Statistics,
McGill University, 805 Sherbrooke Street West, Montreal, QC, H3A 2K6, Canada}
\maketitle

\begin{abstract} 
We study the motion of a particle in a random time-dependent vector field defined by the 2D Navier--Stokes system with a noise. Under suitable non-degeneracy hypotheses we prove that the empirical measures of the trajectories of the pair (velocity field,  particle)  satisfy the LDP with a good rate function. Moreover, we show that the law of a unique stationary solution restricted to the particle component  possesses a positive smooth density with respect to the Lebesgue measure in any finite time. This allows one to define a natural concept of the entropy production, and to  show that its time average is a bounded function of the trajectory. The proofs are based on a new criterion for the validity of the level-$3$ LDP for Markov processes and an application of a general result on the image of probability measures under smooth maps to the laws associated with the motion of the particle. 


 \end{abstract}

\newpage
\tableofcontents
\setcounter{section}{-1}
\section{Introduction}
\label{s0}
The theory of entropic fluctuations in deterministic and stochastic systems of mathematical physics underwent a spectacular development in the last thirty years. It was initiated in the middle of nineties of the last century in physics literature (see~\cite{ECM-1993,ES-1994,GC-1995,GC-1995a}), and was developed rapidly by various research groups. We refer the reader to the papers~\cite{gallavotti-1995, kurchan-1998,LS-1999,maes-1999,ruelle-1999, ES-2002,gaspard-2005,RM-2007,CG-2008, JPR-2011,CJPS-2017} and the references therein for a detailed account of major achievements in the field. The viewpoints and the frameworks adopted in these papers are not necessarily the same, and we start by  briefly describing the approach to the modern theory of entropic fluctuations that we will adopt here, confining ourselves to the discrete-time setting. For additional information, see the paper~\cite{CJPS-2017} and  the forthcoming review articles~\cite{CJNPS-2018,CJPS-2019}. 

The starting point of the theory of entropic fluctuations is the {\it Large Deviation Principle\/} (LDP) for the empirical measures associated with trajectories.\footnote{All the concepts used in this introduction are defined in the main text.}  Namely, denoting by~$\XXXX$ the phase space of the system in question and by $\{u_k\}_{k\ge0}$ a random trajectory, we introduce the {\it empirical measures\/} by 
\begin{equation} \label{empiricalmeasure}
	\nnu_t=t^{-1}\sum_{k=0}^{t-1}\delta_{\uuu_k}, \quad t\ge1,
\end{equation}
where $\uuu_k=(u_l,l\ge k)$. Thus, $\nnu_t$ is a {\it random probability measure\/} on the product space~$\XXX=\XXXX^{\Z_+}$, where~$\Z_+$ is the set of non-negative integers. If the LDP holds for the sequence $\{\nnu_t\}_{t\ge1}$, then we get an object---the rate function~$\III$---giving a detailed information on its large-time asymptotics. Very often $\III$~does not depend on a choice of trajectory, which makes it an important characteristic of the system. 

Suppose now that the system under study possesses a natural time reversal operation~$\theta$ that can be lifted to an involution~$\ttheta$ in the space of probability measures~$\PP(\XXX)$ (on which~$\III$ is defined). One can ask then how~$\III$ transforms under the action of~$\ttheta$. It was observed in~\cite{BL-2008,BC-2015,CJPS-2017} that, under some additional hypotheses, there is an affine function $\ep:\PP(\XXX)\to\R$ such that 
\begin{equation} \label{level3FR}
\III(\llambda\circ\ttheta)=\III(\llambda)+\ep(\llambda)
\end{equation}
for a large class of measures $\llambda\in\PP(\XXX)$. Identity~\eqref{level3FR} is called {\it level-$3$ fluctuation relation\/}, and the second term on its right-hand side is called the {\it mean entropy production with respect to~$\llambda$\/}. In the Markovian situation, under some regularity hypotheses, the quantity~$\ep(\llambda)$ is the integral of a function $\sigma:\XXX\to\R$ with respect to~$\llambda$ (which will be denoted by~$\langle\sigma,\llambda\rangle$). More generally,  in practically all cases of interest, the mean entropy production  can be written in the form
\begin{equation} \label{meanep}
	\ep(\llambda)=\lim_{t\to\infty}t^{-1}\langle\sigma_t,\llambda\rangle,
\end{equation}
where $\{\sigma_t\}$ is a sequence of measurable functions on~$\XXX$. The functions~$\sigma$ and~$\sigma_t$ (called {\it entropy production functional\/} and {\it entropy production in time~$t$\/}) may be very irregular, and their identification is often a delicate question. Furthermore, the study of the large time behaviour of the quantities $\langle\sigma,\nnu_t\rangle$ or~$t^{-1}\sigma_t$, which are called the {\it time average of the entropy production\/}, is typically a  difficult mathematical problem. Of particular importance are the convergence to a limit and the LDP as $t\to\infty$ because these properties are related to the emergence of the arrow of time and its quantitive description. Namely, if the sequence $\{t^{-1}\sigma_t\}$ has a non-vanishing deterministic limit~$\bar\sigma$ (called {\it mean entropy production rate\/}), then the law of the process~$\{u_k\}$ and its image under the time reversal~$\ttheta$ separate from each other as $t\to\infty$ and eventually become mutually singular.  Moreover, if $\{t^{-1}\sigma_t\}$ satisfies the LDP (or even local LDP on a sufficiently large interval), then one can give a detailed description of the above-mentioned separation of measures in terms of the Hoeffding error exponents (see~\cite{JOPS-2012,CJPS-2017,CJNPS-2018}). If, in addition,  the rate function~$I$ of the full LDP for $\{t^{-1}\sigma_t\}$ is obtained from~$\III$ by the contraction relation
\begin{equation} \label{contractionRF}
	I(r)=\inf\{\III(\llambda):\ep(\llambda)=r\}, 
\end{equation} 
then~$I$ has to satisfy the celebrated {\it Gallavotti--Cohen symmetry relation\/} 
\begin{equation} \label{GC-symmetry}
	I(-r)=I(r)+r\quad\mbox{for $r\in\R$}. 
\end{equation}
Finally,  one can prove that the mean entropy production rate~$\bar\sigma$ is always non-negative, and its strict positivity ensures the non-triviality of the error exponents and the emergence of the arrow of time.  Since mathematical justification of the above program amounts to proving a fine form of the second law of thermodynamics for the system under consideration, it should not come as a surprise that for physically relevant models each step of the program is often a formidable mathematical problem. 

\smallskip
Summarising the above discussion, we can state the following steps in the investigation of entropic fluctuations for a given system: 
\begin{enumerate}[label=\bf{(\alph*)}]
	\item LDP for the empirical measures~\eqref{empiricalmeasure}, also called level-$3$ LDP.
	\item Level-$3$ fluctuation relation~\eqref{level3FR}.
	\item Identification of $\sigma_t$, the functional of entropy production in time~$t$, and its relation with physical transport properties.
	\item Law of large numbers for the sequence of time averages $\{t^{-1}\sigma_t\}$.
	\item Strict positivity of the mean entropy production rate $\bar{\sigma}$. 
	\item Local and global LDP for the sequence of time averages $\{t^{-1}\sigma_t\}$.
\end{enumerate}
We emphasise that each of these steps is essentially a separate problem, and they do not need to be studied in the stated order. 

\medskip
The aim of this paper is to address  questions~(a) and~(c) for a fluid particle moving in a two-dimensional periodic box. Namely, we consider the ordinary differential equation (ODE)
\begin{equation} \label{ode-fluidparticle}
	\dot y=u(t,y), \quad y\in \T^2, 
\end{equation} 
where $u(t,y)$ is a time-dependent vector field defined by the 2D Navier--Stokes system subject to an external random forcing. The law of~$u$ is assumed to be invariant under the time translation $t\mapsto t+1$, while the process itself should have good mixing properties. We do not give more details on the random field~$u$, referring the reader to Section~\ref{ss-formulation} for the exact hypotheses.  The ODE~\eqref{ode-fluidparticle} is supplemented with the initial condition
\begin{equation} \label{ode-IC}
	y(0)=p,
\end{equation}
where $p\in\T^2$ is a given point. The solution of~\eqref{ode-fluidparticle}, \eqref{ode-IC} defines a random dynamical system $\varphi_t:\T^2\to\T^2$, $t\ge0$, and  we are interested in the large-time behaviour of the restriction of~$\varphi_t$ to the integer times. More precisely,   let $\TTT:=(\T^2)^{\Z_+}$ and 
\begin{equation} \label{particleEM}
	\llambda_t^p=t^{-1}\sum_{k=0}^{t-1}\delta_{\yyy_k}, \quad t\ge1,
\end{equation}
where $\delta_\yyy\in\PP(\TTT)$ is the Dirac mass at the point $\yyy\in\TTT$, and $\yyy_k=(\varphi_t(p),t\ge k)$. For any $p\in\T^2$, $\{\llambda_t^p\}$ is a sequence of random probability measures on~$\TTT$. The following theorem is a concise and informal formulation of the main results of this paper. The exact statements and further details can be found in Section~\ref{ss-formulation}. 

\begin{mtheorem}
	Under suitable hypotheses on the vector field~$u(t,y)$, there is a $\T^2$-valued random process $\{z_t,t\ge0\}$ such that its almost every trajectory satisfies~\eqref{ode-fluidparticle} and   the following assertions hold.
	
	\smallskip
	{\sc{Stationarity}}. The laws of the processes $\{z_t\}_{t\ge0}$ and $\{z_{1+t}\}_{t\ge0}$ coincide, and the law of each component coincides with the normalised Lebesgue measure on~$\T^2$. 

	\smallskip	
	{\sc{Convergence}}. 
	For any $s\ge1$ and any initial point $p\in\T^2$, the law of the vector $(\varphi_{t}(p),\dots,\varphi_{t+s}(p))$  converges  exponentially fast  in the total variation norm, as $t\to\infty$, to that of $(z_0,\dots,z_s)$. 
	
	\smallskip
	{\sc{Large deviations}}. 
	For any $p\in\T^2$, the sequence $\{\llambda_t^p\}_{t\ge1}$ satisfies the LDP with some good rate function~$\III:\PP(\TTT)\to[0,+\infty]$. 

	\smallskip
	{\sc{Entropy production}}.
	For any $t\ge1$, the law of $(z_1,\dots,z_t)$ has a strictly positive smooth density $\rho_t(x_1,\dots,x_t)$ with respect to the Lebesgue measure on~$\T^{2t}$. Moreover, there is a number $C>0$ such that the entropy production in time~$t$, defined by 
	\begin{equation} \label{EPon0t}	\sigma_t(\yyy^t)=\log\frac{\rho_t(y_1,\dots,y_t)}{\rho_t(y_t,\dots,y_1)}, \quad \yyy^t:=(y_1,\dots,y_t), 
	\end{equation}
	satisfies the inequality $-C\le t^{-1}\sigma(\yyy^t)\le C$ for all $\yyy^t\in\T^{2t}$. 
\end{mtheorem}

Let us mention that the problem of transport of particles in time-dependent or random vector fields was studied by many authors; see, for example,  the papers~\cite{kraichnan-1970,FP-1994,molchanov-1996,KPS-2013} and the references therein. However, most of these works treat questions that are different from those studied here. To the best of our knowledge, the only exception is the article~\cite{KPS-2013}, which establishes the law of large numbers and central limit theorem for the particle position~$y(t)$ considered in the whole space~$\R^2$ (rather than~$\T^2$). This type of results is  not sufficient to get the convergence of the law of~$y$ to a limiting measure or to study the large deviations for empirical measures. We also mention the recent article~\cite{BBP-2018}, which studies another aspect of chaotic behaviour of fluids---the strict positivity of the top Lyapunov exponent for the dynamics of the Lagrangian particle. The hypotheses imposed in~\cite{BBP-2018} are somewhat different from ours and require the noise to be sufficiently irregular in the space variables. 

The mathematical theory of entropic fluctuations for randomly forced  PDEs is in the beginning of its development. The only two cases for which the  complete program (a)--(f) has been carried out are  the 1D Burgers equation and a nonlinear reaction-diffusion system perturbed by a {\it rough\/} kick noise; see~\cite{JNPS-cmp2015}. However, from the physical  point of view, the roughness hypothesis on the noise  is not always justified, especially in the context of the fluid motion.  Although the  Navier--Stokes system perturbed by a {\it smooth\/} random force satisfies the level-$3$ LDP (see\footnote{For the Navier--Stokes system perturbed by a coloured white noise, the level-$2$ LDP was established in~\cite{nersesyan-2019}.}~\cite{JNPS-cpam2015}), in this case the laws of the forward and backward evolutions are typically singular with respect to each other, and the basic object of the theory of entropic fluctuations---the entropy production in time $t$---is not defined. The present paper bypasses this basic obstruction in a physically and mathematically natural way by focusing on the motion of a particle immersed in the fluid for which we show that all the objects of the theory of entropic fluctuations are well defined. In particular, we establish the level-$3$ LDP and a uniform bound for the mean entropy production in time~$t$. At the same time, the points~(b), (d), (e), and~(f) of the above-mentioned program are yet to be studied. Regarding this last remark, the resolution  of the points~(a) and\footnote{The part of~(c) concerning the relation with the physical notion of transport will be discussed elsewhere.}~(c)  is technically involved and relies on two general results presented in an abstract form in Sections~\ref{s-ldp-control} and~\ref{s-measureimage}. The first of them is the  main novelty of the paper and concerns  a new LDP criterion for randomly forced PDEs. Its proof builds on the results of~\cite{JNPS-cpam2015} and singles out some simple controllability properties that are sufficient for the validity of LDP. This approach makes it possible to treat problems with degenerate noises and is likely to have large scope of applicability, including PDEs studied in~\cite{KNS-2018,shirikyan-jems2019}. In contrast to~(a), the proof of~(c) does not require development of new techniques and is based on a direct  application of a particular case of the  general theory presented in~\cite{bogachev2010}. One may anticipate that a successful  resolution of the remaining points will require developments of new tools that may find applications beyond specific questions dictated by the entropic fluctuations program.

\smallskip
The paper is organised as follows. In Section~\ref{s-MR}, we formulate our main results and describe the scheme of their proof. Section~\ref{s-ldp-control} is devoted to the problem of large deviations. There  we  establish a general criterion for the LDP in terms of certain control properties of the system under study. Section~\ref{s-measureimage} deals with the problem of existence of a density and its positivity for images of probability measures under smooth mappings. In Section~\ref{s-NS}, we study the randomly forced 2D Navier--Stokes system coupled with a Lagrangian  particle. Finally, the Appendix gathers some known results used in the main text. 

\subsubsection*{Acknowledgments}
This research was supported by the \textit{Agence Nationale de la Recherche\/} through the grant NONSTOPS (ANR-17-CE40-0006-01, ANR-17-CE40-0006-02, ANR-17-CE40-0006-03), the CNRS collaboration grant \textit{Fluctuation theorems in sto\-chastic systems\/}, and the \textit{Initiative d'excellence Paris-Seine\/}. VJ acknowledges the support of NSERC. The work of CAP has been carried out in the framework of the Labex Archim\`ede (ANR-11-LABX-0033) and of the A*MIDEX project (ANR-11-IDEX-0001-02), funded by the {\it Investissements d'Avenir\/} French Government programme managed by the French National Research Agency (ANR). AS acknowledges the support of the MME-DII Center of Excellence (ANR-11-LABX-0023-01) and is grateful to  F.~Otto for a discussion on the subject of this paper during the conference \href{https://sites.google.com/site/levico2016/}{\it SPDEs and Applications-X\/} in Trento.

\subsubsection*{Notation}
We write~$\Z^d$ for the integer lattice in~$\R^d$, with the convention $\Z=\Z^1$, and use the notations $\N=\{r\in\Z:r\ge1\}$, $\Z_\pm=\{r\in\Z: \pm r\ge0\}$, $[\![a,b]\!]=[a,b]\cap\Z$, and $\Z^d_*=\Z^d\setminus\{0\}$. We denote by $I\subset\R$ a closed interval, by~$\T^2=\R^2/2\pi\Z^2$ the two-dimensional torus, by~$X$  a Polish space, and by~$\HH$ a separable Banach space. We shall always assume that~$X$ is endowed with the Borel $\sigma$-algebra~$\BB(X)$, and we write~$\MM(X)$ for the space of finite signed measures on~$X$ and~$\PP(X)\subset\MM(X)$ for the simplex of probability measures. We recall the  standard functional spaces of the theory of 2D Navier--Stokes equations, where $s\ge1$ is assumed to be an integer.  

\smallskip
\noindent
$H$ denotes the space of divergence-free vector fields on~$\T^2$ with zero mean value. It is endowed with the usual $L^2$ norm $\|\cdot\|$. 

\smallskip
\noindent
$H^s$ is the usual Sobolev space of~$\R^2$-valued function on~$\T^2$ and $V^s=H^s\cap H$. The corresponding norm will be denoted by~$\|\cdot\|_s$.

\smallskip
\noindent
$L^p(I,\HH)$ stands for the space of Borel-measurable functions $f:I\to\HH$ such that 
$$
\|f\|_{L^p(I,\HH)}=\biggl(\int_I\|f(t)\|_\HH^p\dd t\biggr)^{1/p}<\infty.
$$

\noindent
$C(I,\HH)$ denotes the space of bounded continuous functions $f:I\to \HH$, endowed with the natural norm 
$\|f\|_{C(I,\HH)}=\sup_{t\in I}\|f(t)\|_\HH$. 

\smallskip
\noindent
$\XX_s(I)$ is the space of functions $u\in L^2(I,V^{s+1})$ such that $\p_tu\in L^2(I,V^{s-1})$. 

\medskip
Given a measure~$\mu\in\PP(X)$ and a map~$F(\cdot)$ defined on~$X$, we denote by~$F_*(\mu)$ the image of~$\mu$ under~$F$. If~$F$ depends on an additional parameter~$u$, then we shall write~$F_*(u,\mu)$ to denote the image of~$\mu$ for a fixed value of the parameter. For a function $f:X\to\R$ and a measure~$\mu$ on~$X$, we write $\langle f,\mu\rangle$ for the integral of~$f$ against~$\mu$. We shall also use the following notation for spaces of functions and measures. 

\smallskip
\noindent
$L^\infty(X)$ is the space of bounded measurable functions $f:X\to\R$ with the supremum norm~$\|\cdot\|_\infty$. 

\smallskip
\noindent
$C_b(X)$ is the space of bounded continuous functions $f:X\to\R$ endowed with the norm~$\|\cdot\|_\infty$. For a compact space~$X$, we shall simply write~$C(X)$. 

\smallskip
\noindent
$L_b(X)$ is the space of Lipschitz continuous functions $f\in C_b(X)$ with the norm
$$
\|f\|_L=\|f\|_\infty+\sup_{u\ne v}\frac{|f(u)-f(v)|}{d_X(u,v)}.
$$
$C_b(X,\HH)$ and $L_b(X,\HH)$ are defined in a similar way.

\smallskip
\noindent
$\MM(X)$ is endowed with the weak$^*$ topology which is generated by the functionals $\mu\mapsto \langle f,\mu\rangle$ with $f\in C_b(X)$. The restriction of this topology to~$\PP(X)$ can be metrised by the {\it dual-Lipschitz distance\/} defined as 
$$
\|\mu-\nu\|_L^*=\sup_{\|f\|_L\le 1}\bigl|\langle f,\mu\rangle-\langle f,\nu\rangle\bigr|.
$$
For two measures $\mu,\nu\in\PP(X)$, we denote by $\Ent(\mu\,|\,\nu)$ the relative entropy of~$\mu$ with respect to~$\nu$:
$$
\Ent(\mu\,|\,\nu)=\sup_{V\in C_b(X)}\bigl(\langle V,\mu\rangle-\log\langle e^V,\nu\rangle\bigr)=\int_X\log\frac{\dd\mu}{\dd\nu}\dd\mu,
$$
where the second relation holds if~$\mu$ is absolutely continuous with respect to~$\nu$. 

\section{Main results}
\label{s-MR}

\subsection{Formulations}
\label{ss-formulation}

\subsubsection*{Setting of the problem and preliminaries}
We consider the motion of a particle in a random time-dependent vector field defined by the 2D Navier--Stokes system. More precisely, we study the Cauchy problem~\eqref{ode-fluidparticle}, \eqref{ode-IC}, in which $u=(u_1,u_2)$ is a solution of the system of equations
\begin{equation} \label{NS}
	\p_t u+\langle u,\nabla\rangle u-\nu\Delta u+\nabla \pi=\eta(t,x), \quad \diver u=0,\quad x\in\T^2, 
\end{equation}
supplemented with the initial condition
\begin{equation} \label{NS-IC}
	u(0,x)=u_0(x). 
\end{equation}
Here $\pi=\pi(t,x)$ is the pressure of the fluid, $\nu>0$  the kinematic viscosity, $u_0$ is a square-integrable divergence-free vector field on the torus, and~$\eta$ is a random process of the form
\begin{equation} \label{eta}
	\eta(t,x)=\sum_{k=1}^\infty \eta_k(t-k+1,x)\I_k(t),
\end{equation}
where $\I_k$ is the indicator function of the interval $[k-1,k)$, and~$\{\eta_k\}$ is a sequence of i.i.d.\ random variables in $L^2([0,1]\times\T^2)$. To simplify the formulas, we assume (which can be done without loss of generality) that~$\eta_k$'s are divergence-free. To ensure the boundedness of the energy of solutions for $t\ge0$, we require all the functions to have zero mean value with respect to~$x$. 

Our aim is to study the large-time asymptotics of the pair $(u,y)$. Recall that the scale of spaces~$V^s$ is defined at the end of the Introduction. To ensure the existence of the dynamics for~$y$, we assume that $\eta_k\in L^2(J,V^2)$ almost surely, where $J=[0,1]$. In this case, almost every trajectory of~\eqref{NS} with an initial condition $u_0\in V^3$ belongs to the space $C(\R_+,V^3)$, and it follows that the Cauchy problem~\eqref{ode-fluidparticle}, \eqref{ode-IC} has a unique solution $y\in C(\R_+,\T^2)$ for any initial point $p\in\T^2$. We shall write
\begin{equation} \label{curve}
\varUpsilon(t)=\bigl(u(t),y(t)\bigr), \quad t\ge0,
\end{equation}
for the coupled trajectory and consider it as a continuous curve in $V^3\times \T^2$. Under the hypotheses imposed on~$\eta$, the family of trajectories $\{\varUpsilon(t)\}$ corresponding to all possible initial conditions does not form a Markov process. However, their restriction to integer times does, and our goal is to study the large-time  behaviour of the discrete-time process $\varUpsilon_k=\varUpsilon(k)$, $k\in\Z_+$. 

We now describe the class of random forces~$\eta_k$ we deal with. Denote by~$\{e_j\}_{j\in\Z_*^2}$ the $L^2$ normalised trigonometric basis in the space of divergence-free functions with zero mean value:
\begin{equation} \label{trigonometric-basis}
e_j(x)=E_j^{-1} j^\bot\left\{
\begin{aligned}
	 \cos\langle j,x\rangle & \quad \mbox{for $j_1>0$ or $j_1=0$, $j_2>0$}, \\[2pt]
	 \sin\langle j,x\rangle & \quad\mbox{for $j_1<0$ or $j_1=0$, $j_2<0$},
\end{aligned}
\right.
\end{equation}
where $j^\bot=(-j_2, j_1)$ and $E_j=\sqrt2\pi|j|$ (so that $\|e_j\|=1$ for any $j\in\Z_*^2$). Note that~$\{e_j\}$ is an orthogonal basis in any of the spaces~$V^s$ with respect to the inner product $(u,v)_s=(u,(-\Delta)^sv)$. Furthermore, setting $J=[0,1]$,  we fix an orthonormal basis~$\{\psi_l\}_{l\ge1}$ in the space $L^2(J)$ that satisfies the following {\it  Poincar\'e property\/}: there are positive numbers~$C_r$ and $\theta$ such that
\begin{equation} \label{poincareineq}
\|{\mathsf Q}_Ng\|_{L^2(J)}\le C_rN^{-\theta r}\|g\|_{H^r(J)}\quad \mbox{for $g\in H^r(J)$, $N\ge1$},
\end{equation}
where $r\ge1$ is an arbitrary integer, and~${\mathsf Q}_N$ denotes the orthogonal projection in~$L^2(J)$ onto the closed subspace spanned by~$\psi_l$, $l\ge N$. For instance, the trigonometric basis $\{e^{2\pi i\,lt}\}_{l\in\Z}$ satisfies Poincar\'e property with $\theta=1$. We now formulate our hypothesis on the noise~$\eta_k$. 

\begin{itemize}
	\item [\hypertarget{(N)}{\bf(N)}]
	{\sl The random variables~$\eta_k$ can be written as 
	\begin{equation} \label{eta-k}
		\eta_k(t,x)=\sum_{j\in\Z_*^2}\sum_{l\ge1}b_jc_l\xi_{lj}^k\psi_l(t)e_j(x),
	\end{equation}
	where $\xi_{lj}^k$ are independent scalar random variables. Moreover, the law of~$\xi_{lj}^k$ possesses an infinitely smooth density $\rho_{lj}$ with support in  the interval~$[-1,1]$ such that, for some  $\delta>0$ and  all~$j,l$, $\rho_{lj}(r)>0$ for $|r|<\delta$. Finally, there are positive numbers~$C_m, c$, and $\beta>1/2$ such that
	\begin{align}
0<|b_j|&\le C_m|j|^{-m}\quad
\mbox{for all $m\ge1$},\label{b-j}\\
|c_l|&\ge c\,l^{-\beta}\quad\mbox{for all $l\ge1$}, \quad \sum_{l\ge1}c_l^2<\infty.\label{c-l}
	\end{align}}
\end{itemize}
Note that if this hypothesis is satisfied, then almost every realisation of~$\eta_k$ belongs to $L^2(J,V^s)$ for any $s\ge1$. It follows that, with probability~$1$, the restriction to~$J_k=[k-1,k]$ of the solution~$u$ for the Navier--Stokes system~\eqref{NS} with $C^\infty$ initial condition belongs to $C(J_k,V^s)$ for any $s\ge1$. Therefore, the resolving operator for the Cauchy problem~\eqref{ode-fluidparticle}, \eqref{ode-IC} can be made as smooth as we wish by choosing~$s$ sufficiently large. 

\subsubsection*{Large deviations for empirical measures}
Given an interval $I\subset\R$, we define the spaces 
$$
\XX_s(I)=\bigl\{u\in L^2(I,V^{s+1}):\p_tu\in L^2(I,V^{s-1})\bigr\}, \quad \YY(I)=C(I,\T^2),
$$ 
where $s\ge1$ is an integer, and note that~$\XX_s(I)$ is continuously embedded into $C(I,V^s)$. In the case $I=[0,1]$, we often write $\XX_s$ and~$\YY$, respectively. For any integer $s\ge3$ we denote by 
$$
\bfS:V^s\times\T^2\times L^2([0,1],V^{s})\to \XX_s\times\YY, \quad (u_0,p,\eta)\mapsto(u,y),
$$ 
the resolving operator of the  set of equations~\eqref{NS}, \eqref{ode-fluidparticle}, \eqref{NS-IC}, \eqref{ode-IC}. It is well known that, if $s\ge3$, then~$\bfS$  is $(s-2)$-times\footnote{\label{fn-regularity}The index $s-2$ comes from the fact that $u\in \XX_s$ is a continuous function of time with range in $C^{r}(\T^2)$ for any $r<s-1$, and standard results from the theory of ODEs can ensure only the existence of $s-2$ continuous derivatives for~$\bfS^y$.} continuously differentiable in the Fr\'echet sense. We denote by $S(u_0,p,\eta)$  the value of~$\bfS(u_0,p,\eta)$ at $t=1$. Note that $S$ is a map with range in~$V^s\times \T^2$. We  write $\bfS=(\bfS^u,\bfS^y)$ and $S=(S^u,S^y)$, with a natural definition of the $u$- and~$y$-components. 

Our first result deals with the level-$3$ LDP for trajectories issued from an initial point belonging to  the domain of attainability from~$\{0\}\times\T^2$ (which is also the support of the unique stationary distribution for~\eqref{NS}, \eqref{ode-fluidparticle}; see the next subsection on the regularity of laws). Namely, for a fixed $s\ge3$, let~$\KK^s\subset L^2([0,1],V^{s})$ be the support of the law of~$\eta_k$. We define the sets
\begin{equation} \label{attdomain}
	\aA_0^s=\{0\}, \quad \aA_k^s=S^u(\aA_{k-1}^s,\KK^s), \quad k\ge1,
\end{equation}
and denote by~$\aA^s$ the closure of the union $\cup_{k\ge0}\aA_k^s$ in the space~$V^s$. The following lemma is easy to establish, and we omit its proof. 

\begin{lemma} \label{l-attdomain}
	Let Hypothesis~\hyperlink{(N)}{\rm(N)} be satisfied. Then the following properties hold for any integer $s\ge3$. 
	\begin{description}
		\item [\sc{Compactness}.] The set~$\aA^s$ is compact in~$V^s$ and contains the point~$0$. 
		\item [\sc{Compatibility}.] If $r>s$ is another integer, then~$\aA^s$ is the closure of~$\aA^r$ in~$V^s$. 
		\item [\sc{Invariance}.] The set~$\XXXX^s:=\aA^s\times\T^2$ is invariant, that is, $S(\XXXX^s,\KK^s)\subset\XXXX^s$. 
	\end{description}
\end{lemma}

We now introduce the {\it empirical measures\/} for~\eqref{NS}, \eqref{ode-fluidparticle} by the formula 
\begin{equation} \label{empiricalmeasureMP}
	\nnu_t^\varUpsilon=t^{-1}\sum_{n=0}^{t-1}\delta_{\UUpsilon_n}, \quad t\ge1,
\end{equation}
where $\varUpsilon=(u_0,p)$ is an initial point, $\UUpsilon_n=(\varUpsilon_k,k\ge n)$, and~$\varUpsilon_k$ is the value of the solution of~\eqref{NS}, \eqref{ode-fluidparticle}, \eqref{NS-IC}, \eqref{ode-IC} at $t=k$. Setting $\XXX^s=(\XXXX^s)^{\Z_+}$, it is straightforward to see that if~$\varUpsilon\in\XXXX^s$, then $\UUpsilon_n\in\XXX^s$ for any $n\ge0$. The following theorem uses standard  notions of the theory of large deviations.\footnote{For their definitions we refer the reader to  Section~\ref{ss-result}.}

\begin{theorem} \label{t-ldp}
	Let Hypothesis~\hyperlink{(N)}{\rm(N)} be fulfilled and let $s\ge3$ be an integer. Then the family of empirical measures~$\{\nnu_t^\varUpsilon,\varUpsilon\in\XXXX^s\}_{t\ge1}$ satisfies the uniform LDP with some good rate function $\III^s:\PP(\XXX^s)\to[0,+\infty]$. Moreover, $\III^s$ is an affine function on~$\PP(\XXX^s)$ given by the Donsker--Varadhan entropy formula. 
\end{theorem}

\subsubsection*{Regularity of laws for the particle and convergence}

We now focus on  the law of the particle in more detail. Note that, for any $s\ge3$, the compact invariant set~$\XXXX^s$ carries a stationary measure for the Markov process associated with~\eqref{NS}, \eqref{ode-fluidparticle}. More precisely, if~\eqref{curve} is a trajectory for~\eqref{NS}, \eqref{ode-fluidparticle}, then the vector functions $\varUpsilon_k=\varUpsilon(k)$ satisfy the relations
\begin{equation} \label{RDS}
\varUpsilon_k=S(\varUpsilon_{k-1},\eta_k), \quad k\ge1. 
\end{equation}
Since~$\{\eta_k\}$ is a sequence of i.i.d.\ random variables, Eq.~\eqref{RDS} defines a discrete-time homogeneous Markov process in~$V^s\times\T^2$ whose transition function has the form
\begin{equation} \label{transitionfunction}
	P_1(\varUpsilon,\cdot)=S_*(\varUpsilon,\ell),
\end{equation}
where~$\ell$ stands for the law of~$\eta_k$, and the right-hand side denotes the image of~$\ell$ under the mapping $\zeta\mapsto S(\varUpsilon,\zeta)$. By Lemma~\ref{l-attdomain}, the set~$\XXXX^s$ is invariant in the sense that $P_1(\varUpsilon,\XXXX^s)=1$ for any $\varUpsilon\in\XXXX^s$. In what follows we consider the restriction of the Markov process defined by~\eqref{RDS} to~$\XXXX^s$ and denote by~$\PPPP_k$ and~$\PPPP_k^*$ the corresponding Markov operators acting on the spaces~$C(\XXXX^s)$ and~$\PP(\XXXX^s)$, respectively. Since~$\XXXX^s$ is compact, there is at least one stationary measure~$\MMMM\in\PP(\XXXX^s)$. Applying Theorem~\ref{t-expomixing}, one can prove  that~$\MMMM$ is the unique stationary measure for~\eqref{RDS}. Let us note that the uniqueness of a stationary distribution was proved in~\cite{BBP-2018} for the coupled system~\eqref{NS}, \eqref{ode-fluidparticle} with a coloured white noise~$\eta$; however, their approach is not applicable in our situation since it is based on the strong Feller property and requires the noise to be rough in the space variables.

A simple argument based on the uniqueness of the stationary measure proves that~$\MMMM$ is independent of~$s$. Moreover, another  short computation shows  that\footnote{Note, however, that this product structure is not preserved on the level of path measures.} $\MMMM=\mu\otimes\lambda$, where~$\mu$ is the unique stationary measure for~\eqref{NS} and~$\lambda$ is the normalised Lebesgue measure on~$\T^2$. We shall denote by~$\boldsymbol{\MMMM}\in \PP(\XXX^s)$ the corresponding path measure and by~$\mmu\in\PP(\boldsymbol{\aA})$ and~$\llambda\in\PP(\TTT)$ its projections to the $u$- and $y$-components, where $\boldsymbol{\aA}=\aA^{\Z_+}$ and $\TTT=(\T^2)^{\Z_+}$. Similarly, given an initial point $\varUpsilon\in\XXXX^s$, we shall denote by $\boldsymbol{\MMMM}^\varUpsilon\in \PP(\XXX^s)$ the path measure of the trajectory for~\eqref{NS}, \eqref{ode-fluidparticle} issued from~$\varUpsilon$, by~$\MMMM_t^\varUpsilon\in\PP(\XXXX^s)$ its projection the $t^\text{th}$ component, and by $\mmu^\varUpsilon\in\PP(\boldsymbol{\aA})$ and~$\llambda^\varUpsilon\in\PP(\TTT)$ its projections to the $u$- and~$y$-components, respectively. Finally, given an integer interval $I\subset\Z_+$, we denote by~$\llambda_I^\varUpsilon\in\PP(\T^{2|I|})$ the projection of~$\llambda^\varUpsilon$ to~$I$ and define~$\llambda_I$ similarly. We shall write $\llambda_t^\varUpsilon$ and~$\llambda_t$ for $I=[\![1,t]\!]$.

\begin{theorem} \label{t-particle}
Suppose that Hypothesis~\hyperlink{(N)}{\rm(N)} is satisfied. Then the following holds for any integer $t\ge2$. 
\begin{description}
	\item [\sc{Regularity}.] 
	For any $\varUpsilon\in\XXXX^3$, the measure $\llambda_{[\![2,t]\!]}^\varUpsilon$ has a density $\rho_{[\![2,t]\!]}^\varUpsilon$ that belongs to~$C^\infty(\T^{2(t-1)})$, and the function $\varUpsilon\mapsto\rho_{[\![2,t]\!]}^\varUpsilon$ is Lipschitz continuous from~$\XXXX^3$ to~$C^k(\T^{2(t-1)})$ for any $k\ge1$. Moreover, the measure~$\llambda_t$ has a density $\rho_t\in C^\infty(\T^{2t})$.
	\item [\sc{Convergence}.]
	There is $\gamma>0$ such that, for any integer $k\ge1$, we have
	\begin{equation} \label{convergencedensity}
		\sup_{\varUpsilon\in\XXXX^3}\bigl\|\rho_{[\![n+1,n+t]\!]}^\varUpsilon-\rho_t\bigr\|_{C^k(\T^{2t})}\le C_{tk}e^{-\gamma n},\quad n\ge1,
	\end{equation}
	where the constant $C_{tk}>0$ does not depend on~$n$. 
\end{description}
\end{theorem}

Let us note that if $\varUpsilon\in\XXXX^s$ is not infinitely smooth, there is no reason for~$\rho_t^\varUpsilon$ to be $C^\infty$ even for $t=1$. Indeed, as it was mentioned in footnote~\ref{fn-regularity}, the map $\eta\mapsto S^y(\varUpsilon,\eta)$ acting from $L^2([0,1],V^s)$ to~$\T^2$ possesses only finite regularity, unless $\varUpsilon\in\XXXX^s$ is infinitely smooth. Therefore, without any regularisation mechanism, the image of a measure under the action of~$S^y(\varUpsilon,\cdot)$ does not need to have a smooth density. On the other hand, the following remark about finite regularity will be important in the definition of the entropy production.

\begin{remark} \label{r-densitytime1}
	The proof of Theorem~\ref{t-particle} will imply that, for any integer $k\ge0$, there is $s\ge3$ such that, for any $t\ge1$ and~$\varUpsilon\in\XXXX^s$, the measure $\llambda_t^\varUpsilon$ has a density $\rho_t^\varUpsilon\in C^k(\T^{2t})$. Moreover, the mapping $\varUpsilon\mapsto \rho_t^\varUpsilon$ is Lipschitz continuous from~$\XXXX^s$ to~$C^k(\T^{2t})$. 
\end{remark}

\subsubsection*{Strict positivity of densities}
To ensure strict positivity of the densities~$\rho_t$ and to derive a uniform bound on the mean entropy production in time $t$,  we need to replace the random force on the right-hand side of~\eqref{NS} by $\eta^a:=a\eta$, where $a>0$ is a large parameter. We shall denote by~$\rho_t^a$ the densities corresponding to the resulting equation. 

\begin{theorem} \label{t-entropy}
	Suppose that Hypothesis~\hyperlink{(N)}{\rm(N)} is satisfied. Then there is $a_0>0$ such that the following holds for any $a\ge a_0$. 
	\begin{description}
		\item [\sc{Strict positivity}.]
		The functions~$\rho_t^a$ are bounded below by positive numbers. 
		\item [\sc{Uniform bound on the entropy production}.]
		There is $C>0$ such that the entropy production defined by~\eqref{EPon0t} satisfies the inequality
		\begin{equation} \label{entropybound}
			\bigl|t^{-1}\sigma_t(y_1,\dots,y_t)\bigr|\le C\quad\mbox{for all $(y_1,\dots,y_t)\in\T^{2t}$, $t\ge1$}. 
		\end{equation}
	\end{description}
\end{theorem}
As we shall describe in the next section, the uniform bound on the entropy production is an easy consequence of the strict positivity of~$\rho_1^\varUpsilon(y)$. The proof given in Section~\ref{ss-EP} will imply that, for this theorem to be true, it suffices to have a large parameter in front of finitely many Fourier modes in~$x$. On the other hand, the following simple observation shows that~$\rho_1^\varUpsilon(y)$ cannot be strictly positive for any $\varUpsilon\in\XXXX^s$ and $y\in\T^2$, unless the noise is sufficiently large. Indeed, suppose that $\varUpsilon=(0,p)$ and $|y-p|$ is of order~$1$. In this case, the size of the velocity field on the interval $[0,1]$ can be bounded by the norm of the noise. If the latter is of order~$\e>0$, then the particle can travel a distance no larger than~$C\e$, and so  $\rho_1^\varUpsilon(y)=0$ for $|y-p|>C\e$.

\subsection{Schemes of the proofs}
\label{ss-scheme}

\subsubsection*{Theorem~\ref{t-ldp}}
In Section~\ref{s-ldp-control} we shall derive a sufficient condition for the validity of LDP in the context of the Markovian RDS~\eqref{RDS}. Apart from the regularity of~$S$ and a decomposability hypothesis on the law of the random noise, this criterion requires two properties: approximate controllability of the nonlinear system by controls belonging to the support of the law~$\eta$ and the density of the image of the linearised operator; see~\hyperlink{(AC)}{(AC)} and~\hyperlink{(ACL)}{(ACL)}. The verification of these two properties is based on essentially the same idea, which we briefly outline here, leaving the details for Section~\ref{ss-ldpproof}. Note that some related problems on the control of a particle by the vector field appeared in the papers~\cite{nersisyan-2011,nersesyan-2015, BBP-2018}, and our proof uses some ideas from these articles.

\smallskip
Suppose we wish to prove that a point $\varUpsilon_0=(0,p)\in\XXXX^s$ can be exactly  steered to any point $\widehat\varUpsilon=(0,\hat p)$ that is sufficiently close to~$\varUpsilon_0$. Let us set
\begin{equation} \label{U1U2gamma}
U_1(x)=(\cos x_2,0), \quad U_2(x)=(0,\cos x_1), \quad \gamma(t)=\bigl(1-\alpha(t)\bigr)\,p+\alpha(t)\hat p,	
\end{equation}
where $\alpha\in C^\infty(\R)$ is such that $\alpha(t)=0$ for $t\le1/3$ and $\alpha(t)=1$ for $t\ge 2/3$. Writing 
\begin{equation}  \label{gammadot}
	\dot\gamma(t)
	=\dot\alpha(t)(\hat p-p)
	=\bigl(\varphi_1(t),\varphi_2(t)\bigr),
\end{equation}
we define the functions
\begin{equation} \label{uy-control}
	u(t,x)=\varphi_1(t) U_1\bigl(x-\gamma(t)\bigr)+\varphi_2(t) U_2\bigl(x-\gamma(t)\bigr), \quad y(t)=\gamma(t), 
\end{equation}
where  $t\in[0,1]$. Then the vector function $\varUpsilon=(u,y)$ is infinitely smooth,  coincides with~$(0,p)$ and~$(0,\hat p)$ at the endpoints of the interval $[0,1]$, and satisfies Eqs.~\eqref{NS}, \eqref{ode-fluidparticle} with 
\begin{equation} \label{eta-control}
	\eta(t)=\Pi g(t), \quad g(t)=\p_t u+\langle u,\nabla\rangle u-\nu\Delta u,
\end{equation}
where $\Pi:L^2(\T^2,\R^2)\to H$ stands for Leray's projection. It is straightforward to check that $g$ can be written as 
\begin{equation} \label{gtx}
	(\Pi g)(t,x)=\sum_{j\in\Lambda}\alpha_j(t)e_j(x),
\end{equation}
where $\Lambda=\{j=(j_1,j_2)\in\Z_*^2:|j_1|+|j_2|\le 2\}$, the trigonometric basis~$\{e_j\}$ is defined by~\eqref{trigonometric-basis}, and $\alpha_j$'s are smooth functions of $t\in[0,1]$ whose~$C^r$ norms are  proportional to~$|\hat p-p|$ for any $r\ge1$. We claim that~$\Pi g$ is in the support~$\KK^s$ of~$\DD(\eta_k)$, provided that $|\hat p-p|\ll1$. Indeed, it follows from~\hyperlink{(N)}{(N)} that~$\KK^s$ contains any function of the form 
\begin{equation} \label{functioninKK}
	h(t,x)=\sum_{j\in\Lambda}\sum_{l=1}^\infty h_{lj}\psi_l(t)e_j(x),
\end{equation}
where the coefficients satisfy the inequality $|h_{lj}|\le \e l^{-\beta}$ with $\e\ll1$. Since~$\alpha_j$'s are infinitely smooth, it follows from~\eqref{poincareineq} that the coefficients~$\alpha_{jl}$ of the expansion of~$\alpha_j$ in the basis~$\{\psi_l\}$ decay faster than any negative degree of~$l$. Since they are bounded by a number propositional to $|\hat p-p|$, we conclude that $\Pi g\in\KK^s$, provided that $|\hat p-p|\ll1$. 

\subsubsection*{Theorem~\ref{t-particle}}
In Section~\ref{s-measureimage}, we present a sufficient condition for the existence of a regular density for the image of a probability measure under a smooth mapping; see Theorem~\ref{T:2.7}. Roughly speaking, it says that if a smooth map~$F$ with range in a finite-dimensional manifold is such that its derivative is surjective everywhere, then the image of a probability measure~$\ell$ has smooth density, provided that~$\ell$ is regular in an appropriate sense.  Measures satisfying Hypothesis~\hyperlink{(N)}{(N)} do possess the required regularity property, and the position of the particle can be written as a smooth function~$F$ of the noise and the initial condition of the system. The fact that the derivative of~$F$ is surjective will follow from the density of the image for the linearised operator. This will establish the existence of $\rho_{[\![2,t]\!]}^\varUpsilon$. 

\smallskip
To prove convergence~\eqref{convergencedensity}, we first note that the sequence of measures $\{\MMMM_k^\varUpsilon\}_{k\ge1}$ converges, as $k\to\infty$, to~$\MMMM$ exponentially fast in the dual-Lipschitz norm; this is established in Theorem~\ref{t-expomixing}. Let us fix any $s\ge3$, set $\EEEE=L^2(J,V^s)$, and introduce a map
$$
F^t:\XXXX^3\times\underbrace{\EEEE\times\cdots\times\EEEE}_{\mbox{\footnotesize$t$ times}}\to \T^{2t}
$$
that takes $(\varUpsilon,\eta_1,\dots,\eta_t)$ to $(y_1,\dots,y_t)$, where $y_k$ is the $y$-component of the trajectory~$\varUpsilon_k$ for~\eqref{RDS}. In this case, we can write
\begin{equation} \label{lambdan1nt}
	\llambda_{[\![n+1,n+t]\!]}^\varUpsilon
=\E\,F_{*}^t(\varUpsilon_n, \underbrace{\ell\otimes\cdots\otimes\ell}_{\mbox{\footnotesize $t$ times}}).
\end{equation}
Now note that~$\rho_t^\varUpsilon$ is the density of $F_{*}^t(\varUpsilon, \ell\otimes\cdots\otimes\ell)$  with respect to the Lebesgue measure on~$\T^{2t}$. It follows that 
\begin{equation} \label{rhon1nt}
	\rho_{[\![n+1,n+t]\!]}^\varUpsilon(y_1,\dots,y_t)
	=\int_{\XXXX^s}\rho_t^{\upsilon}(y_1,\dots,y_t)\,\MMMM_n^\varUpsilon(\dd\upsilon).
\end{equation}
Since~$\rho_t^\upsilon(y)$ is Lipschitz continuous in~$\upsilon$, together with all its derivatives in~$y$, this will imply the required convergence~\eqref{convergencedensity}. 

\subsubsection*{Theorem~\ref{t-entropy}}
As it was established in Theorem~\ref{t-particle} and Remark~\ref{r-densitytime1}, if an integer $s\ge3$ is sufficiently large, then for any $\varUpsilon\in\XXXX^s$ the projection of the transition function $P_1(\varUpsilon, \cdot)$ to the $y$-component possesses a density $\rho_1^\varUpsilon(y)$,
\begin{equation} \label{transition-y}
	P_1^y(\varUpsilon,\dd y)=\rho_1^\varUpsilon(y)\,\dd y, 
\end{equation}
and the mapping $(\varUpsilon,y)\mapsto \rho_1^\varUpsilon(y)$ is continuous from~$\XXXX^s$ to~$C(\T^2)$. It follows that~$\rho_1^\varUpsilon(y)$ is continuous in~$(\varUpsilon,y)$ and, by the compactness of~$\XXXX^s\times\T^2$, there is $M>0$ such that
\begin{equation} \label{upperbound-density}
\rho_1^\varUpsilon(y)\le M\quad\mbox{for all $\varUpsilon\in\XXXX^s$, $y\in\T^2$}. 
\end{equation}
By the Kolmogorov--Chapman relation, for an arbitrary non-negative function $f:\T^{2t}\to\R$, we have 
\begin{align*}
	\langle f,\llambda_t\rangle
	&=\int\limits_{\XXXX^s(t+1)}
	f(\yyy^t)\,\MMMM(\dd\varUpsilon)P_1(\varUpsilon,\dd \varUpsilon_1)\cdots P_1(\varUpsilon_{t-1},\dd \varUpsilon_t)\\
	&=\int\limits_{\XXXX^s(t)\times\T^2}
	f(\yyy^t)\rho_1^{\varUpsilon_{t-1}}(y_t)\,\MMMM(\dd\varUpsilon)P_1(\varUpsilon,\dd \varUpsilon_1)\cdots P_1(\varUpsilon_{t-2},\dd \varUpsilon_{t-1})\,\dd y_t\\
	&\le M \int\limits_{\XXXX^s(t)\times\T^2}
	f(\yyy^t)\,\MMMM(\dd\varUpsilon)P_1(\varUpsilon,\dd \varUpsilon_1)\cdots P_1(\varUpsilon_{t-2},\dd \varUpsilon_{t-1})\,\dd y_t,
\end{align*}
where $\yyy^t=(y_1,\dots,y_t)\in\T^{2t}$,  $\XXXX^s(t)$ denotes the $t$-fold product of the space~$\XXXX^s$, and we used~\eqref{transition-y} and~\eqref{upperbound-density}. Iterating this argument and using the relation $\MMMM(\XXXX)=1$, we derive
$$
	\langle f,\llambda_t\rangle
	\le M^t\int_{\T^{2t}}f(y_1,\dots,y_t)\,\dd y_1\dots\dd y_t. 
$$
Since $f\ge0$ was arbitrary, it follows that
\begin{equation} \label{density-upperbound}
	\rho_t(y_1,\dots,y_t)\le M^t\quad\mbox{for any $(y_1,\dots,y_t)\in\T^{2t}$}. 
\end{equation}
Note that the upper bound for the density does not require any additional hypotheses on the noise. 

We now turn to the lower bound. Suppose we have proved that 
\begin{equation} \label{positivity-density-particle}
\rho_1^\varUpsilon(y)>0\quad\mbox{for all $\varUpsilon\in\XXXX^s$, $y\in\T^2$}. 
\end{equation}
Then, by continuity and compactness, we can find $m>0$ such that $\rho_1^\varUpsilon(y)\ge m$ for $\varUpsilon\in\XXXX^s$, $y\in\T^2$. Repeating the above argument, one gets that,  for any non-negative function $f:\T^{2t}\to\R$, 
 $$
	\langle f,\llambda_t\rangle
	\ge m^t\int_{\T^{2t}}f(y_1,\dots,y_t)\,\dd y_1\dots\dd y_t,
$$
and so it follows that 
\begin{equation} \label{density-lowerbound}
	\rho_t(y_1,\dots,y_t)\ge m^t\quad\mbox{for any $(y_1,\dots,y_t)\in\T^{2t}$}. 
\end{equation}
Inequalities~\eqref{density-upperbound} and~\eqref{density-lowerbound} allow to define the entropy production in time~$t$ by relation~\eqref{EPon0t} and to derive the  estimate~\eqref{entropybound} for its  time-average. 

\smallskip
The above elementary argument  reduces  the proof of Theorem~\ref{t-entropy} to the verification  of~\eqref{positivity-density-particle}. Theorem~\ref{T:2.8} gives a sufficient condition for the positivity of the density for the image of a probability measure~$\ell$ under a finite-dimensional  smooth map. Roughly speaking, it says that if a point~$\hat p$ has a pre-image in the ``interior'' of the support of~$\ell$, then the density is strictly positive at~$\hat p$. Hence, the proof  further reduces to a problem of exact controllability for the Navier--Stokes system coupled to the Lagrangian particle. We shall show in Section~\ref{ss-EP} that this can be established by modifying the above scheme used in the proof of Theorem~\ref{t-ldp}, provided that the noise contains a large parameter in front of finitely many Fourier modes in the space variables. 

Let us also mention that the above argument cannot be applied to the full system since the transition functions corresponding to different initial points $\varUpsilon=(u,p)\in\XXXX^s$ are not equivalent. It is the integration with respect to $u\in \XXXX^s$ that removes this singularity and allows one to prove the equivalence of the (projections of)  transition probabilities.  Moreover, we conjecture  that the laws of the forward and backward stationary processes of the full system~\eqref{NS}, \eqref{ode-fluidparticle} are not equivalent. Indeed, for the (linear) Stokes system perturbed by a spatially regular white noise, after integrating out the $p$-variable, one gets a Gaussian process for which there exist necessary and sufficient conditions (in terms of the noise) for the equivalence of forward and backward laws; cf.\ Theorem~7.2.1 in~\cite{DZ1996}. In this case, it is not difficult to construct a noise for which the two laws are singular.

\section{Large deviations via controllability}
\label{s-ldp-control}

\subsection{Formulation of the result}
\label{ss-result}
Let~$\HH$ be a separable Hilbert space, let~$\YYYY$ be a compact Riemannian manifold, let $\HHHH=\HH\times\YYYY$ be the product space with natural projections~$\Pi_\HH$ and~$\Pi_\YYYY$ to its components,  and let~$\EEEE$ be a separable Banach space.  We fix a continuous mapping $S:\HHHH\times \EEEE\to\HHHH$ and consider the random dynamical system~\eqref{RDS} in which $\{\eta_k\}$ is a sequence of i.i.d.\ random variables in~$\EEEE$. We shall denote by $\KK\subset\EEEE$ the support of the law of~$\eta_k$ and assume that  there is a compact subset $\aA\subset\HH$ such that $\XXXX:=\aA\times\YYYY$ is invariant for~\eqref{RDS} ($S(\XXXX\times \KK)\subset\XXXX$). We impose the following three hypotheses on the mapping~$S$. 

\begin{itemize}
	\item[\hypertarget{(R)}{\bf(R)}] 
	\sl There is a Banach space~$\VV$ compactly embedded into~$\HH$ such that the image of~$S$ is contained in $\VVVV:=\VV\times\YYYY$, the mapping $S:\HHHH\times\EEEE\to\VVVV$ is twice continuously differentiable, and its derivatives are bounded on bounded subsets. Moreover, there is $\bUpsilon\in\XXXX$ such that $S(\overline{\varUpsilon},0)=\overline{\varUpsilon}$.
	\item[\hypertarget{(AC)}{\bf(AC)}] 
	For any $\e>0$, there is an integer $n\ge1$ such that, for any initial point $\varUpsilon\in\XXXX$ and any target $\widehat\varUpsilon\in\XXXX$, one can find controls $\zeta_1,\dots,\zeta_n\in\KK$ satisfying the inequality
\begin{equation}\label{AC}
d_\HHHH\bigl(S_n(\varUpsilon; \zeta_1, \ldots, \zeta_n),\widehat \varUpsilon\,\bigr)\le \e, 
\end{equation}
where  $S_n(\varUpsilon; \eta_1, \ldots, \eta_n)$ stands for the vector~$\varUpsilon_n$ defined by relations~\eqref{RDS} with $\varUpsilon_0=\varUpsilon$.  
	\item[\hypertarget{(ACL)}{\bf(ACL)}] 
	For any $\varUpsilon\in\XXXX$ and $\eta\in\KK$, the derivative $(D_\eta S)(\varUpsilon,\eta):\EEEE\to\HH\times T_y\YYYY$, with $y=S(\varUpsilon,\eta)$, has a dense image.
\end{itemize}
In applications to randomly forced PDEs, the mapping~$S$ is the time-$1$ shift along the trajectories of the system. The first part of Hypothesis~\hyperlink{(R)}{(R)} is a  regularisation property of the flow, and the second part asserts  that the unperturbed dynamics has at least one fixed point. Hypothesis~\hyperlink{(AC)}{(AC)} is the standard property of global approximate controllability, with control functions in the support of the noise with  no restriction imposed on the time of control.  Hypothesis~\hyperlink{(ACL)}{(ACL)} is a similar property for the linearised equation, but it allows for a larger control space and requires the time of control to be fixed. These two properties are often satisfied if the support of the driving noise is sufficiently large. 

We shall assume, in addition, that the noise has a decomposable structure in the following sense.

\begin{itemize}
	\item[\hypertarget{(D)}{\bf(D)}]\sl 
	The support of~$\ell$ is compact, and there are two sequences of closed subspaces~$\{\FFFF_n\}$ and $\{\GGGG_n\}$ in~$\EEEE$ such that $\dim \FFFF_n<\infty$ and $\FFFF_n\subset \FFFF_{n+1}$ for any $n\ge1$, the union $\cup_n\FFFF_n$ is dense in~$\EEEE$, and the following properties hold.
	\subitem{$\bullet$}
	The space~$\EEEE$ is the direct sum of~$\FFFF_n$ and~$\GGGG_n$, and the norms of the corresponding projections~${\mathsf P}_n$ and~${\mathsf Q}_n$ are bounded uniformly in~$n\ge 1$. 
	\subitem{$\bullet$} 
	The measure~$\ell$ is the product of its projections ${\mathsf P}_{n*}\ell$ and~${\mathsf Q}_{n*}\ell$ for any $n\ge 1$. Moreover, ${\mathsf P}_{n*}\ell$ has $C^1$-smooth density with respect to the Lebesgue measure~on~$\FFFF_n$.
\end{itemize}
Let us note that this condition implies, in particular, that the sequence of projections~$\{{\mathsf P}_n\}$ converges to  the identity operator in~$\EEEE$ in the strong operator topology. In what follows, we deal with the restriction of~\eqref{RDS} to the invariant set~$\XXXX$. We introduce the empirical measures of trajectories by the formula~\eqref{empiricalmeasureMP}, in which  $\UUpsilon_n=(\varUpsilon_k,k\ge n)$ and $\varUpsilon_k=S_k(\varUpsilon;\eta_1,\dots,\eta_k)$. Thus, for each $\XXXX$-valued random variable~$\varUpsilon$, the sequence $\{\nnu_t^\varUpsilon\}$ consists of random probability measures on the product space $\XXX:=\XXXX^{\Z_+}$. 

To formulate the main result of this section, we first recall some definitions. The spaces~$\XXX$ and~$\PP(\XXX)$ are endowed with the Tikhonov and weak$^*$ topologies and the corresponding Borel $\sigma$-algebras. A mapping $\III:\PP(\XXX)\to[0,+\infty]$ is called a {\it good rate function\/} if it is convex and has compact level sets. The latter property reduces to the lower semicontinuity of~$\III$ since~$\PP(\XXX)$ is a compact space. We shall say that the sequence $\{\nnu_t^\varUpsilon\}$ satisfies the {\it uniform LDP\/} with the rate function~$\III$ if 
\begin{align} 
	-\III(\dot\Gamma)
	&\le\liminf_{t\to\infty}t^{-1}\log\inf_{\varUpsilon\in\XXXX}\IP\{\nnu_t^\varUpsilon\in\Gamma\}\notag\\
	&\le\limsup_{t\to\infty}t^{-1}\log\sup_{\varUpsilon\in\XXXX}\IP\{\nnu_t^\varUpsilon\in\Gamma\} 
	\le -\III(\overline\Gamma)\label{LDPforEM}
\end{align}
for any Borel subset $\Gamma\subset\PP(\XXX)$, where $\dot\Gamma$ and~$\overline\Gamma$ stand for the interior and closure of~$\Gamma$, and~$\III(A)$ is the infimum of~$\III$ over~$A$. In view of the Markov property, if $\{\nnu_t^\varUpsilon\}$ satisfies the {\it uniform LDP\/}, then inequality~\eqref{LDPforEM} remains valid if the infimum and supremum are taken over all $\XXXX$-valued random variables~$\varUpsilon$ independent of the sequence~$\{\eta_k\}$. 

A measure $\llambda\in\PP(\XXX)$ is said to be {\it shift-invariant\/} if it is invariant under the mapping $t\mapsto t+1$. The set of all shift-invariant measures is denoted by~$\PP_s(\XXX)$. By Kolmogorov's theorem, any shift-invariant measure can be extended in a unique manner to a shift-invariant measure on~$\XXXX^\Z$, and we use the same notation for the extended measure. Finally, given a shift-invariant measure $\llambda\in\PP(\XXX)$, we denote by~$\llambda_-$ its projection to $\XXX_-:=\XXXX^{\Z_-}$, and by $\{\lambda(\UUpsilon,\cdot),\UUpsilon\in\XXXX^{\Z_-}\}$ the projection to the first component of the regular conditional probability of~$\llambda$ with respect to its projection to~$\XXX_-$. 

\begin{theorem} \label{t-uniformLDP}
	Suppose that Hypotheses~\hyperlink{(R)}{\rm(R)}, \hyperlink{(AC)}{\rm(AC)}, \hyperlink{(ACL)}{\rm(ACL)}, and~\hyperlink{(D)}{\rm(D)} hold for the random dynamical system~\eqref{RDS}. Then the following holds. 
		
	\smallskip
	{\sc{Uniform LDP}}. The empirical measures~$\{\nnu_t^\varUpsilon\}_{t\ge1}$ satisfy the uniform LDP with a good rate function~$\III:\PP(\XXX)\to[0,+\infty]$. In particular, the LDP holds for the empirical measures of a stationary process. 

	\smallskip
	{\sc{Rate function}}. The rate function~$\III$ is affine and is given by the Donsker--Varadhan entropy formula: 
	\begin{equation} \label{DVentropy}
		\III(\llambda)=
		\left\{
		\begin{array}{cl}
		\displaystyle\int_{\XXX_-}\Ent\bigl(\lambda(\UUpsilon,\cdot)\,|\,P_1(\varUpsilon_0,\cdot)\bigr)\,\llambda_-(\dd\UUpsilon) & \quad \mbox{if $\llambda\in\PP_s(\XXX)$},\\
			+\infty & \quad\mbox{otherwise}, 
		\end{array}
		\right.
	\end{equation}
	where $\Ent(\mu\,|\,\nu)$ is the relative entropy of~$\mu$ with respect to~$\nu$, and $P_1(\varUpsilon,\cdot)$ is the transition function for the Markov process defined by~\eqref{RDS}.  
\end{theorem}

The above theorem is applicable to various parabolic-type PDEs with a smooth random force. In this context, the case when all the Fourier modes are forced was studied in~\cite{JNPS-cpam2015,nersesyan-2019} (see also~\cite{gourcy-2007b,WX-2018} for the case of an irregular noise). The scope of applicability of Theorem~\ref{t-uniformLDP} is much larger, allowing for treatment of PDEs with very degenerate noise, such as those studied in~\cite{shirikyan-asens2015,KNS-2018}. Furthermore, even though the Donsker--Varadhan formula~\eqref{DVentropy} is by now very well known (see~\cite[Section~5.4]{DS1989} or~\cite[Section~6.5]{DZ2000}), to the best of our knowledge, all available proofs deal with the case of {\it strong\/} Feller Markov processes. Our proof presented in Section~\ref{ss-RF} is valid for Markov processes with Feller property in a compact metric space, and its extension to the non-compact case does not encounter any difficulties.

The proof of Theorem~\ref{t-uniformLDP} is based on Kifer's criterion for LDP and a  result on the asymptotoic behaviour of generalised Markov semigroups. The scheme of the proof is presented in Section~\ref{ss-proofLDPscheme}, and the details are given in Sections~\ref{ss-asymptoticsFK}--~\ref{ss-RF}.

\begin{remark}
	It is tempting to use the explicit formula~\eqref{DVentropy} for the large deviations rate function to derive the level-$3$ fluctuation relation~\eqref{level3FR}. Namely, for an integer $k\in\Z$ and a measure $\llambda\in\PP(\XXXX^\Z)$, we denote by~$\Z_{k}$ the set of the integers not exceeding~$k$ and by~$\llambda_-^k$ the projection of~$\llambda$ to $\XXXX^{\Z_{k}}$,  so that $\Z_{0}=\Z_-$ and $\llambda_-^0=\llambda_-$. Using the explicit  formula for the relative entropy in terms of densities and the relation $\llambda_-^1(\dd\UUpsilon,\dd\varUpsilon_1)=\llambda_-(\dd\UUpsilon)\lambda(\UUpsilon,\dd\varUpsilon_1)$, for any $\llambda\in\PP_s(\XXX)$ we can write\footnote{It is easy to give a rigorous meaning to the formal expressions used in the calculations below. Since these calculations do not play a role in this work, we omit the details.}
	\begin{align*}
	\III(\llambda)&=\int_{\XXX_-}\biggl\{\int_\XXXX\log\frac{\lambda(\UUpsilon,\dd\varUpsilon_1)}{P_1(\varUpsilon_0,\dd\varUpsilon_1)}\lambda(\UUpsilon,\dd\varUpsilon_1)\biggr\}\llambda_-(\dd\UUpsilon)\\
	&=\int_{\XXX_-\times\XXXX}\log\frac{\llambda_-(\dd\UUpsilon)\lambda(\UUpsilon,\dd\varUpsilon_1)}{\llambda_-(\dd\UUpsilon)P_1(\varUpsilon_0,\dd\varUpsilon_1)}\llambda_-^1(\dd\UUpsilon,\dd\varUpsilon_1)
	=\Ent\bigl(\llambda_-^1\,|\,\llambda_-\otimes P_1\bigr),
	\end{align*}
where $\frac{\mu(\dd x)}{\nu(\dd x)}$ denotes the density of~$\mu$ with respect to~$\nu$, and $\llambda_-\otimes P_1$ stands for the measure acting on a function~$F$ by the formula 
$$
\langle F, \llambda_-\otimes P_1\rangle=\int_{\XXX_-}\biggl\{\int_{\XXXX}F(\UUpsilon,\varUpsilon_1)P_1(\varUpsilon_0,\dd\varUpsilon_1)\biggr\}\llambda_-(\dd\UUpsilon). 
$$
Now let $\theta:\XXXX^\Z\to\XXXX^\Z$ be the natural time reversal taking $(\varUpsilon_k,k\in\Z)$ to $(\varUpsilon_{-k},k\in\Z)$ and let $\ttheta:\PP(\XXXX^\Z)\to \PP(\XXXX^\Z)$ be the associated involution in the space of measures. Assuming that $P_1(\varUpsilon_0,\dd\varUpsilon_1)$ has a positive density $\rho(\varUpsilon_0,\varUpsilon_1)$ with respect to a reference measure, using the above formula for~$\III$, and carrying out some simple transformations, we get
\begin{align}
\III(\llambda\circ\ttheta)-\III(\llambda)
&=\Ent\bigl((\llambda\circ\ttheta)_-^1\,|\,(\llambda\circ\ttheta)_-\otimes P_1\bigr)-\Ent\bigl(\llambda_-^1\,|\,\llambda_-\otimes P_1\bigr)\notag\\
&= \int_\XXX\log \frac{\rho(\varUpsilon_0,\varUpsilon_1)}{\rho(\varUpsilon_1,\varUpsilon_0)}\,\llambda(\dd\UUpsilon). \label{level3FRformal}
\end{align}
Hence, denoting by $\sigma(\UUpsilon)$ the integrand in~\eqref{level3FRformal}, we obtain the level-$3$ fluctuation relation~\eqref{level3FR}, in which $\ep(\llambda)$ is the mean value of~$\sigma$ with respect to~$\llambda$.

Unfortunately, the above argument is purely formal  since the logarithmic ratio  in~\eqref{level3FRformal} may not be well defined, as is expected in the case of the Navier--Stokes system with a smooth noise. Thus, the validity of level-$3$ LDP is not sufficient for the fluctuation relation~\eqref{level3FR} to be true.  On the other hand, the above argument can be justified under some additional hypotheses on the map~$S$ and the driving noise~$\eta_k$; see~\cite{JNPS-cmp2015}.
\end{remark}

\subsection{General scheme of the proof of Theorem~\ref{t-uniformLDP}}
\label{ss-proofLDPscheme}

\subsubsection*{Reduction to LDP for finite segments}
The first step in the proof of Theorem~\ref{t-uniformLDP} consists of an application of the Dawson--G\"artner theorem, which allows one to reduce the required result to the LDP for the sequence 
\begin{equation} \label{m-empirical}
	\nnu_t^\varUpsilon(r)=\frac{
	1}{t}\sum_{k=0}^{t-1}\delta_{\UUpsilon_k^r},
\end{equation}
where $\UUpsilon_k^r=[\varUpsilon_k,\dots,\varUpsilon_{k+r-1}]$, and $\{\varUpsilon_k\}$ is the trajectory defined by~\eqref{RDS} with $\varUpsilon_0=\varUpsilon$. Thus, $\{\nnu_t^\varUpsilon(r)\}$ is a sequence of random probability measures on the $r$-fold product~$\XXXX(r)$ of the space~$\XXXX$. In view of Theorem~4.6.1 in~\cite{DZ2000}, if for all $r\ge1$ the sequence~$\{\nnu_t^\varUpsilon(r)\}$ satisfies a uniform LDP with a good rate function~$\III_r:\PP(\XXXX(r))\to[0,+\infty]$, then so does the sequence~$\{\nnu_t^\varUpsilon\}$, with the rate function
\begin{equation} \label{ratefunction-I}
\III(\llambda)=\sup_{r\ge1}\III^r\bigl(\Pi^r_*(\llambda)\bigr),
\end{equation}
where $\Pi^r:\XXX\to\XXXX(r)$ stands for the natural projection to the first~$r$ components. We shall prove the uniform LDP for~$\{\nnu_t^\varUpsilon(r)\}$ with an arbitrary~$r\ge1$, establish a variational formula for the corresponding rate function~$\III^r$, and use relation~\eqref{ratefunction-I} to obtain the Donsker--Varadhan entropy formula~\eqref{DVentropy}.

\subsubsection*{Application of Kifer's theorem}
To prove the uniform LDP for the sequence~$\{\nnu_t^\varUpsilon(r)\}_{t\ge1}$ for a fixed $r\ge1$, we shall apply Kifer's theorem~\cite{kifer-1990}, which is  recalled in Section~\ref{ss-kifer}. To this end, we define the set $\Theta=\{\theta=(t,\varUpsilon), t\in\N,\varUpsilon\in\XXXX\}$ and endow it with a partial order~$\prec$ defined by the following rule: 
$$
(t_1,\varUpsilon_1)\prec (t_2,\varUpsilon_2)\quad\mbox{if and only if}\quad t_1\le t_2. 
$$
The sequence $\{\nnu_t^\varUpsilon(r)\}$ will be regarded as a directed family indexed by $\theta\in\Theta$, and  in what follows we shall often write $\{\nnu_\theta\}$, dropping the fixed integer~$r$ from the notation. Let us suppose that, for any $V\in C(\XXXX(r))$, the limit 
\begin{equation} \label{limit-pressure}
	Q^r(V)=\lim_{\theta\in\Theta}t^{-1}\log \E \exp\bigl(t\langle V,\nnu_\theta\rangle\bigr)
\end{equation}
exists, and let~$\III^r:\MM(\XXXX(r))\to[0,+\infty]$ be  its Legendre transform; see relation~\eqref{legendre-abstract} for a definition. If, in addition to the existence of limit~\eqref{limit-pressure}, there exists  a dense subspace $\VV\subset C(\XXXX(r))$ such that, for any $V\in\VV$, the equation\footnote{The lower semicontinuity of~$\III^r$ and the inversion formula for the Legendre transform imply that Eq.~\eqref{equilibrium} has at least one solution.}
\begin{equation} \label{equilibrium}
	\langle V,\ssigma\rangle-\III^r(\ssigma)=Q^r(V)
\end{equation}
has a unique solution $\ssigma\in\PP(\XXXX(r))$, then the validity of the LDP follows immediately from Theorem~\ref{t-kifer}. We show in the next step how to reduce the above two properties (existence of limit~\eqref{limit-pressure} and uniqueness of a solution of Eq.~\eqref{equilibrium} for $V$ in a dense subspace) to a study of the large-time asymptotics of a Feynman--Kac semigroup. 

\subsubsection*{Reduction to a study of  Feynman--Kac semigroups}
Let us consider the following random dynamical system in~$\XXXX(r)$: 
\begin{equation} \label{RDS-segments}
	\UUpsilon_k(r)=\bfS\bigl(\UUpsilon_{k-1}(r),\eta_k\bigr), \quad k\ge1,
\end{equation}
where $\UUpsilon_k(r)=[\varUpsilon_k^1,\dots,\varUpsilon_k^r]$, $\{\eta_k\}$ is the sequence of i.i.d.\ random variables in~$\EEEE$ entering~\eqref{RDS}, and the mapping $\bfS^r:\XXXX(r)\times\EEEE\to\XXXX(r)$ is given by 
\begin{equation} \label{bold-S}
\bfS^r(\varUpsilon^1,\dots,\varUpsilon^r,\eta)=\bigl[\varUpsilon^2,\dots,\varUpsilon^r,S(\varUpsilon^r,\eta)\bigr].
\end{equation}
Equation~\eqref{RDS-segments} is supplemented with the initial condition
\begin{equation} \label{IC-segments}
	\UUpsilon_0(r)=\UUpsilon(r)\in\XXXX(r). 
\end{equation}
Given a function $V\in C(\XXXX(r))$, we consider the operator
\begin{equation} \label{FK-segments}
	\bigl(\PPPP_k^V(r)f\bigr)\bigl(\UUpsilon(r)\bigr)=\E \bigl(\exp\bigl\{V(\UUpsilon_1(r))+\cdots+V(\UUpsilon_k(r))\bigr\}f(\UUpsilon_k(r))\bigr),
\end{equation}
acting in the space~$C(\XXXX(r))$. The Markov property implies that the sequence $\{\PPPP_k^V(r)\}$ is a semigroup in~$C(\XXXX(r))$. A key observation is that, for $V\in C(\XXXX(r))$ and $\UUpsilon(r)\in\XXXX(r)$,
\begin{equation} \label{key-observation}
	Q^r(V)=\lim_{k\to\infty}k^{-1}\log \bigl(\PPPP_k^V(r){\bf1}\bigr)(\UUpsilon(r)), 
\end{equation}
provided that the limit on the right-hand side exists uniformly with respect to the initial point~$\UUpsilon(r)$ and does not depend on it. The latter property is a consequence of the following proposition, which is established in Section~\ref{ss-asymptoticsFK} with the help of Theorem~\ref{t-feynmankac}. 

\begin{proposition} \label{p-asymptoticsFK}
	Under the Hypotheses of Theorem~\ref{t-uniformLDP}, for any integer $r\ge1$ and any function $V\in L_b(\XXXX(r))$, there is a number $\lambda_V>0$,  a positive function $h_V\in C(\XXXX(r))$, and a measure $\mmu_V\in\PP(\XXXX(r))$ such that 
	\begin{gather} 
	\langle h_V,\mmu_V\rangle=1, \quad \PPPP_1^V(r)h_V=\lambda_V h_V, \quad
	\PPPP_1^V(r)^*\mmu_V=\lambda_V \mmu_V,
	\label{eigenvectors} \\
	\bigl\|\lambda_V^{-k}\PPPP_k^V(r)f-\langle f,\mmu_V\rangle h_V\bigr\|_{L^\infty(\XXXX(r))}\to 0 \quad\mbox{as $k\to\infty$},\label{convergenceFK}
	\end{gather}
	where $f\in C(\XXXX(r))$ is an arbitrary function.
\end{proposition}

Convergence~\eqref{convergenceFK}, combined with~\eqref{key-observation} and the inequality 
$$
t^{-1}\bigl|\log\E\exp\bigl(t\langle W,\nnu_\theta\rangle\bigr)-\log\E\exp\bigl(t\langle V,\nnu_\theta\rangle\bigr)\bigr|
\le \|W-V\|_{\infty},
$$
implies that limit~\eqref{limit-pressure} exists for any $V\in C(\XXXX(r))$. Let us briefly outline the well-known argument proving that Proposition~\ref{p-asymptoticsFK} also implies the uniqueness of a solution $\ssigma\in\PP(\XXXX(r))$ for Eq.~\eqref{equilibrium} with an arbitrary~$V$ in the space~$L_b(\XXXX(r))$, which is dense in~$C(\XXXX(r))$; cf.\ \cite[Section~4]{kifer-1990} and~\cite[Section~4]{JNPS-cpam2015}. 

\smallskip
Let us fix any $V\in C(\XXXX(r))$. For any $W\in C(\XXXX(r))$, we consider a semigroup $\QQQQ_k^W(r):C(\XXXX(r))\to C(\XXXX(r))$ with the generator given by 
$$
\QQQQ_1^W(r)f
=\lambda_V^{-1}h_V^{-1}\PPPP_1^V(e^Wh_Vf)
=\lambda_V^{-1}h_V^{-1}\PPPP_1^{V+W}(h_Vf). 
$$
In the case $W=0$, we shall write $\QQQQ_k(r)$. A straightforward calculation shows that $\QQQQ_k$ is Markovian (that is, $\QQQQ_k{\bf1}={\bf1}$) and 
$$
\QQQQ_k^W(r)f
=\lambda_V^{-k}h_V^{-1}\PPPP_k^{V+W}(h_Vf). 
$$
It follows from Proposition~\ref{p-asymptoticsFK} that, for any $W\in L_b(\XXXX(r))$, we have 
$$
Q_V^r(W):=\lim_{k\to\infty}k^{-1}\log \QQQQ_k^W(r){\bf1}=\log\lambda_{V+W}-\log\lambda_V=Q^r(V+W)-Q^r(V).
$$
By the Lipschitz continuity of~$Q^r$ and~$Q_V^r$, the left-most and right-most terms coincide for any $W\in C(\XXXX(r))$. Denoting by~$I_V^r:\PP(\XXXX(r))\to[0,+\infty]$ the Legendre transform of~$Q_V^r$, we see that
\begin{equation} \label{relationforRF}
	I_V^r(\ssigma)=I^r(\ssigma)+Q^r(V)-\langle V,\ssigma\rangle\quad\mbox{for any $\ssigma\in\PP(\XXXX(r))$}. 
\end{equation}
Thus, a measure $\ssigma\in\PP(\XXXX(r))$ is a solution for~\eqref{equilibrium} if and only if $I_V^r(\ssigma)=0$.  Now note that, by Proposition~\ref{p-asymptoticsFK}, the dual semigroup~$\QQQQ_1^V(r)^*$ has a unique stationary measure, which is given by $\ssigma_V=h_V\mmu_V$. Hence, the required uniqueness of solution of~\eqref{equilibrium} will be established if we prove that any~$\ssigma\in\PP(\XXXX(r))$ satisfying $I_V^r(\ssigma)=0$ is a stationary measure for~$\QQQQ_1^V(r)^*$. 

To this end, we repeat the argument used in the proof of Lemma~2.5 in~\cite{DV-1975-I}. Namely, as will be established in Proposition~\ref{p-UFP}, we have 
\begin{equation} \label{DV-r}
	I_V^r(\ssigma)=\sup_{g>0}\int_{\XXXX(r)}\log\frac{g}{\QQQQ_1^V(r)g}\,\dd\ssigma, 
\end{equation}
where the supremum is taken over all positive continuous functions $g:\XXXX(r)\to\R$. If $I_V^r(\ssigma)=0$, then the supremum on the right-hand side of~\eqref{DV-r} is attained at the function $g\equiv1$. It follows that, for any $d\in C(\XXXX(r))$, the function 
$$
F(\e)=\int_{\XXXX(r)}\log\frac{1+\e d}{\QQQQ_1^V(r)(1+\e d)}\,\dd\ssigma
$$
is well defined for $|\e|\ll 1$ and has a local minimum at $\e=0$. Calculating its derivative at zero, we obtain $\langle d,\ssigma\rangle-\langle \QQQQ_1^V(r)d,\ssigma\rangle=0$. Recalling that $d\in C(\XXXX(r))$ was arbitrary, we see that~$\ssigma$ is a stationary measure for~$\QQQQ_1^V(r)^*$.

\smallskip
We have thus established the first part of Theorem~\ref{t-uniformLDP}, and we turn  to the explicit expression for the rate function. Relation~\eqref{DVentropy} is proved in~\cite{DV-1983} in the case when the process is strong Feller. We present here a different argument applicable to our setting. To emphasise its universal character, we do it in a more general setting, under minimal hypotheses. 

\subsubsection*{Donsker--Varadhan entropy formula}
The first step is the derivation of a variational formula for the level-$2$ rate function; cf.~\cite[Section~2]{DV-1975-I}.  Let $\frX$ be a compact metric space and let $P_1(u,\Gamma)$ be a Feller transition function. Given $V\in C(\frX)$, we denote by $\{\PPPP_k^V\}$ a semigroup in~$C(\frX)$ whose generator is given by 
\begin{equation} \label{FK-abstract}
	(\PPPP_1^Vf)(x)=\int_\frX e^{V(y)}f(y)P_1(x,\dd y), \quad f\in C(\frX). 
\end{equation}
In the case $V\equiv0$, we shall write $\PPPP_k$.

\begin{proposition} \label{p-UFP}
Suppose that, for any $V\in C(\frX)$, the limit 
$$
Q(V)=\lim_{k\to\infty}\frac1k\log(\PPPP_k^V{\mathbf1})(x)
$$ 
exists uniformly in $x\in\frX$ and does not depend on~$x$. Then $Q$ is a $1$-Lipschitz convex function such that 
\begin{equation} \label{QVC}
Q(V+C)=Q(V)+C\quad\mbox{for any $V\in C(\frX)$ and $C\in\R$},	
\end{equation}
and its Legendre transform  $I:\MM(\frX)\to[0,+\infty]$ has the form 
\begin{equation} \label{DV-ratefunction}
I(\lambda)=
\left\{
\begin{array}{cl}
\displaystyle
	\sup_{g>0}\int_\frX\log\frac{g}{\PPPP_1g}\,\dd\lambda & \quad\mbox{for  $\lambda\in\PP(\frX)$},\\[4pt]
	+\infty &\quad\mbox{otherwise},
\end{array}
\right.
\end{equation}
where the supremum is taken over all positive functions $g\in C(\frX)$. 
\end{proposition}

We now denote by~$\frX(r)$ the $r$-fold product of the space~$\frX$ and, given a function  $V\in C(\frX(r))$, consider a semigroup $\PPPP_k^V(r)$ on~$C(\frX(r))$ with the generator\,\footnote{In the language Markov processes, this means that $\PPPP_k^V(r)$ is the Feynman--Kac semigroup associated with the evolution of words of length~$r$.}
\begin{equation} \label{FKr-abstract}
	\bigl(\PPPP_1^V(r)f\bigr)(\xxx^r)=\int_\frX e^{V(x_2,\dots,x_r,y)}f(x_2,\dots,x_r,y)P_1(x_r,\dd y),
\end{equation}
where $\xxx^r=[x_1,\dots,x_r]\in\frX(r)$. In the case $V\equiv0$, we shall write $\PPPP_k(r)$. Finally, let us denote $\bfrX=\frX^\N$ and $\bfrX_-=\frX^{\Z_-}$. 

\begin{proposition} \label{p-RF}
	Suppose that, for any integer $r\ge1$ and any $V\in C(\frX(r))$, there is a uniform limit 
	$$
	\QQQ^r(V)=\lim_{k\to\infty}\frac{1}{k}\log \bigl(\PPPP_k^V(r){\bf1}\bigr)(\xxx^r), 
	$$
	independent of $\xxx^r\in\frX(r)$. Let $\III^r:\PPPP(\frX(r))\to [0,+\infty]$ be the Legendre transform of~$\QQQ^r$ and let $\III:\PP(\boldsymbol{\frX})\to[0,+\infty)$ be defined by~\eqref{ratefunction-I}. Then, for any shift-invariant measure $\llambda\in\PP(\bfrX)$, we have
	\begin{equation} \label{entropyformula}
		\III(\llambda)=\int_{\bfrX_-}\Ent\bigl(\lambda(\xxx,\cdot)\,|\,P_1(x_0,\cdot)\bigr)\,\llambda_-(\dd\xxx),
	\end{equation}
	where we use the same conventions as in~\eqref{DVentropy}.  
\end{proposition}

Propositions~\ref{p-UFP} and~\ref{p-RF} are established in Sections~\ref{ss-UFP} and~\ref{ss-RF}, respectively. Going back to the proof of Theorem~\ref{t-uniformLDP}, we note that Proposition~\ref{p-RF} implies~\eqref{DVentropy} for $\llambda\in\PP_s(\XXX)$. The fact that $\III(\llambda)$ is infinite when $\llambda$ is not shift-invariant follows from the observation that $\nnu_t^\varUpsilon$ is {\it exponentially equivalent\/}\footnote{See Section~4.2.10 in~\cite{DZ2000} for a definition.} to a sequence of random probability measures concentrated on shift-invariant measures on~$\XXX$; see~\cite[Section~1]{DV-1983}. Namely, together with~$\nnu_t^\varUpsilon$, let us consider the sequence 
$$
\tnnu_t^\varUpsilon=t^{-1}\sum_{n=0}^{t-1}\delta_{\tUUpsilon_n(t)},
$$
where $\{\tUUpsilon_0(t)\}$ is a $t$-periodic sequence whose first~$t$ components coincide with those of~$\UUpsilon_0$, and $\tUUpsilon_n(t)$ is obtained from~$\tUUpsilon_0(t)$ by deleting the first~$t$ components. It is straightforward to check that~$\tnnu_t^\varUpsilon$ and~$\nnu_t^\varUpsilon$ are exponentially equivalent and that
$$
\IP\bigl\{\tnnu_t^\varUpsilon\in \PP_s(\XXX)\bigr\}=1\quad\mbox{for any $t\ge1$, $\varUpsilon\in\XXXX$}.
$$
Since exponentially equivalent sequences satisfy the same LDP, we conclude that the rate function~$\III$ is  infinite on $\MM(\XXX)\setminus\PP(\XXX)$. Finally, the proof of the affine property of~$\III$ given in~\cite[Theorem~3.5]{DV-1983} uses only relation~\eqref{DVentropy} and therefore remains valid in our setting. This completes the proof of Theorem~\ref{t-uniformLDP}. 

\subsection{Proof of Proposition~\ref{p-asymptoticsFK}}
\label{ss-asymptoticsFK}

We first  outline the main idea of the proof, which is based on an application of Theorem~\ref{t-feynmankac}. According to that result, to prove the required claims, we need to check the uniform Feller and uniform irreducibility properties~\hyperlink{(UF)}{(UF)} and~\hyperlink{(UI)}{(UI)}. The first of them will be established with the help of a coupling technique; see Proposition~\ref{p-coupling}. On the other hand, the uniform irreducibility is not valid in~$\XXXX(r)$, and we have to restrict ourselves to the domain of attainability~$\aA(r)$, for which the validity of~\hyperlink{(UI)}{(UI)} follows easily from the approximate controllability~\hyperlink{(AC)}{(AC)}. Thus, we can apply Theorem~\ref{t-feynmankac} with $X=\aA(r)$. Finally, to establish convergence~\eqref{convergenceFK}, we shall prove that, for any $\UUpsilon\in\XXXX(r)$, there is $\tUUpsilon\in\aA(r)$ such that 
\begin{equation} \label{approximation}
\bigl|\log(\PPPP_k^V(r)f)(\UUpsilon)-\log(\PPPP_k^V(r)f)(\tUUpsilon)\bigr|\le C\|f\|_{\infty} \quad\mbox{for all $k\ge1$},
\end{equation}
where $f\in C(\XXXX(r))$ is an arbitrary function, and the constant $C>0$ does not depend on~$\UUpsilon$, $\tUUpsilon$, and~$k$. The details are split  into three steps. 

\smallskip
{\it Step 1: Reduction to the domain of attainability\/}. 
Let us recall that $\KK\subset\EEEE$ stands for the support of the law~$\ell$. Setting $\bUUpsilon=[\bUpsilon,\dots,\bUpsilon]$, we  define a sequence $\{\aA_k(r)\}_{k\ge0}$ of compact subsets of~$\XXXX(r)$ by the following rule:
$$
\aA_0(r)=\{\bUUpsilon\}, \quad 
\aA_k(r)=\overline{\bfS^r(\aA_{k-1}(r),\KK)}\quad\mbox{for $k\ge1$}, 
$$
where $\bfS^r$ is defined by~\eqref{bold-S}, and $\overline{B}$ stands for the closure of~$B\subset\XXXX(r)$. Since $S(\bUpsilon,0)=\bUpsilon$, the sequence $\{\aA_k(r)\}_{k\ge0}$ is increasing. We  denote by~$\aA(r)\subset\XXXX(r)$ the closure of the union~$\cup_{k\ge1}\aA_k(r)$. A simple compactness argument yields that, for any $\delta>0$, there is an integer $m\ge1$ such that~$\aA(r)$ is a subset of the $\delta$-neighbourhood of~$\aA_m(r)$.

Denoting by ${\mathsf P}:\XXXX(r)\to\XXXX$ the projection taking $[\varUpsilon^1,\dots,\varUpsilon^r]$ to~$\varUpsilon^r$, let us show that ${\mathsf P}(\aA(r))=\XXXX$. Indeed, since~$\aA(r)$ is compact and~$\mathsf P$ is continuous, we see that the projection ${\mathsf P}(\aA(r))$ is closed, and so it suffices to prove that it is dense in~$\XXXX$. Fix $\varUpsilon\in\XXXX$ and $\e>0$. By~\hyperlink{(AC)}{(AC)}, there is an integer $k\ge r$ and vectors $\eta_1,\dots,\eta_k\in\KK$ such that 
\begin{equation} \label{r-approximation}
d_\HHHH\bigl(S_k(\bUpsilon;\eta_1,\dots,\eta_k),\varUpsilon\bigr)<\e. 	
\end{equation}
Let us denote by $\bfS_k^r(\UUpsilon;\eta_1,\dots,\eta_k)$ the trajectory of~\eqref{RDS-segments} issued from~$\UUpsilon$. Inequality~\eqref{r-approximation} and the definition of~$\aA_k(r)$ imply that the vector ${\mathsf P}\bfS_k^r(\UUpsilon;\eta_1,\dots,\eta_k)$ belongs to the $\e$-neighbourhood of~$\varUpsilon$. This proves the required density.

\smallskip
We now prove~\eqref{approximation}. Without loss of generality we may assume that  $f\in C(\XXXX(r))$ is non-negative. Let $\UUpsilon\in\XXXX(r)$. Since ${\mathsf P}(\aA(r))=\XXXX$, we can find $\tUUpsilon\in\aA(r)$ such that ${\mathsf P}(\UUpsilon)={\mathsf P}(\tUUpsilon)$. Since $\bfS_k^r(\UUpsilon;\eta_1,\dots,\eta_k)$ depends only on the $r^\text{\rm th}$ component of~$\UUpsilon$ for $k\ge r$, we have
\begin{equation} \label{bfS-independence}
	\bfS_k^r(\UUpsilon;\eta_1,\dots,\eta_k)
	=\bfS_k^r(\tUUpsilon;\eta_1,\dots,\eta_k)
	\quad\mbox{for $k\ge r$}. 
\end{equation}
Denoting by~$\UUpsilon_k$ and~$\tUUpsilon_k$ the left- and right-hand terms in~\eqref{bfS-independence}, we see that $\UUpsilon_k=\tUUpsilon_k$ for $k\ge r$. It follows from~\eqref{FK-segments} that 
\begin{align*}
	\bigl(\PPPP_k^V(r)f\bigr)(\UUpsilon)
	&=\E \,\bigl(\exp\bigl\{V(\UUpsilon_1)+\cdots+V(\UUpsilon_k)\bigr\}f(\UUpsilon_k)\bigr)\\
	&\le \exp(2r\|V\|_\infty)\,
	\E \,\bigl(\exp\bigl\{V(\tUUpsilon_1)+\cdots+V(\tUUpsilon_k)\bigr\}f(\UUpsilon_k)\bigr)\\
	&=\exp(2r\|V\|_\infty)\bigl(\PPPP_k^V(r)f\bigr)(\tUUpsilon). 
\end{align*}
By symmetry, we can exchange the roles of~$\UUpsilon$ and~$\tUUpsilon$, and the resulting inequalities imply~\eqref{approximation} with $C=2r\|V\|_\infty$. Thus, we need to construct $h_V\in C(\XXXX(r))$, $\mmu_V\in\PP(\XXXX(r))$, and~$\lambda_V>0$ satisfying relations~\eqref{eigenvectors} and to establish~\eqref{convergenceFK} with the $L^\infty$-norm on~$\aA(r)$. To this end, we shall prove that the Hypotheses of Theorem~\ref{t-feynmankac} are satisfied with $X=\aA(r)$ and $\CC=L_b(\aA(r))$. 

\smallskip
{\it Step 2: Uniform irreducibility\/}. Our goal is to find an integer $n\ge1$ and a number $p>0$ such that
\begin{equation} \label{UI-transition}
	\IP_\UUpsilon\bigl\{\UUpsilon^n\in B_{\XXXX(r)}(\hUUpsilon,\e)\bigr\}\ge p
	\quad\mbox{for any $\UUpsilon,\hUUpsilon\in \aA(r)$},
\end{equation}
where the subscript~$\UUpsilon$ on the left-hand side means that we consider the trajectory of~\eqref{RDS-segments} issued from~$\UUpsilon$. Simple arguments based on the concepts of the support of a measure and of compactness show that~\eqref{UI-transition} follows if for any $\e>0$ we can find an integer $n\ge1$ such that, for arbitrary $\UUpsilon,\hUUpsilon\in \aA(r)$ and some suitable $\eta_1,\dots,\eta_n\in\KK$, 
\begin{equation} \label{AC-timen}
	d_r\bigl(\UUpsilon_n,\hUUpsilon\bigr)<\e, 
\end{equation}
where $d_r$ stands for the distance in~$\XXXX(r)$ defined as the maximum of the distances between the components, and $\UUpsilon_k=\bfS_k^r(\UUpsilon,\eta_1,\dots,\eta_k)$ is the trajectory of~\eqref{RDS-segments} issued from~$\UUpsilon$; see Section~3.3.2 in~\cite{KS-book} and Section~4 in~\cite{JNPS-cpam2015}. 

The construction of the controls~$\eta_1,\dots,\eta_n$ is carried out in two steps: we first steer the trajectory to a point close to~$\bUUpsilon$ and then use the definition of~$\aA(r)$ to steer it further to the neighbourhood of~$\hUUpsilon$. More precisely, as it was mentioned in Step~1, we can find $m\ge1$ such that $\aA(r)$ is included in the $\e/2$-neighbourhood of~$\aA_m(r)$. Hence, there is $\hUUpsilon_1\in\aA_m(r)$ such that $d_r(\hUUpsilon,\hUUpsilon_1)<\e/2$. Furthermore, by the definition of~$\aA_m(r)$, we can find vectors $\eta_1,\dots,\eta_{m}\in\KK$ such that $\bfS_m^r(\bUUpsilon;\eta_1,\dots,\eta_{m})=\hUUpsilon_1$. By continuity, there is a number $\delta>0$ such that, for any $\UUpsilon'\in\XXXX(r)$ satisfying  $d_r(\UUpsilon',\bUUpsilon)\le \delta$, we have $d_r(\bfS_m^(\UUpsilon';\eta_1,\dots,\eta_{m}),\hUUpsilon_1)<\e/2$, so that 
\begin{equation} \label{approximate-control1}
	d_r\bigl(\bfS_m^r(\UUpsilon';\eta_1,\dots,\eta_{m}),\hUUpsilon\bigr)<\e.
\end{equation}
By~\hyperlink{(AC)}{(AC)}, there is an integer $l\ge1$ and controls $\zeta_{1},\dots,\zeta_l\in\KK$ such that $d_r(S_l(\Pi(\UUpsilon);\zeta_1,\dots,\zeta_l),\bUpsilon)<\delta$. This observation and the relation $S(\bUpsilon,0)=\bUpsilon$ yield that 
$$
d_r\bigl(\bfS_{l+r-1}^r(\UUpsilon;\zeta_1,\dots,\zeta_l,\underbrace{0,\dots,0}_{\text{$r-1$ times}}\!),\hUUpsilon\bigr)<\delta.
$$
Combining this with~\eqref{approximate-control1}, we derive 
$$
d_r\bigl(\bfS_{l+m+r-1}^r(\UUpsilon;\zeta_1,\dots,\zeta_l,0,\dots,0, \eta_1,\dots,\eta_m),\hUUpsilon\bigr)<\e.
$$
This proves the required inequality~\eqref{AC-timen} with the integer $n=m+l+r-1$ not depending on~$\UUpsilon$ and~$\hUUpsilon$. 

\smallskip
{\it Step 3: Uniform Feller property\/}. 
We shall show that, for any $\varUpsilon,\varUpsilon'\in\XXXX(r)$, $k\ge r$, and non-negative functions $V,f\in L_b(\XXXX(r))$, 
\begin{equation} \label{lipschitz}
	\bigl|\bigl(\PPPP_k^V(r)f\bigr)(\UUpsilon)-\bigl(\PPPP_k^V(r)f\bigr)(\UUpsilon')\bigr|
	\le C\|f\|_{L}\bigl\|\PPPP_k^V(r){\bf1}\bigr\|_{\infty}d_r(\UUpsilon,\UUpsilon'),
\end{equation}
where $C>0$ is a number not depending on $k$ and~$f$, and both $L^\infty$ and $L_b$ norms on the right-hand side are taken over~$\XXXX(r)$. This will obviously imply the validity of~\hyperlink{(UF)}{(UF)} with $\CC=\{f\in L_b(\XXXX(r)):f\ge1\}$. 

Fix the  initial points $\UUpsilon,\UUpsilon'\in\XXXX(r)$ and denote by~$\{\UUpsilon_{\!k}\}$ and~$\{\UUpsilon_{\!k}'\}$ the trajectories of~\eqref{RDS-segments} issued from them. Let~$\{\tUpsilon_k\}$ and~$\{\tUpsilon_k'\}$ be the trajectories constructed in Corollary~\ref{c-coupledtrajectories} for the initial points~$\mathsf P(\UUpsilon)=\varUpsilon_r$ and~$\mathsf P(\UUpsilon')=\varUpsilon_r'$, respectively. Note that they depend on the choice of the parameter $q\in(0,1)$ that will be specified below. We set 
$$
\tUUpsilon_{\!k}=[\tUpsilon_{k-r+1},\dots,\tUpsilon_{k}], \quad \tUUpsilon_{\!k}'=[\tUpsilon_{k-r+1}',\dots,\tUpsilon_{k}'], 
$$ 
where $\tUpsilon_{j-r}=\varUpsilon_{j}$  and $\tUpsilon_{j-r}'=\varUpsilon_{j}'$ for $2\le j\le r$. We set $D(j)=d_r(\tUUpsilon_j,\tUUpsilon_j')$, and introduce the events
\begin{align*}
	G_l(q)&=\bigl\{D(j)\le q^{j-r+1}d_r\bigl(\UUpsilon,\UUpsilon'\bigr)\mbox{ for }0\le j<l, D(l)>q^{l-r+1}d_r\bigl(\UUpsilon,\UUpsilon'\bigr) \bigr\}, \\
	G_l'(q)&=\bigl\{D(j)\le q^{j-r+1}d_r\bigl(\UUpsilon,\UUpsilon'\bigr)\mbox{ for }0\le j\le l\bigr\},
\end{align*}
where $l\ge1$. It follows from~\eqref{squeezing} that
\begin{equation} \label{squeezing-segments}
	\IP\bigl(G_l(q)\bigr)\le C_1q^l d_r(\UUpsilon,\UUpsilon')\quad\mbox{for all $l\ge0$},
\end{equation}
where $C_1>0$ does not depend on~$l$. Since the laws of the trajectories~$\{\UUpsilon_{\!k}\}$ and~$\{\UUpsilon_{\!k}'\}$ coincide with those of~$\{\tUUpsilon_{\!k}\}$ and~$\{\tUUpsilon_{\!k}'\}$ respectively, we have
$$
\bigl(\PPPP_k^V(r)f\bigr)(\UUpsilon)
=\E\,\bigl(\exp\{V(\tUUpsilon_1)+\cdots+V(\tUUpsilon_k)\}f(\tUUpsilon_k)\bigr)=\E\,\bigl(\Xi_k(\UUpsilon)f(\tUUpsilon_k)\bigr),
$$
where $\Xi_k(\UUpsilon)=\exp\{V(\tUUpsilon_1)+\cdots+V(\tUUpsilon_k)\}$, and a similar representation holds for $(\PPPP_k^V(r)f)(\UUpsilon')$. Setting 
\begin{align*}
I_k^l(\UUpsilon,\UUpsilon')&=\E\bigl\{\I_{G_l(q)}\bigl(\Xi_k(\UUpsilon)f(\tUUpsilon_k)
-\Xi_k(\UUpsilon')f(\tUUpsilon_k')\bigr)\bigr\},\\
J_k(\UUpsilon,\UUpsilon')&=\E\bigl\{\I_{G_l'(q)}\bigl(\Xi_k(\UUpsilon)f(\tUUpsilon_k)
-\Xi_k(\UUpsilon')f(\tUUpsilon_k')\bigr)\bigr\},
\end{align*}
where $\I_G$ stands for the indicator function of~$G$, we can write
\begin{align} 
	\Delta_k(\UUpsilon,\UUpsilon'):&=\bigl(\PPPP_k^V(r)f\bigr)(\UUpsilon)-\bigl(\PPPP_k^V(r)f\bigr)(\UUpsilon')\notag\\
	&=\sum_{l=1}^kI_k^l(\UUpsilon,\UUpsilon')+J_k(\UUpsilon,\UUpsilon'). \label{PkVfU}
\end{align}
The Markov property and inequality~\eqref{squeezing-segments} imply that 
\begin{align}
	I_k^l(\UUpsilon,\UUpsilon')
	&\le \E\bigl\{\I_{G_l(q)}\Xi_k(\UUpsilon)f(\tUUpsilon_k)\bigr\}=\E\bigl\{\I_{G_l(q)}\,\E\bigl(\Xi_k(\UUpsilon)f(\tUUpsilon_k)\,|\,\FF_l\bigr)\bigr\}\notag\\
	&\le\|f\|_{\infty}\exp\bigl(l\,\|V\|_{\infty}\bigr)\,\E\bigl\{\I_{G_l(q)}\,\bigl(\PPPP_{k-l}^V(r){\bf1}\bigr)(\tUUpsilon_l)\bigr\}\notag\\
	&\le \|f\|_{\infty}\exp\bigl(l\,\|V\|_{\infty}\bigr)\,\bigl\|\PPPP_k^V(r){\bf1}\bigr\|_{\infty}\,\IP\bigl(G_l(q)\bigr)\notag\\
	&\le C_1\|f\|_{\infty}\exp\bigl(l\,\|V\|_{\infty}-l\log q^{-1}\bigr)\,\bigl\|\PPPP_k^V(r){\bf1}\bigr\|_{\infty}	d_r(\UUpsilon,\UUpsilon'). \label{Ikl}
\end{align}
To estimate $J_k=J_k(\UUpsilon,\UUpsilon')$, we write
\begin{align}
	J_k
	&=\E\bigl\{\I_{G_k'(q)}\,\Xi_k(\UUpsilon)(f(\UUpsilon_k')-f(\UUpsilon_k))\bigr\}
	+\E\bigl\{\I_{G_k'(q)}(\Xi_k(\UUpsilon)-\Xi_k(\UUpsilon'))f(\UUpsilon_k)\bigr\}
	\notag\\
	&=:J_k^1(\UUpsilon,\UUpsilon')+J_k^2(\UUpsilon,\UUpsilon').
	\label{JkU}
\end{align}
Using the Lipschitz continuity of~$f$, we derive
\begin{align}
	J_k^1(\UUpsilon,\UUpsilon')
	&\le C_2\,q^k\|f\|_L\bigl\|\PPPP_k^V(r){\bf1}\bigr\|_{\infty}d_r(\UUpsilon,\UUpsilon').
	\label{Jk1U}
\end{align}
Furthermore, the Lipschitz continuity of~$V$ implies that, for $q\le1/2$, 
\begin{align*}
\bigl|\Xi_k(\UUpsilon)-\Xi_k(\UUpsilon')\bigr|
&= \Xi_k(\UUpsilon)\biggl\{\exp\biggl(\sum_{j=1}^k\bigl|V(\tUUpsilon_j)-V(\tUUpsilon_j'\bigr|\biggr)-1\biggr\}	\\
&\le \Xi_k(\UUpsilon)\Bigl\{\exp\bigl(2q\,\|V\|_L\, d_r(\UUpsilon,\UUpsilon')\bigr)-1\Bigr\}\\
&\le C_3(V)\, d_r(\UUpsilon,\UUpsilon')\,\Xi_k(\UUpsilon),
\end{align*}
on the set $G_k'(q)$. It follows that
\begin{align}
	J_k^2(\UUpsilon,\UUpsilon')
	&\le C_3(V)\,\|f\|_{\infty}\bigl\|\PPPP_k^V(r){\bf1}\bigr\|_{\infty}d_r(\UUpsilon,\UUpsilon'). 
	\label{Jk2U}
\end{align}
Combining this with~\eqref{PkVfU}--\eqref{Jk2U}, we derive 
\begin{multline*}
\bigl|\Delta_k(\UUpsilon,\UUpsilon')\bigr|
\le C_4(V)\|f\|_{L}\bigl\|\PPPP_k^V(r){\bf1}\bigr\|_{\infty}d_r(\UUpsilon,\UUpsilon')\sum_{l=0}^k\exp\bigl(l\,\|V\|_{\infty}-l\log q^{-1}\bigr).	
\end{multline*}
Taking $q<\exp(-\|V\|_{\infty})$, we arrive at~\eqref{lipschitz}.

\subsection{Proof of Proposition~\ref{p-UFP}}
\label{ss-UFP}
The fact the $Q$ is a $1$-Lipschitz convex function satisfying~\eqref{QVC} is well known, as is the relation $I(\lambda)=+\infty$ for $\lambda\in\MM(\frX)\setminus\PP(\frX)$. We thus confine ourselves to the proof of~\eqref{DV-ratefunction} for $\lambda\in\PP(\frX)$. 

\medskip
{\it Step~1}. Let us denote by~$J(\lambda)$ the supremum on the right-hand side of~\eqref{DV-ratefunction}. We first prove that 
\begin{equation} \label{8.71}
I(\lambda)\ge J(\lambda)\quad\mbox{for any $\lambda\in\PP(\frX)$}.
\end{equation}
To this end, fix $\lambda\in\PP(\frX)$ and~$\e>0$. Let $g\in C(\frX)$ be such that $g\ge1$ and 
\begin{equation} \label{8.72}
J(\lambda)< \int_\frX\log\frac{g}{\PPPP_1g}\,\dd\lambda+\e. 
\end{equation}
Set $V=\log\frac{g}{\PPPP_1g}$. A simple calculation based on the semigroup property and the inequality $1\le \PPPP_1g\le \|g\|_\infty$ shows that
$$
\|g\|_\infty^{-1}\le \PPPP_k^V{\mathbf1}\le\|g\|_\infty,
$$
and so  $Q(V)=0$. Inequality~\eqref{8.72} now implies that
$$
J(\lambda)<\langle V,\lambda\rangle+\e\le I(\lambda)+\e.
$$
Since~$\e>0$ is arbitrary, we arrive at~\eqref{8.71}. 

\smallskip
{\it Step~2}. To establish the opposite inequality in~\eqref{8.71}, we again  fix~$\e>0$. Let $V\in C(\frX)$ be such that
\begin{equation} \label{optimalV}
	I(\lambda)<\langle V,\lambda\rangle +\e,  \quad Q(V)=0.
\end{equation}
The existence of such a function follows from the definition of~$I$ and the relation~\eqref{QVC}. We now set 
$$
g_\e=e^V\sum_{k=0}^\infty e^{-\e k}\,\PPPP_k^V{\mathbf1}. 
$$
The second relation in~\eqref{optimalV} implies that the series converges uniformly in~$u\in\frX$ and defines a continuous function on~$\frX$. It is straightforward to check that
$$
\PPPP_1g_\e=\sum_{k=0}^\infty 
e^{-\e k}\,\PPPP_{k+1}^V{\mathbf1}=e^{\e}(e^{-V}g_\e-1).
$$
It follows that 
\begin{equation} \label{8.73}
\log\frac{g_\e}{\PPPP_1g_\e}\ge V-\e-\log\bigl(1-e^{V}g_\e^{-1}\bigr)\ge V-\e. 
\end{equation}
Integrating~\eqref{8.73} with respect to~$\lambda$ and using ~\eqref{optimalV}, we derive 
$$
\int_\frX\log\frac{g_\e}{\PPPP_1g_\e}\dd\lambda
\ge \langle V,\lambda\rangle-\e\ge I(\lambda)-2\e.
$$ 
Since~$\e>0$ is arbitrary, we arrive at the required inequality. 

\subsection{Proof of Proposition~\ref{p-RF}}
\label{ss-RF}
{\it Step~1: A formula for~$\III$}. 
We first note that $\PPPP_k^V(r)$ falls into the framework of Proposition~\ref{p-UFP} if we define the transition function by
$$
P_1^r(\xxx^r,\dd\yyy^r)=\delta_{[x_2,\dots,x_r]}(\dd y_1,\dots,\dd y_{r-1})P_1(x_r,\dd y_r). 
$$
Therefore, replacing $g$ by~$e^V$ in~\eqref{DV-ratefunction} and using approximation of a bounded measurable function by continuous functions, we can write
$$
\III^r(\lambda)=\sup_{V\ge0}\bigl\langle V-\log\bigl(\PPPP_1(r)e^V\bigr),\lambda\bigr\rangle,\quad \lambda\in\PP(\frX(r)),
$$
where the supremum is taken over all non-negative bounded measurable functions $V:\frX(r)\to\R$. Combining this with~\eqref{ratefunction-I}, we derive
 \begin{equation} \label{8.65}
\III(\llambda)=\sup_{r\ge1}\,\sup_{V\ge0}
\int_{\frX(r)} \Bigl(V(\xxx^r)-\log \int_{\frX}e^{V(x_2,\dots,x_{r},y)}
P_1(x_r,\dd y)\Bigr)\llambda^r(\dd \xxx^r),
\end{equation}
where $\llambda^r$ stands for the image of~$\llambda$ under the projection~$\Pi^r$ to the first~$r$ components. Since~$\llambda$ is shift-invariant, we can replace $[x_2,\dots,x_{r},y]$ by $[x_1,\dots,x_{r-1},y]$ in the integral over~${\frX}$. Let us denote by~$\llambda^r(\xxx^{r-1};\,\cdot\,)$ the regular conditional probability of~$\llambda^r$ given the first $r-1$ coordinates and let 
$$
F_V(\xxx^{r-1})=\int_{\frX}V(\xxx^{r-1},y)\llambda^r(\xxx^{r-1};\dd y)
-\log \int_{\frX}e^{V(\xxx^{r-1},y)}P_1(x_{r-1},\dd y).
$$
We can  rewrite~\eqref{8.65} as
\begin{equation} \label{8.66}
\III(\llambda)=\sup_{r\ge1}\,\sup_{V\ge0}
\int_{\frX(r-1)} F_V(\xxx^{r-1})\,\llambda^{r-1}(\dd \xxx^{r-1}).
\end{equation}
Denoting by~$\JJJ(\llambda)$ the expression on the right-hand side of~\eqref{entropyformula}, we now prove that~$\III$ is bounded from above and from below by~$\JJJ$.  

\smallskip
{\it Step~2: Upper bound\/}. We recall the convention that any shift-invariant measure $\llambda\in\PP(\XXX)$ can be extended (in a unique manner) to~$\frX^\Z$. For any integer $r\ge1$, we write $\Z_r=\Z\cap(-\infty,r]$ and, given a measure $\llambda\in\PP_s(\bfrX)$, denote by $\llambda(\zzz,\xxx^{r-1};\,\cdot\,)$ the regular conditional probability of the projection of~$\llambda$ to~$\frX^{\Z_r}$ given $[\zzz,\xxx^{r-1}]\in\frX^{\Z_{r-1}}$. Let 
$$
{\widetilde F}_V(\zzz,\xxx^{r-1})
=\int_{\frX} V(\xxx^{r-1},y)\llambda(\zzz,\xxx^{r-1};\dd y)
-\log \int_{\frX}e^{V(\xxx^{r-1},y)}P_1(x_{r-1},\dd y).
$$
It is straightforward to check that
\begin{equation} \label{8.67}
\int_{{\frX(r-1)}} F_V(\xxx^{r-1})\,\llambda^{r-1}(\dd \xxx^{r-1})
=\int_{\bfrX_-}\int_{{\frX(r-1)}}{\widetilde F}_V(\zzz,\xxx^{r-1})
\llambda(\dd \zzz,\dd \xxx^{r-1}). 
\end{equation}
The definition of the relative entropy implies that
$$
{\widetilde F}_V(\zzz,\xxx^{r-1})
\le \Ent\bigl(\llambda(\zzz,\xxx^{r-1};\,\cdot\,)\,|\,P_1(x_{r-1},\cdot)\bigr)
\quad\mbox{for any $V\in C_b(\frX(r))$}.
$$
Substituting this and~\eqref{8.67} into~\eqref{8.66}, we obtain
$$
\III(\llambda)\le \sup_{r\ge1}\int_{\bfrX_-}\int_{{\frX(r-1)}}
\Ent\bigl(\llambda(\zzz,\xxx^{r-1};\cdot)\,|\,P_1(x_{r-1},\cdot)\bigr)
\llambda(\dd \zzz,\dd \xxx^{r-1}).
$$
In view of the stationarity of~$\llambda$, the expression under the supremum on the right-hand side of this inequality does not depend on~$r$ and coincides with~$\JJJ(\llambda)$. 

\smallskip
{\it Step~3: Lower bound}. 
For any integer $r\ge0$, we define the space~$\bfrX_r=\frX^{[\![-r,0]\!]}$ and denote by $\xxx_r=[x_{-r},\dots,x_0]$ its points. To prove that $\III\ge\JJJ$, we first rewrite~\eqref{8.66} in the form
\begin{equation} \label{8.69}
\III(\llambda)=\sup_{r\ge0}\sup_{V\ge0}
\int_{\bfrX_r}F_V(\xxx_r)\llambda_r(\dd \xxx_r).
\end{equation}
Here, bounded measurable functions~$V$ depend on $r+2$ variables, and, with a slight abuse of notation, we write
\begin{equation} \label{8.070}
F_V(\xxx_r)=\int_{\frX}V(\xxx_r,y)\llambda^{[\![-r,1]\!]}(\xxx_r;\dd y)
-\log \int_{\frX}e^{V(\xxx_r,y)}P_1(x_0,\dd y),
\end{equation}
where $\llambda^{[\![-r,1]\!]}(\xxx_r;\,\cdot\,)$ denotes the regular conditional probability for the projection of~$\llambda$ to~$\frX^{[\![-r,1]\!]}$ given $\xxx_r\in\bfrX_r$. In what follows, it will be convenient to consider $\llambda^{[\![-r,1]\!]}(\xxx_r;\,\cdot\,)$ as a function of the entire trajectory $\xxx\in\bfrX_-$ depending only on~$\xxx_r$, and accordingly we shall write $\llambda^{[\![-r,1]\!]}(\xxx;\,\cdot\,)$. Let $\Pi_r:\bfrX_-\to \bfrX_r$ be the projection taking~$\xxx$ to~$\xxx_r$ and let $\{\FF_r\}_{r\ge0}$ be the filtration on~$\bfrX_-$ generated by the projections $\Pi_s$, $0\le s\le r$. We claim that, for any bounded measurable function $f:\frX\to\R$, the sequence $\{\langle f,\llambda^{[\![-r,1]\!]}(\xxx,\cdot)\rangle\}_{r\ge0}$ considered on the probability space $(\bfrX_-,\llambda_-)$ is a martingale with respect to the filtration~$\{\FF_r\}$. To see this, let us consider the probability space $(\frX^\Z,\llambda)$ and the bounded random variable $\xi([x_j]_{j\in\Z})=f(x_1)$ on it. By the definition of the regular conditional probability, we have 
$$
\langle f,\llambda^{[\![-r,1]\!]}(\xxx,\cdot)\rangle=\E^{\llambda_-}(\xi\,|\,\FF_{r}),
$$ 
where $\E^{\llambda_-}(\cdot\,|\,\cdot)$ denotes the conditional expectation with respect to~$\llambda_-$. This relation immediately implies the required martingale property. 

Applying Doob's martingale convergence theorem, we see that, for any bounded measurable function $f:\frX\to\R$, the sequence $\langle f,\llambda^{[\![-r,1]\!]}(\xxx,\cdot)\rangle$ converges for $\llambda_-$-almost every $\xxx\in\bfrX_-$. By Theorem~A.5.2 in~\cite{KS-book}, there is a random probability measure $\mu(\xxx,\,\cdot\,)$ such that 
\begin{equation} \label{8.072}
\llambda^{[\![-r,1]\!]}(\xxx,\,\cdot\,)\rightharpoonup\mu(\xxx,\,\cdot\,)
\quad\mbox{for $\llambda_-$-a.e.~$\xxx\in \bfrX_-$}. 
\end{equation}
It is straightforward to check that $\mu(\xxx,\,\cdot\,)$ is the projection to the first component of the regular conditional probability with respect to~$\llambda$ given $\xxx\in\bfrX_-$. By uniqueness, it must coincide with~$\lambda(\xxx,\,\cdot\,)$ for $\llambda_-$-almost every $\xxx\in\bfrX_-$. 

\smallskip
We now recall that 
\begin{equation} \label{RE-formula}
	E_r(\xxx_r):=\Ent\bigl(\llambda^{[\![-r,1]\!]}(\xxx_r,\,\cdot\,)\,|\,P_1(x_0,\,\cdot\,)\bigr)=\sup_{V\ge0}F_V(\xxx_r),
\end{equation}
where the supremum is taken over all non-negative bounded measurable functions $V:\bfrX_r\times\frX\to\R$. Moreover, the supremum in~\eqref{RE-formula} is saturated by the sequence of functions 
$$
V_N(\xxx_r,y)=\biggl\{\biggl(\log\frac{\dd\llambda^{[\![-r,1]\!]}(\xxx_r,\cdot)}{\dd P_1(x_0,\,\cdot\,)}\biggr)\wedge N\biggr\}\vee (-N)+N. 
$$
Therefore, for any $\e>0$ and $\xxx_r\in\bfrX_r$, we can find an integer $N=N_\e(\xxx_r)\ge1$ such that  
\begin{equation} \label{RE-approximation}
F_{V_{N}}(\xxx_r)\ge
\left\{
\begin{array}{cl}
	E_r(\xxx_r)-\e\quad&\mbox{if $E_r(\xxx_r)<\infty$},\\[3pt]
	\e^{-1}\quad&\mbox{if $E_r(\xxx_r)=\infty$}.
\end{array} 
\right.
\end{equation}
If $E_r(\xxx_r)=+\infty$ for some $r\ge0$ on a set of positive $\llambda_r$-measure, then both~$\III(\llambda)$ and~$\JJJ(\llambda)$ are equal to~$+\infty$. In the opposite case, combining~\eqref{8.69}, \eqref{RE-formula}, and~\eqref{RE-approximation}, we see that 
\begin{equation} \label{8.073}
\III(\llambda)\ge \sup_{r\ge0}\int_{\bfrX_-}
\Ent\bigl(\llambda^{[\![-r,1]\!]}(\xxx,\cdot)\,|P_1(x_0,\,\cdot\,)\bigr)\llambda_-(\dd\xxx)-\e. 
\end{equation}
Since the relative entropy is  lower-semicontinuous and non-negative, using~\eqref{8.072} with $\mu(\xxx,\cdot)=\lambda(\xxx,\cdot)$ and  applying Fatou's lemma, we obtain
\begin{multline*}
\liminf_{r\to\infty}\int_{\bfrX_-}
\Ent\bigl(\llambda^{[\![-r,1]\!]}(\xxx,\cdot)\,|P_1(x_0,\,\cdot\,)\bigr)\,\llambda_-(\dd\xxx)\\
\ge \int_{\bfrX_-}
\Ent\bigl(\lambda(\xxx,\,\cdot)\,|\,P_1(x_0,\,\cdot\,)\bigr)\llambda_-(\dd\xxx). 
\end{multline*}
Combining this with~\eqref{8.073} and recalling that~$\e>0$ was arbitrary, we derive the required inequality $\III(\llambda)\ge\JJJ(\llambda)$. This completes the proof of Proposition~\ref{p-RF}.

\section{Image of measures under non-degenerate maps}
\label{s-measureimage}
In this section, we discuss some general results on the absolute continuity of the image of measures under finite-dimensional smooth maps. This type of properties are well known, and a comprehensive presentation can be found in~\cite{bogachev2010}. Here we only need a sufficient condition for the existence, regularity, positivity and Lipschitz dependence on the parameter of the density for the image measure. For the reader's convenience, we give a proof of the results we need (Theorems~\ref{T:2.7} and~\ref{T:2.8}) in Appendix~\ref{s-proofs-measures}.

\smallskip
Let~$\EEEE$ be a separable Hilbert space, let $\XXXX$ be a compact subset in a separable Hilbert\footnote{The reader not willing to deal with infinite-dimensional manifolds may assume that~$\HHHH$ has the same structure as in Section~\ref{s-ldp-control}.} manifold~$\HHHH$ with a metric~$d$ (see Section~II.1 in~\cite{lang1985}), and let~$\YYYY$ be a  Riemannian manifold without boundary. We consider a function $F:\XXXX\times\EEEE\to\YYYY$ satisfying the following  hypothesis for an integer $k\ge0$. 

\begin{itemize}
	\item [\hypertarget{(F)}{\bf(F)}]
	\sl For any  $\varUpsilon\in\XXXX$, the mapping $\eta\mapsto F(\varUpsilon,\eta)$ is~$(k+1)$-times continuously differentiable, and the derivative $\p_\eta^{k+1}F$ is a Lipschitz-continuous function  of~$(\varUpsilon,\eta)$ on bounded subsets of~$\XXXX\times\EEEE$.
\end{itemize}

Our aim is to study the image of signed measures on~$\EEEE$ under maps with the above property. Namely, given $\varUpsilon\in \XXXX$ and $\ell\in\MM(\EEEE)$, we denote by $\lambda_\varUpsilon=F_*(\varUpsilon,\ell)$ the image of~$\ell$ under the map $\eta\mapsto F(\varUpsilon,\eta)$. We impose  the following hypothesis on~$\ell$. 

\begin{itemize}
	\item[\hypertarget{(P)}{\bf(P)}] 
	\sl The support of~$\ell$ is a compact subset of~$\EEEE$, and there is an orthonormal basis~$\{\varphi_j\}$ in~$\EEEE$ such that~$\ell$ can be written as the product of its one-dimensional projections~$\ell_j$ onto the vector spaces~$\EEEE_j$ spanned by~$\varphi_j$. Moreover, for any $j\ge1$, the measure~$\ell_j$ has a density $\rho_j\in C^k(\EEEE_j)$. 
\end{itemize}
 
\begin{theorem}\label{T:2.7}
Let us assume that a function~$F$ and a measure $\ell\in \PP(\EEEE)$ satisfy Hypotheses~\hyperlink{(F)}{\rm(F)} and~\hyperlink{(P)}{\rm(P)} with some integer $k\ge0$, and for any~$\varUpsilon\in \XXXX$ and $\eta\in\supp\ell$  we have 
\begin{equation} \label{2.7}
	\Image\bigl(\p_\eta F(\varUpsilon,\eta)\bigr)=T_y\YYYY, 
\end{equation}
where $T_y\YYYY$ stands for the tangent space of~$\YYYY$ at~$y=F(\varUpsilon,\eta)$. Then, for any $\varUpsilon\in\XXXX$, the measure $\lambda_\varUpsilon$ is absolutely continuous with respect to the volume measure on~$\YYYY$, and the corresponding density $\rho(\varUpsilon,y)$, defined for $\varUpsilon\in\XXXX$ and $y\in\YYYY$, is continuous in~$(\varUpsilon,y)$ and $C^k$-smooth in~$y$. Moreover, there is $C>0$ such that 
\begin{equation} \label{2.8b}
\|\rho(\varUpsilon_1,\cdot)-\rho(\varUpsilon_2,\cdot)\|_{C^k(\YYYY)} 
\le C\,d_\XXXX(\varUpsilon_1,\varUpsilon_2)
\quad \mbox{for all $\varUpsilon_1,\varUpsilon_2\in\XXXX$}.
\end{equation}	
\end{theorem}

\begin{theorem}\label{T:2.8} 
Suppose that the hypotheses of Theorem~\ref{T:2.7} are satisfied with some integer $k\ge0$, and let points  $\widehat \varUpsilon\in \XXXX$,  $\hat \eta\in \EEEE$,  $\hat y\in\YYYY$ be such that $F(\widehat\varUpsilon,\hat\eta)=\hat y$ and 
\begin{gather} \label{2.14}
\rho_j(\hat \eta_j)>0\quad\text{for all $j\ge1$}, 	
\end{gather}
where $\hat\eta_j=\lag \hat\eta,\varphi_j\rag_\EEEE$. Then $\rho(\widehat\varUpsilon,\hat y)>0$.
\end{theorem}

\section{Application to the 2D Navier--Stokes system with a particle}
\label{s-NS}

\subsection{Large Deviation Principle}
\label{ss-ldpproof}
In this section we prove Theorem~\ref{t-ldp}. To this end, we shall make use of  Theorem~\ref{t-uniformLDP}. According to that result, it suffices to check the validity of Hypotheses~\hyperlink{(R)}{(R)}, \hyperlink{(AC)}{(AC)}, \hyperlink{(ACL)}{(ACL)}, and~\hyperlink{(D)}{(D)}.  Let $s\ge3$ be an integer, set  $\HH=V^s$, $\YYYY=\T^2$, $\EEEE=L^2(J,V^s)$, where $J=[0,1]$, and denote by~$\aA$ and~$\XXXX$ the sets~$\aA^s$ and~$\aA^s\times\T^2$, respectively. It is straightforward to see that any point $\varUpsilon^0=(0,p)\in\XXXX^s$ satisfies the relation $S(\varUpsilon^0,0)=\varUpsilon^0$. Moreover, in view of the regularising property of the Navier--Stokes system and infinite differentiability of its resolving operator with respect to the initial condition and the right-hand side (see Chapters~I and~VII in~\cite{BV1992}), the regularity property~\hyperlink{(R)}{(R)} is satisfied with~$\VV=V^{s+1}$. Furthermore, it is easy to see that~\hyperlink{(D)}{(D)} follows immediately from~\hyperlink{(N)}{(N)}. It remains to verify the validity of~\hyperlink{(AC)}{(AC)} and~\hyperlink{(ACL)}{(ACL)}. 

\subsubsection*{Approximate controllability~\hyperlink{(AC)}{(AC)}}
We need to prove that any initial state $\varUpsilon=(u_0,p)\in\XXXX^s$ can be steered to the small neighbourhood of any target state $\widehat\varUpsilon=(\hat u,\hat p)\in\XXXX^s$ with the help of controls belonging to the support~$\KK^s$ of the law $\ell=\DD(\eta_k)$. This will be done in three steps. 

\smallskip
{\it Step~1: Reduction to $\hat u=0$}.
Suppose we can prove that, for an arbitrary $r>0$ and $p_1\in\T^2$, any point $\varUpsilon=(u_0,p)\in\XXXX^s$ can be steered to the $r$-neighbourhood of $(0,p_1)$ at some time $n_1\ge1$ depending only on~$r$. Fix $\e>0$ and points $\varUpsilon, \widehat\varUpsilon\in\XXXX^s$. Recalling the definition of the sets~$\aA_k^s$ and using the fact that they form an increasing sequence, we can find an integer $n_2\ge1$, depending only on~$\e>0$, an initial point $\varUpsilon_1=(0,p_1)\in\XXXX^s$, and controls $\zeta_1^2,\dots,\zeta_{n_2}^2\in\KK^s$, such that 
$$
d_{\XXXX^s}\bigl(S_{n_2}(\varUpsilon_1;\zeta_1^2,\dots,\zeta_{n_2}^2),\widehat\varUpsilon\bigr)\le\e/2.
$$ 
The continuity of~$S$ implies the existence of a number $r>0$ such that 
\begin{equation} \label{dXSn2}
	d_{\XXXX^s}\bigl(S_{n_2}(\varUpsilon_1';\zeta_1^2,\dots,\zeta_{n_2}^2),\widehat\varUpsilon\bigr)\le\e,
\end{equation}
for $d_{\XXXX^s}(\varUpsilon_1',\varUpsilon_1)\le r$. By assumption, there are controls $\zeta_1^1,\dots,\zeta_{n_1}^1\in\KK^s$ such that 
$$
d_{\XXXX^s}\bigl(S_{n_1}(\varUpsilon;\zeta_1^1,\dots,\zeta_{n_1}^1),\varUpsilon_1\bigr)\le r. 
$$
Combining this with inequality~\eqref{dXSn2} in which $\varUpsilon_1'=S_{n_1}(\varUpsilon;\zeta_1^1,\dots,\zeta_{n_1}^1)$ and setting $n=n_1+n_2$, we see that the controls $(\zeta_1,\dots,\zeta_n)=(\zeta_1^1,\dots,\zeta_{n_1}^1,\zeta_1^2,\dots,\zeta_{n_2}^2)$ are such that~\eqref{AC} holds.  

\smallskip
{\it Step~2: Reduction to $u_0=0$}. Suppose  that, given arbitrary $\delta>0$ and $\hat p\in\T^2$, we can steer  any point $\varUpsilon_1=(0,p_1)\in\XXXX^s$ to the $\delta$-neighbourhood of~$(0,\hat p)$ at some time $n_2\ge1$ depending only on~$\delta$. Fix $\e>0$ and $\varUpsilon=(u_0,p)$, $\widehat\varUpsilon=(0,\hat p)$  in~$\XXXX^s$. By continuity of~$S$, we can find a number $r>0$ such that, for any point $\varUpsilon_1'=(u_1,p_1)\in\XXXX^s$ satisfying the inequality $d_{\XXXX^s}(\varUpsilon_1',\varUpsilon_1)\le r$ with $\varUpsilon_1=(0,p_1)$, there are controls $\zeta_1^2,\dots,\zeta_{n_2}^2\in\KK^s$ for which~\eqref{dXSn2} holds. In view of the dissipativity of the homogeneous Navier--Stokes system in the space~$V^s$, we can find an integer $n_1\ge1$ depending only on~$r$ such that 
$$
d_{\XXXX^s}\bigl(S_{n_1}(\varUpsilon;0,\dots,0),\varUpsilon_1\bigr)\le r,
$$
where $\varUpsilon_1=(0,p_1)$ with some $p_1\in\T^2$. Combining this with~\eqref{dXSn2}, we see that~\eqref{AC} holds for  the controls $(\zeta_1,\dots,\zeta_n)=(0,\dots,0,\zeta_1^2,\dots,\zeta_{n_2}^2)$, where $n=n_1+n_2$. 

\smallskip
{\it Step~3: Proof in the case $u_0=\hat u=0$}.
It suffices to prove that, for any points $p,\hat p\in\T^2$ whose distance from each other is less than a fixed number $\varkappa>0$, we can find $\zeta\in\KK^s$ satisfying $S(\varUpsilon,\zeta)=\widehat\varUpsilon$, where $\varUpsilon=(0,p)$ and $\widehat\varUpsilon=(0,\hat p)$. Indeed, suppose this property is established and fix any points $\varUpsilon,\widehat\varUpsilon\in\XXXX^s$ with zero $u$-component. We can find an integer $n\ge1$, depending only on~$\varkappa$, and points $\varUpsilon_k=(0,p_k)$, $k=1,\dots,n-1$, such that $|p_k-p_{k-1}|\le\varkappa$, where $p_0=p$ and $p_n=\hat p$. Applying the above-mentioned exact controllability result, we can find $\zeta_1,\dots,\zeta_n\in\KK^s$ such that $S(\varUpsilon_{k-1},\zeta_k)=\varUpsilon_k$ for $1\le k\le n$. This implies that $S_n(\varUpsilon;\zeta_1,\dots,\zeta_n)=\widehat\varUpsilon$. 

\smallskip
We now use the argument described in Section~\ref{ss-scheme} (see the scheme of the proof of Theorem~\ref{t-ldp}) to establish the exact controllability for~$p$ and~$\hat p$ that are close enough. Let us define an $\XXXX^s$-valued curve $\varUpsilon(t)=(u(t,x),y(t))$ by relations~\eqref{uy-control}, in which $U_1,U_2,\gamma,\varphi_1,\varphi_2$ are as in~\eqref{U1U2gamma} and~\eqref{gammadot}. The endpoints of~$\{\varUpsilon(t),t\in J\}$ coincide with~$\varUpsilon$ and~$\widehat\varUpsilon$, and Eqs.~\eqref{NS}, \eqref{ode-fluidparticle} hold with a right-hand side~$\eta$ given by~\eqref{eta-control}. We need to prove that~$\Pi g\in\KK^s$. 

To establish  this,  we write~$\Pi g$ in the form~\eqref{gtx} and note that 
\begin{equation} \label{Crnormalpha}
	\|\alpha_j\|_{C^r(J)}\le M_r|p-\hat p|,
\end{equation}
where the numbers~$M_r$ do not depend~$p$ and~$\hat p$. Expanding~$\alpha_j$ in the orthonormal basis~$\{\psi_l\}$,
$$
\alpha_j(t)=\sum_{l\ge1}\alpha_{jl}\psi_l(t), 
$$
we can rewrite~\eqref{gtx} in the form
\begin{equation} \label{gtx-expanded}
	(\Pi g)(t,x)=\sum_{|j|_1\le 2}\sum_{l\ge1}b_j\alpha_{jl}\psi_l(t)e_j(x),
\end{equation}
where $|j|_1=|j_1|+|j_2|$. In view of~\eqref{Crnormalpha} and Poincar\'e's inequality~\eqref{poincareineq}, we have 
\begin{equation} \label{alphaljestimate}
	|\alpha_{jl}|\le C_rM_r |p-\hat p|\,l^{-\theta r}\quad\mbox{for $|j|_1\le2$,  $l,r\ge1$.}
\end{equation}
On the other hand, it follows from~\hyperlink{(N)}{(N)} that if $\e>0$ is sufficiently small, then any function~$h$ of the form~\eqref{functioninKK} (where $\Lambda=\{|j|_1\le2\}$) belongs to the support~$\KK^s$, provided that the coefficients $h_{lj}$ satisfy the inequality
\begin{equation} \label{hjl}
	|h_{lj}|\le \e l^{-\beta}\quad \mbox{for $|j|_1\le2$, $l\ge1$}. 
\end{equation}
Choosing $r\ge1$ so large that $\theta r>\beta$ and assuming that $C_rM_r|p-\hat p|\le\e$, we derive from~\eqref{alphaljestimate} and~\eqref{hjl} that~$\Pi g\in\KK^s$. This completes the proof of~\hyperlink{(AC)}{(AC)}. 

\subsubsection*{Approximate controllability of the linearisation~\hyperlink{(ACL)}{(ACL)}}
We need to prove the density of the image for the linear operator 
$$
(D_\eta S)(\varUpsilon,\eta):L^2(J,V^s)\to V^s\times T_{y(1)}\T^2,
$$ 
where $y(1)=S^y(\varUpsilon,\eta)$. In what follows, we  identify the tangent space~$T_y\T^2$ with~$\R^2$. 

We first note that, for any $\zeta\in\EEEE$, the vector function $(D_\eta S)(\varUpsilon,\eta)\zeta=(v,z)$ is a solution of the equations 
\begin{align}
	\p_tv+Lv+Q(u)v&=\zeta, \label{NS-linearised}\\
	\dot z-v\bigl(t,y(t)\bigr)-(D_xu)\bigl(t,y(t)\bigr)z(t)&=0,
	\label{ode-linearised}
\end{align}
where $L=-\nu\Pi\Delta$, $Q(u)v=\Pi(\langle u,\nabla\rangle v+\langle v,\nabla\rangle u)$, and $(u,y)=\bfS(\varUpsilon,\eta)$. These equations are supplemented with the initial conditions
\begin{equation} \label{IClinearised}
	v(0)=0, \quad z(0)=0. 
\end{equation}
Since $s\ge3$, it is easy to check the regularity of~$u$ and~$v$ is sufficient to ensure the well-posedness of~\eqref{NS-linearised}, \eqref{ode-linearised}, \eqref{IClinearised}. For given $\hat v\in V^s$, $\hat q\in\R^2$, and $\e>0$, we need to find $\zeta\in L^2(J,V^s)$ such that 
\begin{equation} \label{eclose}
	\|v(1)-\hat v\|_s<\e, \quad |z(1)-\hat q|<\e. 
\end{equation}
There is no loss of generality in assuming that $\hat v\in C^\infty(\T^2)$. Let $\gamma:J\to\R^2$ be defined by the relation 
\begin{equation} \label{gammat}
	\gamma(t)=\alpha(t)\hat q+\beta(t)\bigl(\hat v(y(1))+(D_xu)(1,y(1))\,\hat q\bigr),
\end{equation}
where $\alpha,\beta\in C^\infty(J)$ are such that 
$$
\mbox{$\alpha(t)=\beta(t)=0$ for $t\le1/3$},
\quad\mbox{$\alpha(t)=1$ for $t\ge2/3$}, \quad \beta(1)=0, \quad \dot\beta(1)=1.
$$ 
Writing $(\varphi_1(t), \varphi_2(t)):=\dot\gamma(t)-(D_xu)(t,y(t))\gamma(t)$ and choosing a small parameter $\delta>0$, we set 
$$
v_\delta(t,x)=\theta_\delta(t)\bigl(\varphi_1(t)U_1(x-y(t))+\varphi_2(t)U_2(x-y(t))\bigr)+\bigl(1-\theta_\delta(t)\bigr)\hat v,
$$
where the functions~$U_i$ are defined in~\eqref{U1U2gamma}, and~$\theta_\delta\in C^\infty(\R)$ is such that $0\le \theta_\delta\le1$, $\theta_\delta(t)=1$ for $t\le 1-\delta$ and $\theta_\delta(t)=0$ for $t\ge1$. Finally, we denote by~$z(t)$ the solution of~\eqref{ode-linearised} with zero initial condition. Then the vector function $(v,z)$, defined on~$J$, belongs to $\XX_s\times C^1(J,\R^2)$, vanishes at $t=0$, and satisfies~\eqref{NS-linearised}, \eqref{ode-linearised} on~$J$ with  $\zeta\in C(J,V^s)$. Moreover, we have $v(1)=\hat v$. Therefore, to complete the proof, it remains to show that the second inequality in~\eqref{eclose} holds for an appropriate choice of~$\delta>0$. 

To establish  this, we first note that there is a number~$R_1$ not depending on~$\delta$ such that 
\begin{equation} \label{supvDu}
\sup_{(t,x)\in J\times\T^2}\Bigl(\bigl|v_\delta(t,x)\bigr|+\bigl|(D_xu)(t,x)\bigr|\Bigr)\le R_1. 
\end{equation}
Combining this with~\eqref{ode-linearised} and applying Gronwall's inequality, we can find $R_2$ such that 
\begin{equation} \label{supzdelta}
	\sup_{t\in J}|\dot z_\delta(t)|\le R_2. 
\end{equation}
The choice of~$v_\delta$ implies that the function $\gamma(t)$, vanishing at $t=0$, is a solution of~\eqref{ode-linearised} on the interval~$[0,1-\delta]$. By uniqueness of solutions for ODEs, we conclude that $z_\delta(1-\delta)=\gamma(1-\delta)$. Therefore, using~\eqref{supzdelta}, we derive
$$
|z_\delta(1)-\hat q|
\le |z_\delta(1)-z_\delta(1-\delta)| 
+ |z_\delta(1-\delta)-\hat q|
\le R_2\delta+|\gamma(1-\delta)-\hat q|. 
$$
Since both  right-most part  terms of these inequalities go to zero as $\delta\to0$, we conclude that $|z_\delta(1)-\hat q|<\e$ for $\delta\ll1$. This completes the proof of~\hyperlink{(ACL)}{(ACL)} and of Theorem~\ref{t-ldp}. 

\begin{remark} \label{r-exactcontrol}
	The approximate controllability property established above shows that the $y$-component of the linearised problem is exactly controllable. Namely, for any $\hat q\in \R^2$ there is $\zeta\in L^2(J,V^s)$ such that $(D_\eta S^y)(\varUpsilon, \eta)\zeta=\hat q$. This follows immediately from the approximate controllability and the fact that the image of the derivative is a linear subspace in the finite-dimensional vector space~$T_y\T^2$.  
\end{remark}

\subsection{Regularity of laws and convergence}
\label{ss-regularitylaw}
In this section we prove Theorem~\ref{t-particle}. The proof is divided into three steps. 

\smallskip
{\it Step 1: Reduction\/}. Let us denote by~$\XXXX^\infty$ the intersection of the sets~$\XXXX^s$ with $s\ge1$. Consider the following two statements: 
\begin{itemize}
	\item [\hypertarget{(i)}{\bf(i)}]
For any $\varUpsilon\in\XXXX^\infty$, the measure~$\llambda_t^\varUpsilon$ has a density~$\rho_t^\varUpsilon\in C^\infty(\T^{2t})$.
	\item [\hypertarget{(ii)}{\bf(ii)}]
For any $k\ge1$ there is an integer $s\ge3$ such that the mapping $\varUpsilon\mapsto\rho_t^\varUpsilon$ acting from~$\XXXX^\infty$ to~$C^k(\T^{2t})$ is Lipschitz continuous with respect to the norm of~$\XXXX^s$. 
\end{itemize}
Assuming that these  two statements hold, the regularity part of Theorem~\ref{t-particle} is deduced as follows. The regularising property of the Navier--Stokes system implies that, for any $\varUpsilon\in\XXXX^3$, the function $S(\varUpsilon,\eta)$ belongs to~$\XXXX^\infty$ with probability~$1$. Recall the definition of the map~$F^t$ in Section~\ref{ss-scheme} (see the scheme of the proof of Theorem~\ref{t-particle}), and note that~$\llambda_t^\varUpsilon$ coincides with the image of the $t$-fold product of~$\ell$ under the map~$F^t(\varUpsilon,\cdot)$. The independence of the random variables $\eta_1,\dots,\eta_t$ implies that
\begin{equation}
	\llambda_{[\![2,t]\!]}^\varUpsilon
	=\int_{\XXXX^\infty}F_*^{t-1}(\upsilon,\underbrace{\ell\otimes\cdots\otimes\ell}_{\mbox{\footnotesize $t-1$ times}})\,\MMMM_1^\varUpsilon(\dd\upsilon)
	=\int_{\XXXX^\infty}\llambda_{t-1}^\upsilon\,\MMMM_1^\varUpsilon(\dd\upsilon).
\end{equation}
Combining this with the property~\hyperlink{(i)}{(i)}, we see that~$\llambda_{[\![2,t]\!]}^\varUpsilon$ has a density given by 
\begin{equation} \label{rhoupsilon2t}
	\rho_{[\![2,t]\!]}^\varUpsilon(y_1,\dots,y_{t-1})=\int_{\XXXX^\infty}\rho_{t-1}^\upsilon(y_1,\dots,y_{t-1})\,\MMMM_1^\varUpsilon(\dd\upsilon). 
\end{equation}
Furthermore, for any integer $s\ge3$, the map $S(\varUpsilon,\eta)$ is Lipschitz continuous on~$\XXXX^3\times \EEEE$ with range in~$\XXXX^s$, where $\EEEE=L^2(J,V^s)$. It follows that the same holds for the map $\varUpsilon\mapsto\MMMM_1^\varUpsilon$ acting from~$\XXXX^3$ to the space $\PP(\XXXX^s)$ endowed with the dual-Lipschitz metric. Denoting by~$C_s$ the corresponding Lipschitz constant and using the property~\hyperlink{(ii)}{(ii)} and relation~\eqref{rhoupsilon2t}, for any $\varUpsilon_1,\varUpsilon_2\in\XXXX^3$ we can write
\begin{align*}
	\bigl\|\rho_{[\![2,t]\!]}^{\varUpsilon_1}-\rho_{[\![2,t]\!]}^{\varUpsilon_2}\bigr\|_{C^k}
	&=\sup_{\alpha,y}\, \biggl|\int_{\XXXX^s}
	(\p^\alpha\rho_{t-1}^{\upsilon})(y)\bigl(\MMMM_1^{\varUpsilon_1}-\MMMM_1^{\varUpsilon_2}\bigr)
	(\dd\upsilon)\biggr| \\
	&\le C_s\bigl\|\rho_{t-1}^\cdot \bigr\|_{L_b(\XXXX^s,C^k)}d_{\XXXX^3}(\varUpsilon_1,\varUpsilon_2),
\end{align*}
where the supremum is taken over all $y\in\T^{2(t-1)}$ and $\alpha\in\Z_+^2$ with $|\alpha|\le k$. We have thus established the Lipschitz continuity of the function $\varUpsilon\mapsto\rho_{[\![2,t]\!]}^\varUpsilon$ from~$\XXXX^3$ to $C^k$ for any $k\ge1$. 

To prove that~$\llambda_t$ possesses a regular density, we note that the stationary measure~$\MMMM$ is concentrated on~$\XXXX^\infty$. It follows that (cf.~\eqref{rhoupsilon2t})
\begin{equation*} 
	\llambda_t=\int_{\XXXX^\infty}\llambda_t^\upsilon\,\MMMM(\dd\upsilon). 
\end{equation*}
In view of~\hyperlink{(i)}{(i)}, this implies the existence of a density given by
\begin{equation} \label{stationarymeasure-y}
	\rho_t(y_1,\dots,y_t)=\int_{\XXXX^\infty}\rho_t^\upsilon(y_1,\dots,y_t)\,\MMMM(\dd\upsilon). 
\end{equation}
The infinite smoothness of~$\rho_t$ follows now from the property~\hyperlink{(ii)}{(ii)}.

\smallskip
{\it Step 2: Proof of~\hyperlink{(i)}{\rm(i)} and~\hyperlink{(ii)}{\rm(ii)}\/}. Fix an integer $k\ge1$, set $s=k+3$, and recall that~$\llambda_t^\varUpsilon$ is the image of the $t$-fold product $\ell^t=\ell\otimes\cdots\otimes\ell$ under the map
$$
F^t:\XXXX^s\times \underbrace{L^2(J,V^s)\times\cdots\times L^2(J,V^s)}_{\mbox{\footnotesize $t$ times}}\to\T^{2t}, \quad
(\varUpsilon,\eta)\mapsto (y_1,\dots,y_t). 
$$
We claim that the hypotheses of Theorem~\ref{T:2.7} are satisfied for~$F^t$. If this  holds, then the statements~\hyperlink{(i)}{(i)} and~\hyperlink{(ii)}{(ii)}  follow immediately from the conclusions of Theorem~\ref{T:2.7}. 

In our setting, the manifold~$\YYYY$ is the $t$-fold product~$\T^{2t}$ of the two-dimensional torus, $\XXXX=\XXXX^s$, and~$\EEEE$ is the $t$-fold product of~$L^2(J,V^s)$. Hence, to prove the existence of a $C^k$-smooth density and its Lipschitz continuity in~$\varUpsilon\in\XXXX^s$, we need to check that the law~$\ell^t$ satisfies~\hyperlink{(P)}{(P)}, the map~$F^t$ is $(k+1)$-times continuously differentiable, and that the derivative of~$F^t$ is surjective.

It is straightforward to see that if Hypothesis~\hyperlink{(N)}{(N)} holds for~$\ell$, then so does~\hyperlink{(P)}{(P)} with the orthonormal basis $\{\varphi_j\}_{j\ge1}$ obtained from $\{\psi_l e_j,l\ge1,j\in\Z_*^2\}$ by normalisation. It follows that the product measure~$\ell^t$ also satisfies~\hyperlink{(P)}{(P)}. To prove that~$F^t$ is $C^{k+1}$, we recall that the resolving operator for the Navier--Stokes system is an  infinitely smooth map from $V^r\times L^2(J,V^r)$ to~$\XX_r(J)$ for any $r\ge 2$. Moreover, standard results in the theory of ODEs imply that the solution of~\eqref{ode-fluidparticle} is $C^{r-2}$ function of the initial condition and the vector field $u\in \XX_r(J)$, and  the $C^{k+1}$ regularity of $F^1(\varUpsilon,\eta)=S^y(\varUpsilon,\eta)$ follows. The $C^{k+1}$ regularity of~$F^t$ follows now by recurrence. 

Finally, we prove the  surjectivity of the derivative $(D_{\eeta^t}F^t)(\varUpsilon,\eeta^t)$, where $\eeta^t=(\eta_1,\dots,\eta_t)\in\EEEE$. We argue by induction. For $t=1$, this is the claim of Remark~\ref{r-exactcontrol}. Assume that the statement holds 
for $t$. Then 
$$
F^{t+1}(\varUpsilon,\eeta^{t+1})=\bigl(F^{t}(\varUpsilon,\eeta^{t}),S^y(\varUpsilon_t,\eta_{t+1}\bigr)\bigr), \quad \varUpsilon_t=S_t(\varUpsilon,\eeta^t). 
$$
It follows that $(D_{\eeta^{t+1}}F^{t+1})(\varUpsilon,\eeta^{t+1})$ can be written as 
$$
\bigl((D_{\eeta^t}F^t)(\varUpsilon,\eeta^t), (D_\varUpsilon S^y)(\varUpsilon_t,\eta_{t+1})\circ(D_{\eeta^t}S_t)(\varUpsilon,\eeta^t)+(D_\eta S^y)(\varUpsilon_t,\eta_{t+1})\bigr).
$$
Using the induction hypothesis and Remark~\ref{r-exactcontrol}, we conclude that this mapping is surjective.

\smallskip
{\it Step~3: Convergence\/}.  It remains to prove the convergence part  of Theorem~\ref{t-particle}. 
Note that, by Theorem~\ref{t-expomixing},  
\begin{equation} \label{expoconvM}
	\sup_{\varUpsilon\in\XXXX^3}\bigl\|\MMMM_n^\varUpsilon-\MMMM\bigr\|_L^*\le C\,e^{-\gamma n}\quad\mbox{for $n\ge1$},
\end{equation} 
where $C$ and~$\gamma$ are positive numbers, and the dual-Lipschitz norm can be taken over~$\XXXX^s$ for any integer $s\ge3$. It now follows from  relations~\eqref{stationarymeasure-y} and~\eqref{rhon1nt} that 
\begin{align*}
\|\rho_{[\![n+1,n+t]\!]}^\varUpsilon-\rho_t\|_{C^k(\T^{2t})}
&\le \sup_{\alpha,y}\, \biggl|\int_{\XXXX^s}
	(\p^\alpha\rho_{t}^{\upsilon})(y)\bigl(\MMMM_n^{\varUpsilon}-\MMMM\bigr)
	(\dd\upsilon)\biggr|\\
	&\le \bigl\|\rho_t^\cdot\bigr\|_{L_b(\XXXX^s,C^k)} \bigl\|\MMMM_n^\varUpsilon-\MMMM\bigr\|_L^*\\
	&\le Ce^{-\gamma n}\bigl\|\rho_t^\cdot\bigr\|_{L_b(\XXXX^s,C^k)}. 
\end{align*}
This completes the proof of  Theorem~\ref{t-particle}. 

\subsection{Strict positivity of densities}
\label{ss-EP}
This section is devoted to the proof of Theorem~\ref{t-entropy}, which was reduced to deriving inequality~\eqref{positivity-density-particle}. To this end, we shall apply Theorem~\ref{T:2.8}. As it was established in the proof of Theorem~\ref{t-particle}, if we take $s=3$, then the map $S^y(u_0,\eta)$ acting from $\EEEE:=L^2(J,V^3)$ to~$\T^2$ satisfies the hypotheses of Theorem~\ref{T:2.7} with $k=1$, so that the measure $S_*^y(\varUpsilon,\ell)=P_1(\varUpsilon,\cdot)$ has a $C^1$-smooth density~$\rho_1^\varUpsilon(y)$ with respect to the Lebesgue measure on~$\T^2$. Let us denote by~$\{\varphi_{lj}(t,x)\}$ the orthonormal basis in~$\EEEE$ formed of the normalised orthogonal functions $\psi_l(t)e_j(x)$, where $l\ge1$ and $j\in\Z_*^2$. To establish the strict positivity of~$\rho_1^\varUpsilon$, it suffices to prove that, for any $\varUpsilon=(u_0,p_0)\in\XXXX^s$ and $\hat p\in\T^2$, there is $\hat\eta\in\KK$ such that 
\begin{align}
	S^y(\varUpsilon,\hat\eta)&=\hat p,\label{exact-control}\\
	\rho_{lj}\bigl(\langle \hat \eta,\varphi_{lj}\rangle_\EEEE\bigr)&>0
	\quad\mbox{for all $l\ge1$, $j\in\Z_*^2$}.
	\label{rho-lj-positivity}
\end{align}
This is done in three steps. We first describe explicitly a subset~$\KK_0$ in the support~$\KK$ of the measure~$\ell$ such that~\eqref{rho-lj-positivity} holds for any $\hat\eta\in\KK_0$. We next combine a result from the control theory of the Navier--Stokes equation (see Theorem~\ref{t-AS}) with the argument  used in the verification of~\hyperlink{(AC)}{(AC)} to construct a time-regular finite-dimensional control for which~\eqref{exact-control} holds. And, finally, we prove that the resulting control belongs to~$\KK_0$. 

\smallskip
{\it Step~1: Description of~$\KK_0$}. 
Recall that the noise in~\eqref{NS} is now replaced by~$\eta^a$, so that its restriction to the interval~$J_k=[k-1,k]$ has the form (if we replace~$t$ by $t+k-1$)
\begin{equation}
	\eta_k^a=a\sum_{l=1}^\infty\sum_{j\in\Z_*^2}
	b_jd_jc_l\xi_{jl}^k\varphi_{lj}(t,x),
\end{equation}
where $d_j=\|e_j\|_s^{-1}$. Let us denote by~$\ell^a$ the law of~$\eta_k^a$ and by~$\rho_{lj}^a$ the density of its projection to the one-dimensional space spanned by~$\varphi_{lj}$. Since the density of~$\xi_{lj}^k$ is positive on the interval $(-\delta,\delta)$, we have 
\begin{equation} \label{rho-jl}
	\rho_{lj}^a(r)>0 \quad\mbox{for $|r|<\delta_{lj}:=ab_jd_jc_l\delta$}. 
\end{equation}
Thus, defining~$\KK_0$ as the set of functions $\psi\in\EEEE$ such that $|\langle\psi,\varphi_{lj}\rangle_\EEEE|<\delta_{lj}$, we see that any element of~$\KK_0$ satisfies~\eqref{rho-lj-positivity}. 

\smallskip
{\it Step~2: Construction of the control\/}. 
We first reduce the problem to a result on approximate controllability of the Lagrangian component. Namely, let us fix a number $\tau>0$ such  that, for any $\hat p\in\T^2$, the ball $B_{\T^2}(\hat p,\tau)$ is homeomorphic to the unit disc in~$\R^2$, and suppose that  we have proved the following property for any $\varUpsilon=(u_0,p_0)\in\XXXX^s$: 
\begin{itemize}
	\item [\hypertarget{(C)}{(C)}]
	{\sl There is a continuous map $\varPhi^\varUpsilon:B_{\T^2}(\hat p,\tau)\to\EEEE$ whose image is contained in~$\KK_0$ such that} 
	\begin{equation} \label{supSy}
	\sup_{p\in B_{\T^2}(\hat p,\tau)} \bigl|S^y(\varUpsilon,\varPhi^\varUpsilon(p))-p\bigr|\le\tau.
	\end{equation}
\end{itemize}
In this case, for any $\hat p\in\T^2$, the map 
$$
\varPhi_1:B_{\T^2}(\hat p,\tau)\to B_{\T^2}(\hat p,\tau), \quad 
\varPhi_1(p)=p-S^y(\varUpsilon,\varPhi^\varUpsilon(p))+\hat p,
$$
is well defined and continuous. By Brouwer's theorem, there is a fixed point $\bar p\in B_{\T^2}(\hat p,\tau)$ for~$\varPhi_1$, and it is easy to see that the function $\eta=\varPhi^\varUpsilon(\bar p)$ belongs to~$\KK_0$ and satisfies~\eqref{exact-control}.

\smallskip
The proof of Property~\hyperlink{(C)}{(C)}, with an arbitrary $\tau>0$, is based on the  Agrachev--Sarychev theorem (see Theorem~\ref{t-AS} in Section~\ref{ss-AS-control}) and a modification of the construction used in Section~\ref{ss-ldpproof} to establish~\hyperlink{(AC)}{(AC)}. Let us fix a small parameter $\varkappa>0$ that will be chosen below and define the space $\HH_1=\lspan\{e_j:|j|_1\le 2\}$. By the Agrachev--Sarychev theorem, there is a number $\delta_1>0$ and a continuous mapping $\varPsi_1:\XXXX^s\to C^{\infty}([0,1/2],\HH_1)$ such that 
\begin{gather}
	\supp \varPsi_1(\varUpsilon)\subset[\delta_1,1/2-\delta_1],\label{supportofpsi1}\\
	\sup_{\varUpsilon\in\XXXX^s}\|S_{1/2}^u\bigl(\varUpsilon,\varPsi_1(\varUpsilon)\bigr)\|_s\le \varkappa. \label{approximation-1}
\end{gather}
We now set $p_1=S_{1/2}^y(\varUpsilon,\varPsi_1(\varUpsilon))$ and note that $p_1=p_1(\varUpsilon)$ is a continuous $\T^2$-valued function of~$\varUpsilon\in\XXXX^s$. We define a map  $\gamma:[1/2,1]\to\T^2$ by 
\begin{equation} \label{gamma-2}
	\gamma(t)=\gamma(t;\varUpsilon,p)=\bigl(1-\alpha(t)\bigr)p_1(\varUpsilon)+\alpha(t)p,
\end{equation}
where $\alpha\in C^\infty(\R)$, $\alpha(t)=0$ for $t\le 2/3$ and $\alpha(t)=1$ for $t\geq 5/6$, so that $\gamma(1/2)=p_1(\varUpsilon)$ and $\gamma(1)=p$. Recalling that~$U_1$ and~$U_2$ were introduced in~\eqref{U1U2gamma} and defining $\varphi=(\varphi_1,\varphi_2)$ as the time-derivative of~$\gamma$, for $1/2\le t\le 1$ we set (cf.~\eqref{gammadot} and~\eqref{uy-control})
$$
\bar u(t,x)=\varphi_1(t)U_1\bigl(x-\gamma(t)\bigr)
+\varphi_2(t)U_2\bigl(x-\gamma(t)\bigr), \quad \bar y(t)=\gamma(t;\varUpsilon,p).
$$
It is straightforward to check that the pair $(\bar u(t,x),\bar y(t))$ defined on the time interval~$[1/2,1]$ is a solution of~\eqref{NS}, \eqref{ode-fluidparticle} with the right-hand side 
\begin{equation} \label{Psi-2}
\eta(t)=\varPsi_2(\varUpsilon,p)(t):=\Pi g(t), \quad g(t)=\p_t\bar u+\langle\bar u,\nabla\rangle\bar u-\nu\Delta \bar u.	
\end{equation}
Moreover, the construction implies that~$\varPsi_2$ is continuous in~$(\varUpsilon,p)\in\XXXX^s\times\T^2$ with range in $C^\infty([1/2,1])$, and that
\begin{equation}
	\bar u(t)=0\quad\mbox{for $1/2\le t\le 2/3$}, \qquad \bar y(1/2)=p_1(\varUpsilon), \qquad \bar y(1)=p. 
\end{equation}
By continuity, we can find $\varkappa>0$ such that, if $\varUpsilon_1=(u_1,p_1(\varUpsilon))\in\XXXX^s$ satisfies the inequality $\|u_1\|_s\le \varkappa$, then the solution $(u,y)$ of~\eqref{NS}, \eqref{ode-fluidparticle} issued from~$\varUpsilon_1$ at time $t=1/2$  satisfies the inequality $|y(1)-p|\le\tau$. Thus, defining $\varPhi^\varUpsilon:\XXXX^s\to\EEEE$ by the relation
$$
\varPhi^\varUpsilon(p)=
\left\{
\begin{array}{cl}
	\varPsi_1(\varUpsilon)(t) & \quad\mbox{for $0\le t\le 1/2$},\\[4pt]
	\varPsi_2(\varUpsilon,p)(t)& \quad\mbox{for $1/2\le t\le 1$},
\end{array}
\right.
$$
we see that~\eqref{supSy} holds. Moreover, the construction implies that the image of~$\varPhi^\varUpsilon(p)$ is an $\HH_1$-valued function of~$t$ that vanishes in the neighbourhood of $t=1/2$ and whose restrictions to the intervals $[0,1/2]$ and~$[1/2,1]$ are infinitely smooth, so that $\varPhi^\varUpsilon(p)\in C^\infty(J,V^s)$. Finally, the continuity properties of~$\varPsi_1$ and~$\varPsi_2$ imply that $\varPhi^\varUpsilon(p)$ is a continuous function of~$(\varUpsilon,p)\in\XXXX^s\times\T^2$ with range in $C^\infty(J,V^s)$. Thus, it remains to prove that its image is contained in~$\KK_0$.  

\smallskip
{\it Step~3: Description of the image\/}. 
Let us take an arbitrary $(\varUpsilon,p)\in\XXXX^s\times\T^2$ and consider the function $\psi=\varPhi^\varUpsilon(p)$. Writing it in the form
\begin{equation} \label{psi-sum}
\psi(t,x)=\sum_{l=1}^\infty\sum_{j\in\Z_*^2}\psi_{lj}\varphi_{lj}(t,x),	
\end{equation}
where $\psi_{lj}=\psi_{lj}(\varUpsilon,p):=\langle \psi,\varphi_{lj}\rangle_\EEEE$, and recalling the definition of~$\KK_0$, we see that it suffices to prove the inequalities
\begin{equation} \label{psi-lj-positive}
|\psi_{lj}|<\delta_{lj}\quad\mbox{for all $l\ge1$ and $j\in\Z_*^2$},
\end{equation}
where $\delta_{lj}$ is defined in~\eqref{rho-jl}. Note that the function $\psi$ takes values in~$\HH_1$, so that $\psi_{lj}=0$ for $|j|_1>2$. On the other hand, in view of Poincar\'e's inequality~\eqref{poincareineq}, for $|j|_2\le2$ we have 
$$
|\psi_{lj}|\le C_r l^{-\theta r}, \quad l,r\ge1,
$$
where the numbers~$C_r$ do not depend on~$(\varUpsilon,p)\in\XXXX^s\times\T^2$. Thus, the coefficient~$\psi_{lj}$ will satisfy~\eqref{psi-lj-positive} if 
\begin{equation} \label{compatibility}
	C_rl^{-\theta r}<ab_jd_jc_l\delta\quad
	\mbox{for $|j|_2\le 2$, $l\ge1$}.
\end{equation}
Since $r\ge1$ can be chosen arbitrarily large, recalling~\eqref{c-l}, we see that~\eqref{compatibility} is certainly satisfied if $a>0$ is sufficiently large. This completes the proof of Theorem~\ref{t-entropy}. 

\begin{remark} \label{r-parameter}
	The above proof gives that the conclusions of Theorem~\ref{t-entropy} remain valid if only the Fourier components~$e_j$ with $|j|_1\le 2$ are multiplied by a large parameter. 
\end{remark}

\begin{remark}
Since the control function~$\eta$ entering~\eqref{NS} has full range, we could have avoided use of the Agrachev--Sarychev theorem by connecting a given initial function~$u_0$ with zero and defining the corresponding control by relation~\eqref{Psi-2}. However, with this construction, the control~$\varPsi_1(\varUpsilon)$ on the interval $[0,1/2]$ would have much larger dimension, and Remark~\ref{r-parameter} would not be valid. 
\end{remark}

\section{Appendix}
\label{s-appendix}
In this section, we recall various results used in the main text. Section~\ref{ss-expomixingcoupling} deals with the property of exponential mixing and construction of coupling operators. In Section~\ref{ss-kifer}, we formulate Kifer's criterion for the validity of LDP. In Section~\ref{ss-feynmankac}, we discuss the large-time asymptotics of generalised Markov semigroups. Section~\ref{s-proofs-measures} is the devoted to the proof of the theorems on the image of probability measures under smooth maps. Finally, in Section~\ref{ss-AS-control}, we recall the Agrachev--Sarychev theorem.

\subsection{Exponential mixing and coupling operators}
\label{ss-expomixingcoupling}
In this section, we discuss the problem of uniqueness of stationary measure and exponential mixing for the RDS~\eqref{RDS} in the phase space~$\XXXX=\aA\times\YYYY$, where~$\aA$ is a compact subset of a separable Hilbert space and~$\YYYY$ is a compact manifold (without boundary). We denote by~$P_1(\varUpsilon,\Gamma)$ the transition function for~\eqref{RDS}, and by~$\PPPP_k$ and~$\PPPP_k^*$ the corresponding Markov operators acting in the spaces~$C(\XXXX)$ and~$\PP(\XXXX)$, respectively. The following result has been essentially established in~\cite{KNS-2018,shirikyan-jems2019}; however, since our setting here  is slightly different, we give a short proof.

\begin{theorem} \label{t-expomixing}
	Assume that Hypotheses~\hyperlink{(R)}{\rm(R)}, \hyperlink{(AC)}{\rm(AC)}, \hyperlink{(ACL)}{\rm(ACL)}, and~\hyperlink{(D)}{\rm(D)} are satisfied. Then the discrete-time Markov process associated with~\eqref{RDS} has a unique stationary measure $\mu\in\PP(\XXXX)$. Moreover, there are positive numbers~$\gamma$ and~$C$ such that 
	\begin{equation} \label{expomixing}
		\|\PPPP_k^*\nu-\mu\|_L^*\le Ce^{-\gamma k}\quad\mbox{for $\nu\in\PP(\XXXX)$, $k\ge0$.}
	\end{equation}
\end{theorem}

\begin{proof} 
Set $\DDD_\delta=\{(\varUpsilon,\varUpsilon')\in \XXXX\times \XXXX:d_\XXXX(\varUpsilon,\varUpsilon')\le\delta\}$. In view of Theorem~1.1 in~\cite{shirikyan-jems2019}, it suffices to prove that the following local stabilisation property holds:
	\begin{itemize}
		\item [\hypertarget{(LS)}{\bf(LS)}]
{\sl For any $R>0$ and any compact set $\KKKK\subset \EEEE$ there is a finite-dimensional subspace $\EE\subset \EEEE$, positive numbers~$C$, $\delta$, and~$q<1$, and a continuous mapping 
$$
\varPhi:\DDD_\delta\times B_\EEEE(R)\to \EE, \quad (\varUpsilon,\varUpsilon',\eta)\mapsto \eta', 
$$
which is continuously differentiable in~$\eta$ and satisfies the following inequalities for any $(\varUpsilon,\varUpsilon')\in \DDD_\delta:$}
\begin{align}
\sup_{\eta\in B_\EEEE(R)}
\bigl(\|\varPhi(\varUpsilon,\varUpsilon',\eta)\|_\EEEE+\|D_\eta\varPhi(\varUpsilon,\varUpsilon',\eta)\|_{\LL(\EEEE)}\bigr)
&\le C\,d_\XXXX(\varUpsilon,\varUpsilon'),\label{1.5}\\
\sup_{\eta\in \KKKK}d_\XXXX\bigl(S(\varUpsilon,\eta),S(\varUpsilon',\eta+\varPhi(\varUpsilon,\varUpsilon',\eta))\bigr)
&\le q\,d_\XXXX(\varUpsilon,\varUpsilon').\label{localsqueezing}
\end{align}
	\end{itemize}
	
We shall show that Hypotheses~\hyperlink{(R)}{(R)} and~\hyperlink{(ACL)}{(ACL)} imply~\hyperlink{(LS)}{(LS)}. Let us recall that, given a separable Hilbert space~$H$ and a continuous  linear operator $A:\EEEE\to H$ with a dense range, we can construct an approximate inverse for~$A$ in the following way. Setting $G=AA^*$, it is easy to check that 
$$
G(G+\gamma I)^{-1}f\to f\quad\mbox{as $\gamma\to0^+$ for any $f\in H$},
$$
so that $A^*(G+\gamma I)^{-1}$ is an approximate right inverse of~$A$. Moreover, if~${\mathsf P}_n$ are finite-dimensional projections converging to the identity in~$\EEEE$ for the strong operator topology (see~\hyperlink{(D)}{(D)}), then considering the operators ${\mathsf P}_MA^*(G+\gamma I)^{-1}$, we obtain a family of continuous finite-dimensional operators $R_\e:H\to\EEEE$ such that $AR_\e\to I$ in the strong operator topology. 

We now apply  this procedure to the derivatives $A(\varUpsilon,\eta):=(D_\eta S)(\varUpsilon,\eta)$ acting continuously from $\EEEE$ to $\HH\times T_y\YYYY$, where $y=\Pi_\YYYY S(\varUpsilon,\eta)$. Set
$$
G(\varUpsilon,\eta)=A(\varUpsilon,\eta)\,A(\varUpsilon,\eta)^*, \quad R_{M,\gamma}(\varUpsilon,\eta)
={\mathsf P}_MA(\varUpsilon,\eta)^*(G(\varUpsilon,\eta)+\gamma I)^{-1}. 
$$ 
Repeating the compactness argument used in the proof of Proposition~2.3 in~\cite{KNS-2018}, it is not difficult to show that, for any $\e>0$, there are $M_\e\ge1$ and $\gamma_\e>0$ such that the operator $R_\e(\varUpsilon,\eta)=R_{M_\e,\gamma_\e}(\varUpsilon,\eta)$ (which  acts continuously from $\HHHH_y:=\HH\times T_y\YYYY$ to~$\EEEE$ and has a finite-dimensional range)  satisfies the following inequalities for any $\varUpsilon\in\XXXX$, $\eta\in\KK$, and $f\in \VV\times T_y\YYYY$:
\begin{align}
\bigl\|(D_\eta S)(\varUpsilon,\eta)	R_\e (\varUpsilon,\eta)f-f\bigr\|_{\HHHH_y} & \le \e\, \bigr\|f\bigr\|_{\VV\times T_y\YYYY}, \label{1.6}\\ 
\bigl\|R_\e(\varUpsilon,\eta)\bigr\|_{\LL(\HHHH_y,\EEEE)}+\bigl\|(D_\eta R_\e)(\varUpsilon,\eta)\bigr\|_{\LL(\HHHH_y\times \EEEE,\EEEE)}&\le C_1(\e),  \label{1.7}
\end{align}
where $C_1(\e)$ does not depend on~$(\varUpsilon,\eta)$, and the tangent space~$T_y\YYYY$ is endowed with the norm induced by  the Riemannian metric of~$\YYYY$. We now fix points $\varUpsilon=(u,y)\in\XXXX$ and~$\eta\in B_\EEEE(R)$, together with some local charts around~$y$ and $\Pi_\YYYY S(\varUpsilon,\eta)$, and use Taylor's formula to write
\begin{equation} \label{1.8}
S(\varUpsilon',\eta')-S(\varUpsilon,\eta)=(D_\varUpsilon S)(\varUpsilon,\eta)(\varUpsilon'-\varUpsilon)+(D_\eta S)(\varUpsilon,\eta)(\eta'-\eta)
+r(\varUpsilon,\varUpsilon',\eta,\eta').
\end{equation}
Here $\varUpsilon'$ and~$\eta'$ are sufficiently close to~$\varUpsilon$ and~$\eta$, respectively, so that~$\Pi_\YYYY\varUpsilon$ and~$\Pi_\YYYY S(\varUpsilon',\eta')$ belong to the above-mentioned local charts\footnote{This enables one to write the differences $\varUpsilon'-\varUpsilon$ and $S(\varUpsilon',\eta')-S(\varUpsilon,\eta)$ and to consider them as elements of $\HH\times\R^d$. In particular, we shall write $\|\varUpsilon'-\varUpsilon\|_\HHHH$ for the distance between~$\varUpsilon$ and~$\varUpsilon'$.}, and the remainder term~$r$ satisfies the inequality 
\begin{equation} \label{1.9}
\bigl\|r(\varUpsilon,\varUpsilon',\eta,\eta')\bigr\|_{\HHHH}\le C_2(R)\,\bigl(\bigl\|\varUpsilon-\varUpsilon'\bigr\|_{\HHHH}^2+\bigl\|\eta-\eta'\bigr\|_\EEEE^2\bigr),
\end{equation}
where $\varUpsilon\in\XXXX$ and $\eta\in B_\EEEE(R)$. Define  
\begin{equation} \label{1.10}
\varPhi(\varUpsilon,\varUpsilon',\eta)=-R_\e(\varUpsilon,\eta)(D_\varUpsilon S)(\varUpsilon,\eta)(\varUpsilon'-\varUpsilon),
\end{equation}
and note that
\begin{equation} \label{1.11}
\bigl\|(D_\varUpsilon S)(\varUpsilon,\eta)(\varUpsilon'-\varUpsilon)\bigr\|_{\VV\times T_y\YYYY} 
\le C_3(R)\bigl\|\varUpsilon'-\varUpsilon\bigr\|_{\HH\times T_y\YYYY}.
\end{equation}
Combining this inequality with~\eqref{1.7} and Hypothesis~\hyperlink{(R)}{(R)}, we see that~\eqref{1.5} holds. Furthermore, it follows from \eqref{1.6}--\eqref{1.11} that
\begin{align*}  
&\bigl\|S(\varUpsilon,\eta)-S(\varUpsilon',\eta+\varPhi(\varUpsilon,\varUpsilon',\eta))\bigr\|_\HHHH \\
&\qquad \le\e\,\bigl\|(D_\varUpsilon S)(\varUpsilon,\eta)(\varUpsilon'-\varUpsilon)\bigr\|_{\VV\times T_y\YYYY}+\bigl\|r(\varUpsilon,\varUpsilon',\eta,\eta+\varPhi(\varUpsilon,\varUpsilon',\eta))\bigl\|_\HHHH\nonumber\\
&\qquad \le \bigl(C_3(R)\e+ C_4(R)\,\|\varUpsilon'-\varUpsilon\|_\HHHH \bigl)\bigl\|\varUpsilon'-\varUpsilon\bigr\|_\HHHH 
\le q\,\bigl\|\varUpsilon'-\varUpsilon\bigr\|_\HHHH,
\end{align*} 
where $\varUpsilon,\varUpsilon'\in \DDD_\delta$ and $\eta\in B_\EEEE(R)$, the number $q>0$ is arbitrary, and the positive numbers~$\e$  and~$\delta$ are sufficiently  small. This gives inequality~\eqref{localsqueezing} and proves the local stabilisability. 
\end{proof}

\begin{remark} \label{r-expomixing} The above proof gives that Hypothesis~\hyperlink{(AC)}{(AC)} in Theorem~\ref{t-expomixing} can be replaced by the following weaker variant which requires approximate controllability to some fixed point $\widehat\varUpsilon\in\XXXX$: 
	\begin{itemize}
		\item[\hypertarget{(ACP)}{\bf(ACP)}]
		For any $\e>0$, there is an integer $n\ge1$ such that, for any initial point $\varUpsilon\in\XXXX$, one can find controls $\zeta_1,\dots,\zeta_n\in\KK$ satisfying inequality~\eqref{AC}. 
	\end{itemize}
\end{remark}

In the proof of the LDP, we also needed the existence of coupled trajectories that converge to each other exponentially fast. Their existence is established with the help of the following coupling construction. 

\begin{proposition} \label{p-coupling}
	Under the hypotheses of Theorem~\ref{t-expomixing}, for any $q\in (0,1)$, there is a number $C>0$, a probability space $(\Omega,\FF,\IP)$, and measurable mappings 
	$$
	\RR,\RR':\XXXX\times \XXXX\times\Omega\to \XXXX,
	$$  
	such that, for any $\varUpsilon,\varUpsilon'\in \XXXX$, the pair $(\RR(\varUpsilon,\varUpsilon',\cdot), \RR'(\varUpsilon,\varUpsilon',\cdot))$ is a coupling for $(P_1(\varUpsilon,\cdot),P_1(\varUpsilon',\cdot))$, and 
 \begin{equation}\label{1.18}	
\IP\bigl\{\|\RR(\varUpsilon,\varUpsilon',\cdot)- \RR'(\varUpsilon,\varUpsilon',\cdot)\|_\HHHH> q\,\|\varUpsilon-\varUpsilon'\|_\HHHH\bigr\}
\le C\,\|\varUpsilon-\varUpsilon'\|_\HHHH. 
\end{equation} 
\end{proposition}

This result has been proved
in~\cite{shirikyan-jems2019} (see Step~3 of the proof of Theorem~1.1), provided that Hypotheses~\hyperlink{(R)}{(R)}, \hyperlink{(LS)}{(LS)}, and~\hyperlink{(D)}{(D)} are satisfied. As was established in the proof of Theorem~\ref{t-expomixing}, under Hypotheses~\hyperlink{(R)}{(R)} and~\hyperlink{(ACL)}{(ACL)} property~\hyperlink{(LS)}{(LS)} holds with any $q<1$, so that Proposition~\ref{p-coupling} also holds. 

We now describe the construction of coupled trajectories. Let $(\widetilde\Omega,\widetilde\FF,\widetilde\IP)$ be the  product of countably many copies of the probability space constructed in Proposition~\ref{p-coupling}. For $\varUpsilon,\varUpsilon'\in\XXXX$ and $\omega=(\omega^1,\omega^2,\dots)\in\widetilde\Omega$, we set 
\begin{align*}
	\tUpsilon_0&=\varUpsilon, & \quad \tUpsilon_0'&=\varUpsilon',\\
	\tUpsilon_k&=\RR(\tUpsilon_{k-1},\tUpsilon_{k-1}',\omega^k),& \quad 
	\tUpsilon_k'&=\RR'(\tUpsilon_{k-1},\tUpsilon_{k-1}',\omega^k),
\end{align*}
where $k\ge1$. For any $q\in(0,1)$ and $k\ge1$, define the event
$$
G_k(q,\varUpsilon,\varUpsilon'):=\bigl\{d_\XXXX(\varUpsilon_k,\varUpsilon_k')> q\,d_\XXXX(\varUpsilon_{k-1},\varUpsilon_{k-1}')\bigr\}. 
$$
The following result is a straightforward consequence of Proposition~\ref{p-coupling}, and its proof, based on an application of the Markov property, can be carried out by repeating the argument in~\cite[Section~3.2.2]{KS-book} (see also~\cite[Section~4.4]{shirikyan-asens2015}).

\begin{corollary} \label{c-coupledtrajectories}
	Under the hypotheses of Theorem~\ref{t-expomixing}, the trajectories $\{\tUpsilon_k\}$ and~$\{\tUpsilon_k'\}$ constructed above have the following properties. 
	\begin{description}
		\item [\sc Coupling.] The laws of the processes $\{\tUpsilon_k\}_{k\ge0}$ and~$\{\tUpsilon_k'\}_{k\ge0}$ regarded as random variables with range in~$\XXX=\XXXX^{\Z_+}$ coincide with those of the trajectories for~\eqref{RDS} issued from the initial points~$\varUpsilon$ and~$\varUpsilon'$, respectively. 
		\item [\sc Estimate.]
		There is a number~$C>0$ depending on~$q\in(0,1)$ such that, for any integer $k\ge1$ and any points $\varUpsilon,\varUpsilon'\in\XXXX$, we have 
		\begin{equation} \label{squeezing}
			\widetilde\IP\bigl(G_k(q,\varUpsilon,\varUpsilon')\bigr)
			\le C\,\E\,d_\XXXX(\varUpsilon_{k-1},\varUpsilon_{k-1}').
		\end{equation}
	\end{description}
\end{corollary}

\subsection{Kifer's criterion}
\label{ss-kifer}
Let~$X$ be a compact metric space, let $\Theta$ be a directed set, let $\{r_\theta\}_{\theta\in\Theta}$ be a non-decreasing net of positive numbers converging  to~$+\infty$, and let $\{\mu_\theta\}_{\theta\in\Theta}$ be a net of random probability measures on~$X$ with an underlying space $(\Omega,\FF,\IP)$. We assume that, for any $V\in C(X)$, the following limit exists: 
\begin{equation} \label{pressure-kifer}
	Q(V)=\lim_{\theta\in\Theta} r_\theta^{-1}\log\E\exp\bigl(r_\theta\langle V,\mu_\theta\rangle\bigr). 
\end{equation}
This is a 1-Lipschitz function on~$C(X)$ such that $Q(V+C)=Q(V)+C$ for any $V\in C(X)$ and $C\in\R$. Let $I:\MM(X)\to[0,+\infty]$ be the Legendre transform of~$Q$: 
\begin{equation} \label{legendre-abstract}
	I(\sigma)=\left\{
	\begin{array}{cl}
		\displaystyle\sup_{V\in C(X)}\bigl(\langle V,\mu\rangle-Q(V)\bigr) & \mbox{for $\sigma\in\PP(X)$},\\
		+\infty &\mbox{otherwise}. 
	\end{array}
	\right.
\end{equation}
It is easy to see that~$I$ is a good rate function (see Section~\ref{ss-result} for a definition). Recall that $\sigma\in \PP(X)$ is called an {\it equilibrium state\/} for $V\in C(X)$ if 
$$
Q(V)=\langle V,\sigma\rangle-I(\sigma). 
$$
The following theorem is due to Kifer~\cite{kifer-1990}. 

\begin{theorem} \label{t-kifer}
	In addition to the existence of limit~\eqref{pressure-kifer}, suppose  that there exists a dense vector space $\VV\subset C(X)$ such that for any $V\in \VV$ there is a unique equilibrium state $\sigma\in\PP(X)$. Then the LDP holds for~$\{\mu_\theta\}$ with the speed~$\{r_\theta\}$ and the rate function~$I$. 
\end{theorem}

\subsection{Asymptotics of Feynman--Kac semigroups}
\label{ss-feynmankac}
Let $X$ be a compact metric space, let $\MM_+(X)$ be the set of non-negative Borel measures on~$X$, and let $\{P(x,\cdot),x\in X\}\subset \MM_+(X)$ be a family such that $P(x,X)>0$ for any $x\in X$, and the mapping $x\mapsto P(x,\cdot)$ is continuous from~$X$ to the space~$\MM_+(X)$ endowed with the weak$^*$ topology. We denote by $P_k(x,\cdot)$ the $k$-fold iteration of $P(x,\cdot)$, and by
$$
\PPPP_k:C(X)\to C(X), \quad \PPPP_k^*:\MM(X)\to\MM(X) 
$$
the semigroups with the generators~$\PPPP$ and~$\PPPP^*$ defined by 
$$
(\PPPP f)(x)=\int_X P(x,\dd y)f(y), \quad (\PPPP^*\sigma)(\Gamma)=\int_X P(x,\Gamma)\sigma(\dd x),
$$
where $f\in C(X)$, $\sigma\in \MM(X)$, and $\Gamma\in\BB(X)$. Note that $P_k(x,\cdot)$ is the kernel of the operator~$\PPPP_k$. Recall that a family $\CC\subset C(X)$ is said to be {\it  determining\/} if, for any two measures $\mu,\nu\in\MM_+(X)$, the validity of the relation $\langle f,\mu\rangle=\langle f,\nu\rangle$ for all $f\in\CC$ implies that $\mu=\nu$. In addition to the above hypotheses, let us assume that the following two properties hold. 
	
\begin{description}
	\item[\hypertarget{(UI)}{(UI)} \sc Uniform irreducibility.] For any $\e>0$ there is an integer $n\ge1$ and a number $p>0$ such that 
	\begin{equation} \label{ntransition}
		P_n\bigl(x,B_X(\hat x,\e)\bigr)\ge p\quad\mbox{for all $x,\hat x\in X$}.
	\end{equation} 
	\item[\hypertarget{(UF)}{(UF)} \sc Uniform Feller property.] \sl There is a determining family $\CC\subset C(X)$ such that, for any $f\in \CC$, the sequence $\{\|\PPPP_k{\bf1}\|_\infty^{-1}\PPPP_kf,k\ge0\}$ is uniformly equicontinuous. 
\end{description}
The following theorem is established in~\cite[Section~2]{JNPS-cpam2015}. 

\begin{theorem} \label{t-feynmankac}
Under the above hypotheses, there is a number $\lambda>0$, a measure $\mu\in\PP(X)$, and a positive function $h\in C(X)$ such that $\langle h,\mu\rangle=1$ and 
\begin{align}
	\PPPP_1h&=\lambda h, \quad \PPPP_1^*\mu=\lambda\mu,\label{eigenvalues}\\
	\lambda^{-k}\PPPP_k f&\to \langle f,\mu\rangle h\quad 
	\mbox{in $C(X)$ as $k\to\infty$},
	\label{convergence-functions}\\
	\lambda^{-k}\PPPP_k^* \sigma&\to \langle h,\sigma\rangle \mu\quad \mbox{in $\MM_+(X)$ as $k\to\infty$},
	\label{convergence-measures}
\end{align}
where $f\in C(X)$ and $\sigma\in\MM_+(X)$ are arbitrary. 
\end{theorem}

\subsection{Proofs of Theorems~\ref{T:2.7} and~\ref{T:2.8}}
\label{s-proofs-measures}

\begin{proof}[Proof of Theorem~\ref{T:2.7}]
We follow the arguments in~\cite{AKSS-aihp2007} and~\cite[Section~9.6]{bogachev2010}. The proof in our case is simpler since the derivative of~$F$ has full rank everywhere. 

{\it Step~1: Localisation in~$\varUpsilon$\/}. We first prove that it suffices to establish the result in the neighbourhood of each point $\varUpsilon\in\XXXX$. Namely, suppose that for any $\varUpsilon\in\XXXX$ there is $\delta=\delta_\varUpsilon>0$ such that, for any $\varUpsilon'\in B_\HHHH(\varUpsilon,\delta)$, the measure $F_*(\varUpsilon',\ell)$ has a $C^k$-smooth density $\rho(\varUpsilon',y)$ that satisfies the inequality in~\eqref{2.8b} for $\varUpsilon_1,\varUpsilon_2\in B_\HHHH(\varUpsilon,\delta)$. In this case, the existence and regularity of the density for $F_*(\varUpsilon,\ell)$ with $\varUpsilon\in\XXXX$ is trivial. To prove inequality~\eqref{2.8b}, we consider the open cover $\{\dot B_\HHHH(\varUpsilon,\delta_\varUpsilon/2)\}_{\varUpsilon\in\XXXX}$ of the compact set~$\XXXX$ and select a finite sub-cover $\{\OO^m, 1\le m\le M\}$, where $\OO^m=\dot B_\HHHH(\varUpsilon^m,\delta_m/2)$ for some $\varUpsilon^m\in\XXXX$ and $\delta_m>0$. Denoting by~$\delta$ the minimum of the numbers~$\delta_m$, $1\le m\le M$, we note that it suffices to establish~\eqref{2.8b} for points $\varUpsilon_1,\varUpsilon_2$ satisfying $d(\varUpsilon_1,\varUpsilon_2)<\delta$. For any such pair, we can find $m\in[\![1,M]\!]$ such that $\varUpsilon_1,\varUpsilon_2\in B_\HHHH(\varUpsilon^m,\delta_m)$, so that~\eqref{2.8b} holds by assumption. 

\smallskip
{\it Step~2: Localisation in~$\eta$\/}. 
Let us fix any $\varUpsilon\in\XXXX$ and assume that for any point $\eta\in\KK$ we can find positive numbers~$\delta_\eta$ and~$\gamma_\eta$ such that, if a measure~$\ell\in\MM(\EEEE)$ satisfies~\hyperlink{(P)}{(P)}, and $\psi:\EEEE\to\R$ is a $C^\infty$-function with a support contained in~$B_\EEEE(\eta,\gamma_\eta)$, then $F_*(\varUpsilon,\psi\ell)$ has a $C^k$-smooth density $\rho_\psi(\varUpsilon,\cdot)$ for $\varUpsilon \in B_\HHHH(\varUpsilon,\delta_\eta)$, and inequality~\eqref{2.8b} holds for $\varUpsilon_1,\varUpsilon_2 \in B_\HHHH(\varUpsilon,\delta_\eta)$. In this case, we consider the open cover $\{\dot B_\EEEE(\eta,\gamma_\eta)\}_{\eta\in\KK}$ of the compact set~$\KK$ and choose a finite sub-cover $\{\UU^m, 1\le m\le M\}$, where $\UU^m=\dot B_\EEEE(\eta^m, \gamma_m)$ for some $\eta^m\in\KK$ and $\gamma_m>0$. Let $\{\psi^m\}$ be an infinitely smooth partition of unity on~$\KK$ subordinate to~$\{\UU^m\}$; see~\cite[Section~II.3]{lang1985}. Then we can write \begin{equation} \label{localisation-sum}
	\lambda_{\varUpsilon'}=F_*(\varUpsilon',\ell)=\sum_{m=1}^M F_*(\varUpsilon',\psi^m\ell). 
\end{equation}
Setting $\delta=\min\{\delta_m, 1\le m\le M\}$, where $\delta_m>0$ is the number corresponding to~$\eta^m$, we see that, for any $\varUpsilon'\in B_\HHHH(\varUpsilon,\delta)$, each term of the sum in~\eqref{localisation-sum} possesses a $C^k$-smooth density $\rho_m(\varUpsilon',y)$ that satisfies~\eqref{2.8b} for $\varUpsilon_1,\varUpsilon_2 \in B_\HHHH(\varUpsilon,\delta)$. This proves the required property.
	
\smallskip
{\it Step~3: Proof in the localised case\/}. It remains to establish the property described in the beginning of Step~2. Let us fix any $\widehat\varUpsilon\in\XXXX$ and $\hat\eta\in\KK$ and choose a local chart in the neighbourhood of $\hat y=F(\widehat\varUpsilon,\hat\eta)$, so that $F$ can be written as~$(F_1,\dots,F_d)$. By~\eqref{2.7}, we can find vectors $g_1,\dots,g_d\in\{\varphi_j\}_{j\ge1}$ such that 
\begin{equation} \label{non-degenerate}
	D(\widehat\varUpsilon,\hat\eta):=\bigl|\det\bigl(D_{g_i}F_j(\widehat\varUpsilon,\hat\eta)\bigr)\bigr|>0. 
\end{equation}
Let us denote by~$\EEEE_1$ the vector span of~$g_1,\dots,g_d$ and by~$\EEEE_2$ its orthogonal complement, so that we can write $\eta=\eta^1+\eta^2$ with $\eta^i\in\EEEE_i$. By the implicit function theorem, there are $\delta,\gamma_1,\gamma_2>0$ such that, for any vectors $\varUpsilon\in B_\HHHH(\widehat\varUpsilon,\delta)$ and $\eta^2\in B_{\EEEE_2}(\hat\eta^2,\gamma_2)$, the mapping $\eta^1\mapsto F(\varUpsilon,\eta^1+\eta^2)$ is a $C^{k+1}$-smooth diffeomorphism of the open ball $\dot B_{\EEEE_1}(\hat\eta^1,2\gamma_1)$ onto its image $W(\varUpsilon,\eta^2,2\gamma_1)$, and the determinant $D(\varUpsilon,\eta^1+\eta^2)$ is separated from zero by a number $c>0$. Let us denote by $G(\varUpsilon,\eta^2;\cdot):W(\varUpsilon,\eta^2,2\gamma_1)\to \dot B_{\EEEE^1}(\eta^1,2\gamma_1)$ the inverse function. Decreasing, if necessary, the numbers~$\delta$ and~$\gamma_2$, we can assume that the closure of $W(\varUpsilon,\eta^2,\gamma_1)$ is included in $B:=W(\widehat\varUpsilon,\hat\eta^2,2\gamma_1)$, so that the map $G(\varUpsilon,\eta^2;y)$ is well defined on the product set $B_\HHHH(\widehat\varUpsilon,\delta)\times B_{\EEEE_2}(\hat\eta^2,\gamma_2)\times B$ whose last component is independent of~$\varUpsilon$ and~$\eta^2$.  

Now let $\gamma=\min(\gamma_1,\gamma_2)$ and let $\psi:\EEEE\to\R$ be a smooth function with support in~$B_\EEEE(\hat\eta,\gamma)$. Then the support of~$F_*(\varUpsilon,\psi\ell)$ is contained in~$B$. Set $B^1=B_\EEEE(\hat\eta^1,\gamma)$ and $B^2=B_\EEEE(\hat\eta^2,\gamma)$, and denote by $\ell^i$ the projection of~$\ell$ to~$\EEEE^i$. Hypothesis~\hyperlink{(P)}{(P)} implies that~$\ell^1$ has a smooth density~$\rho^1(\eta^1)$ with respect to the Lebesgue measure. Hence, for any function $f\in C(\YYYY)$ supported in~$B$, Fubini's theorem gives 
\begin{align*}
\langle f,F_*(\varUpsilon,\psi\ell)\rangle
&=\int_{\EEEE} f\bigl(F(\varUpsilon,\eta)\bigr)\psi(\eta)\ell(\dd\eta)\notag\\
&=\int_{B^2}\biggl\{\int_{B^1}  f\bigl(F(\varUpsilon,\eta^1+\eta^2)\bigr)\psi(\eta^1+\eta^2)\rho^1(\eta^1)\dd\eta^1\biggr\}\,\ell^2(\dd\eta^2)\notag\\
&=\int_{B^2}\biggl\{\int_{B}  f(y)\frac{\psi(\eta^1+\eta^2)\rho^1(\eta^1)}{D(\varUpsilon,\eta^1+\eta^2)}\,\dd y\biggr\}\,\ell^2(\dd\eta^2),
\end{align*}
where we performed the change of variable $y=F(\varUpsilon,\eta^1+\eta^2)$ to derive the last line, in which $\eta^1=G(\varUpsilon,\eta^2;y)$. This relation implies that, for $\varUpsilon\in B_\HHHH(\widehat\varUpsilon,\delta)$, the measure $F_*(\varUpsilon,\psi\ell)$ (which is supported in~$B$) has a density given by 
\begin{equation} \label{local-density}
	\rho_\psi(\varUpsilon,y)
	=\int_{B^2}\frac{\psi(G(\varUpsilon,\eta^2;y)+\eta^2)\rho^1(G(\varUpsilon,\eta^2;y))}{D(\varUpsilon,G(\varUpsilon,\eta^2;y)+\eta^2)}\,\ell^2(\dd\eta^2),
\end{equation}
where the denominator satisfies the following inequality in the support of the numerator: 
\begin{equation} \label{lower-bound-determinant}
	D(\varUpsilon,G(\varUpsilon,\eta^2;y)+\eta^2)\ge c. 
\end{equation}
Relations~\eqref{local-density} and~\eqref{lower-bound-determinant} show that~$\rho_\psi$ is $C^k$-smooth and satisfies~\eqref{2.8b} for $\varUpsilon_1,\varUpsilon_2\in B_\HHHH(\widehat\varUpsilon,\delta)$. This completes the proof of the theorem. 
\end{proof}

\begin{proof}[Proof of Theorem~\ref{T:2.8}]
	Let $\psi:\EEEE\to\R$ be any $C^\infty$-function with support in a small ball $B_\EEEE(\hat\eta,\gamma)$ such that $0\le \psi\le1$ and $\psi(\hat\eta)>0$. In this case, $F_*(\varUpsilon,\ell)$ is minorised by the measure~$F_*(\varUpsilon,\psi\ell)$. As was established in the proof of Theorem~\ref{T:2.7}, the latter has a density~$\rho_\psi(\varUpsilon,y)$ given by~\eqref{local-density}. Taking $\varUpsilon=\widehat\varUpsilon$ and noting that $G(\widehat\varUpsilon,\hat\eta^2,\hat y)=\hat\eta^1$, we can write 
\begin{equation} \label{density-hat}
\rho(\widehat\varUpsilon, \hat y)
\ge \rho_\psi(\widehat\varUpsilon, \hat y)
=\int_{B^2}\biggl(\frac{\psi(\hat\eta)\rho^1(\hat\eta^1)}{D(\widehat\varUpsilon,\hat\eta)}+g(\eta^2)\biggr)\ell^2(\dd\eta^2),	
\end{equation}
where $g(\eta^2)$ is a continuous function vanishing at~$\hat\eta^2$. Now note that, in view of~\eqref{2.14}, the  first term in the brackets under the integral is positive, so that the integrand (which is a non-negative function) is strictly positive in the neighbourhood of~$\hat\eta^2$. Moreover, it follows from~\eqref{2.14} that~$\hat\eta^2$ is in the support of~$\ell^2$. We thus conclude that the integral  in~\eqref{density-hat} is positive. 
\end{proof}

\subsection{Agrachev--Sarychev theorem}
\label{ss-AS-control}
Let us consider the Navier--Stokes system on~$\T^2$ controlled by a finite-dimensional external force. After projecting to the space~$H$ of square-integrable divergence-free vector fields with zero mean value, we write it in the form
\begin{equation} \label{Navier-Stokes}
	\p_t u+\nu Lu+B(u)=\eta(t,x),
\end{equation}
where $L=-\Pi\Delta$, $B(u)=\Pi(\langle u,\nabla\rangle u)$, and~$\Pi:L^2(\T^2,\R^2)\to H$ stands for Leray's projection. We consider the problem on some interval~$J_T=[0,T]$ and denote by~$S_T^u(u_0,\eta)$ the map that takes  functions $u_0\in H$ and $\eta\in L^2(J_T,H)$ to~$u(T)$, where $u(t,x)$ is the solution of~\eqref{Navier-Stokes} issued from~$u_0$. The control force~$\eta$ is assumed to have the form
\begin{equation} \label{control-eta}
	\eta(t,x)=\sum_{j\in\Lambda}\eta_j(t)e_j(x),
\end{equation}
where $\Lambda=\{(1,0), (1,1), (-1,0), (-1,-1)\}$, $\{e_j\}$ is the trigonometric basis in~$H$ (see~\eqref{trigonometric-basis}), and~$\eta_j$ are smooth real-valued functions on~$J_T$. We denote by~$\HH_1$ the four-dimensional vector space spanned by~$\{e_j,j\in\Lambda\}$ and endow the space~$C^\infty(J_T,\HH_1)$ with the usual Fr\'echet topology. The following result is  established in~\cite[Sections~4,6]{AS-2006} (see also~\cite[Theorem~2.5]{shirikyan-aihp2007} for a more explicit statement in the 3D case).\footnote{The papers~\cite{AS-2006,shirikyan-aihp2007} deal with the case when the initial point is fixed and the target varies in a compact subset. However, exactly the same arguments enable one to handle the situation in which both the initial and target states vary in compact subsets.} 

\begin{theorem} \label{t-AS}
	Let $s\ge0$ be an integer and let $\CC_0,\CC\subset V^s$ be compact subsets. Then, for any $\e>0$, there is a number $\delta\in(0,T/2)$ and a  continuous function $\varPsi_T:\CC_0\times\CC\to C^\infty(J_T,\HH_1)$ such that the following properties hold.
	\begin{description}
		\item[\sc Support.] 
		For any $u_0\in\CC_0$ and $\hat u\in\CC$, the support of $\varPsi_T(u_0,\hat u)$ is contained in the interval~$[\delta,T-\delta]$. 
		\item[\sc Approximation.]
		We have the inequality
		\begin{equation} \label{approximation-compacts}
			\sup_{u_0\in\CC_0, \hat u\in\CC}\,\bigl\|S_T^u\bigl(u_0,\varPsi_T(u_0,\hat u)\bigr)-\hat u\bigr\|_s\le\e. 
		\end{equation}
	\end{description}
\end{theorem}

\addcontentsline{toc}{section}{Bibliography}
\newcommand{\etalchar}[1]{$^{#1}$}
\def\cprime{$'$} \def\cprime{$'$}
  \def\polhk#1{\setbox0=\hbox{#1}{\ooalign{\hidewidth
  \lower1.5ex\hbox{`}\hidewidth\crcr\unhbox0}}}
  \def\polhk#1{\setbox0=\hbox{#1}{\ooalign{\hidewidth
  \lower1.5ex\hbox{`}\hidewidth\crcr\unhbox0}}}
  \def\polhk#1{\setbox0=\hbox{#1}{\ooalign{\hidewidth
  \lower1.5ex\hbox{`}\hidewidth\crcr\unhbox0}}} \def\cprime{$'$}
  \def\polhk#1{\setbox0=\hbox{#1}{\ooalign{\hidewidth
  \lower1.5ex\hbox{`}\hidewidth\crcr\unhbox0}}} \def\cprime{$'$}
  \def\cprime{$'$} \def\cprime{$'$} \def\cprime{$'$}
\providecommand{\bysame}{\leavevmode\hbox to3em{\hrulefill}\thinspace}
\providecommand{\MR}{\relax\ifhmode\unskip\space\fi MR }
\providecommand{\MRhref}[2]{%
  \href{http://www.ams.org/mathscinet-getitem?mr=#1}{#2}
}
\providecommand{\href}[2]{#2}

\end{document}